\RequirePackage{fix-cm}
\documentclass[twocolumn,envcountsame,final]{svjour3}       
\smartqed  
\usepackage[utf8]{inputenc}
\usepackage[hidelinks]{hyperref}

\usepackage{amsmath,amsthm,mathrsfs,amsfonts,color,bm}
\usepackage{url}
\usepackage{caption}
\usepackage{subcaption}
\usepackage{graphicx}
\usepackage[title]{appendix}
\usepackage[authoryear]{natbib}
\usepackage{booktabs}
\usepackage{multirow}
\usepackage[linesnumbered,lined,ruled,vlined]{algorithm2e}
\usepackage[fixamsmath]{mathtools}

\SetKwInput{KwParas}{Parameters}
\SetKwInput{KwInit}{Initialization}

\newcommand{\T}[0]{\mathbb{T}}
\let\norm\relax
\DeclarePairedDelimiter{\normdummy}{\lVert}{\rVert}
\newcommand{\norm}{\mathop{}\normdummy}

\let\abs\relax
\DeclarePairedDelimiter{\absdummy}{\lvert}{\rvert}
\newcommand{\abs}{\mathop{}\absdummy}

\newcommand{\trace}{\operatorname{tr}}

\newcommand{\cu}[1]{
	\ifcat\noexpand#1\relax
	\bm{#1}
	\else
	\mathbf{#1}
	\fi
}

\newcommand{\tash}[2]{\frac{\partial #1}{\partial #2}}

\newcommand{\diff}{\mathop{}\!\mathrm{d}}

\DeclareMathOperator*{\argmin}{arg\,min}  

\newcommand{\expec}{\mathbb{E}}           
\newcommand{\cov}{\operatorname{Cov}}     

\newcommand{\R}{\mathbb{R}} 

\newcommand{\diag}{\operatorname{diag}}     

\newcommand*{\trans}{{\mkern-1.5mu\mathsf{T}}}

\newcommand{\matern}{Mat\'{e}rn }

\newtheorem{assumption}{Assumption}

\usepackage[normalem]{ulem}

\begin{document}
\sloppy

\title{Hierarchical Non-Stationary Temporal Gaussian Processes With $L^1$-Regularization\thanks{We thank Academy of Finland and Aalto ELEC doctoral school for financial support.}
}

\author{Zheng Zhao \and
	Rui Gao \and
	Simo S\"{a}rkk\"{a}
}

\institute{Zheng Zhao, Rui Gao, and Simo S\"{a}rkk\"{a} \at
	Department of Electrical Engineering and Automation \\
	Aalto University \\
	\email{zheng.zhao@aalto.fi}
}

\date{Received: date / Accepted: date}

\maketitle

\begin{abstract}%
	This paper is concerned with regularized extensions of hierarchical non-stationary temporal Gaussian processes (NSGPs) in which the parameters (e.g., length-scale) are modeled as GPs. In particular, we consider two commonly used NSGP constructions which are based on explicitly constructed non-stationary covariance functions and stochastic differential equations, respectively. We extend these NSGPs by including $L^1$-regularization on the processes in order to induce sparseness. To solve the resulting regularized NSGP (R-NSGP) regression problem we develop a method based on the alternating direction method of multipliers (ADMM) and we also analyze its convergence properties theoretically. We also evaluate the performance of the proposed methods in simulated and real-world datasets.
\end{abstract}

\section{Introduction}
\label{sec:intro}
Gaussian processes (GPs) are commonly-used priors for functions in machine learning and statistics~\citep{gp-carl-edward}, and they are determined by the mean and covariance functions. In GP regression one often uses, for example, the Whittle--\matern or the radial basis function covariance functions which lead to stationary GPs with a few free parameters (e.g., length-scale). These parameters are usually set by human experts or learned from data globally by using maximum likelihood estimation (MLE) or Bayesian methods. However, stationary GPs do not generalize well to very irregular signals such as signals with jumps or discontinuities because of the strong stationarity assumption of the GP prior~\citep{paciorek2006spatial}. Compared to stationary GPs, non-stationary GPs (NSGPs) are able to produce more general priors by extending the covariance functions to be non-stationary. 

There are different ways to construct NSGPs. One common method is to transform the covariance function inputs by non-linear functions which leads to the compositional NSGPs. For example, \citet{Wilson2016DeepKernel} transform the inputs by deep neural networks and then feed them to GPs, which retains a closed-form GP regression model. \citet{damianou2013} feed the outputs of GPs to another layer of GPs, which leads to deep GPs. However, deep GPs often require complicated inference methods due to the hierarchy in the probabilistic model~\citep{Salimbeni2017Doubly}. 

An alternative NSGP treatment is to make the parameters of GPs depend on the input~\citep{OLDPaul1992, higdon1999non}. For example, in temporal GPs, the parameters (e.g., length-scale) can be formulated as deterministic functions of time~\citep{higdon1999non} or random functions of time~\citep{Salimbeni2017ns, Matthew2018JMLR, lassi2019Matern}. This hierarchical input-dependent setting enriches the GP prior in the sense that the GP can have varying characteristics depending on the position of data. However, this NSGP treatment requires careful construction of valid non-stationary covariance functions, such as the one proposed by~\citet{paciorek2006spatial}. The inference methods for this type of NSGPs are also often computationally expensive due to the (non-linear) model hierarchy.

One approach is to use stochastic (partial) differential equation based constructions for this type of NSGPs~\citep{zhao2020SSDGP, Emzir2020}, which have the advantage that the covariance function is always valid by construction. In the case of temporal processes, one can convert NSGPs into state-space models which can be solved by Bayesian smoothers efficiently~\citep{zhao2020SSDGP}. 

This paper is considered with $L^1$-regularized extensions of the aforementioned temporal hierarchical NSGPs. The aim is to include $L^1$-regularization on the GP nodes, especially GP priors of parameters, in order to achieve sparseness or total variation regularization. This kind of regularization is useful in many signal processing and inverse problems applications~\citep{Emmanuel2008, Kaipio2005, Unser2014Book}. The regularization is also closely related to the LASSO regression and elastic nets~\citep{TrevorBook2015}. We consider two NSGP constructions which use explicitly constructed non-stationary covariance functions~\citep{Salimbeni2017ns, heinonen2016} or stochastic differential equation (SDE) representations~\citep{zhao2020SSDGP}, respectively.

\subsection{Related work and contributions}
\label{sec:review-and-contribution}
In this section, the aim is to briefly review relevant literature on NSPGs and $L^1$-regularization as well as present the paper contributions and structure. Recent methods and analysis of compositional NSGPs can be found in, for example,~\citet{Wilson2016DeepKernel}, \citet{damianou2013}, \citet{Calandra2016ManifoldGP}, \citet{Shedivat2017RecurrentDKL}, \citet{Wilson2016StochasticVBDKL}, and~\citet{Salimbeni2017Doubly}. Hierarchically parameterized NSGPs have been recently studied, for example, by \citet{zhao2020SSDGP}, \citet{emzir2019}, \citet{Karla2020}, \citet{heinonen2016}, \citet{lassi2019Matern}, \citet{Cheng2019}, \citet{Matthew2018JMLR}, \citet{paciorek2006spatial}, and~\citet{Ohagan2003}. In particular, \citet{heinonen2016} and \citet{zhao2020SSDGP} model the parameters of NSGPs as GPs, and they approximate the posterior distribution by using maximum a posteriori (MAP), Markov chain Monte Carlo (MCMC), and Bayesian smoothing methods. \citet{Matthew2018JMLR} discuss the connections between the compositional and parameterized NSGPs.

Sparsity is one application of $L^1$-regularization in R-NSGPs. It is worth noting that in this context sparsity is introduced on the posterior estimates of NSGPs to tackle ill-behaving signals, while sparse GP methods \citep{Snelson2006, quinonero2005unifying} form sparse approximations of GPs to speed up computation. Similarly, \citet{Feng2010Sparse} and~\citet{Kou2014Sparse} study the GP LASSO problem where the sparseness is included on the GP innovation variable to speed-up the computation. \citet{GYi2011} study sparse estimation of multivariate GP parameters to deal with high-dimensional datasets.

The methodology in this paper is closely related to constrained Gaussian processes which constrain the processes or predictions by some equalities or inequalities. One of the early works is \citet{Abrahamsen2001} who used inequality constraints in kriging. \citet{Jidling2017CGP} studied a GP regression problem where the GP must obey a linear partial differential equation. This constraint has a physical meaning in modeling the magnetic field which is curl-free. Similarly, \citet{Agrell2019} studied linear operator inequality constrained GPs. Other works on this topic include, for example,~\citet{Maatouk2017}, \citet{Jidling2018},~\citet{Solin:2018},~\citet{LangeHegermann2018}, and~\citet{Swiler2020}. However, the prior work is only concerned with conventional GPs instead of NSGPs which have a hierarchy induced by the parameter processes. 

The contributions of this paper are as follows. 1) We introduce regularized NSGPs (R-NSGPs), where the processes are estimated under $L^1$-regularization. 2) We present R-NSGP constructions for two commonly used types of NSGPs. 3) We develop computational methods for R-NSGP regression problems based on the alternating direction method of multipliers (ADMM). 4) We present a theoretical convergence analysis of the computational methods under mild conditions. 5) The experimental results demonstrate the efficiency of R-NSGPs on synthetic and real datasets.

The structure of this paper is as follows. In Section~\ref{sec:problem-formulation}, we formulate the non-stationary Gaussian process regression problems with $L^1$-regularization. In Section~\ref{sec:admm}, we introduce the ADMM method to solve the regularized NSGP regression problems. Sections~\ref{sec:convergence-analysis} and~\ref{sec:experiment} contain the convergence analysis and the numerical experiments, respectively.

\section{Problem formulation}
\label{sec:problem-formulation}
Consider a non-stationary Gaussian process (NSGP) regression problem with a hierarchical structure
\begin{equation}
	\begin{split}
		f(t)\mid u^\ell(t), u^\sigma(t) &\sim \mathcal{GP}(0, C_f(t, t';u^\ell, u^\sigma)),\\
		u^\ell(t) &\sim \mathcal{GP}(0, C_\ell(t, t')),\\
		u^\sigma(t) &\sim \mathcal{GP}(0, C_\sigma(t, t')),\\
		y_k &= f(t_k) + r_k,
	\end{split}
	\label{equ:ns-gp-reg-model}
\end{equation}
where $f(t)\colon \T\to\R$ is a conditional GP depending on two other GPs $u^\ell(t)\colon \T\to\R$ and $u^\sigma(t)\colon \T\to\R$. We denote by $\T = \left\lbrace t\in\R\colon t\geq t_0 \right\rbrace$ a totally-ordered temporal domain with an initial time $t_0$.

The GPs $u^\ell(t)$ and $u^\sigma(t)$ parameterize the length-scale $\ell(t)$ and magnitude $\sigma(t)$ of $f(t)$, respectively, by some transformation $g\colon\R\to(0, +\infty)$:
\begin{equation}
	\begin{split}
		\ell(t) &= g(u^\ell(t)), \\
		\sigma(t) &= g(u^\sigma(t)).
	\end{split}
\end{equation}
The function $g$ is analogous to the activation function of neural networks in the sense that it determines how the values of $u^\ell(t)$ and $u^\sigma(t)$ affect the parameters $\ell(t)$ and $\sigma(t)$ of $f(t)$. Some commonly used examples include the exponential function $g(u)=\exp(u)$ and the softplus function $g(u) = \exp(u) \,/\, (1 + \exp(u))$. In some cases, it is also useful to introduce a baseline level inside the function. For example, letting $\ell(t) = \exp(u^\ell(t) + b^\ell)$ implies that the length-scale $\ell(t)$ stays most of the time at the baseline level $b^\ell$ when $u^\ell(t)$ is sparse.

The measurements $y_k$ of $f(t)$ in \eqref{equ:ns-gp-reg-model} are contaminated by Gaussian noises $r_k\sim\mathcal{N}(0,R_k)$ for $k=1,2,\ldots$. We let $y_{1:T} = \left\lbrace y_k\colon k=1,2\ldots, T\right\rbrace$ be the set of measurements, the time interval $\Delta t_k = t_{k} - t_{k-1} > 0$ for $k=1,2,\ldots, T$, and the noise covariance $\cu{R}=\diag{\left(R_1, R_2,\ldots,R_T \right)}$. 

We assume, without loss of generality, that the mean functions of $f(t)$, $u^\ell(t)$, and $u^\sigma(t)$ are zero, and we endow the GPs with covariance functions $C_f\colon\T\times\T\to\R$, $C_\ell\colon\T\times\T\to\R$, and $C_\sigma\colon\T\times\T\to\R$, respectively. One possible option for a non-stationary covariance function $C_f$ for $f(t)$ that takes processes $u^\ell(t)$ and $u^\sigma(t)$ as parameters is
\begin{equation}
	\begin{split}
		&C_f(t, t';u^\ell,u^\sigma) = \frac{\sigma(t)\,\sigma(t')\,\ell^{\frac{1}{4}}(t)\,\ell^{\frac{1}{4}}(t')}{\Gamma(\nu)2^{\nu-1}}\bigg(\frac{\ell(t) + \ell(t')}{2}\bigg)^{\!\!-\frac{1}{2}} \\
		&\qquad\times\left(2\,\sqrt{\nu\frac{2(t-t')^2}{\ell(t) + \ell(t')}}\right)^{\!\nu}\,K_\nu\left(2\,\sqrt{\nu \frac{2(t-t')^2}{\ell(t) + \ell(t')}}\right),
		\label{equ:Cf-paciorek}
	\end{split}
\end{equation}
which is a non-stationary generalization of the conventional \matern covariance function by \citet{paciorek2006spatial}. Above in Equation~\eqref{equ:Cf-paciorek}, $K_\nu$ is the modified Bessel function of the second kind, $\Gamma$ is the Gamma function, and $\nu$ is a smoothness parameter. The covariance functions for $u^\ell(t)$ and $u^\sigma(t)$ do not need to be non-stationary because we are mostly concerned with the non-stationarity of $f(t)$.

NSGP regression aims to learn the posterior density
\begin{equation}
	\begin{split}
		&p(f(t), u^\ell(t), u^\sigma(t)\mid y_{1:T}) \\
		&\quad= \frac{p(y_{1:T}\mid f(t))\,p(f(t), u^\ell(t), u^\sigma(t))}{p( y_{1:T})},
	\end{split}
	\label{equ:ns-gp-posterior}
\end{equation}
which is unfortunately often intractable due to the non-linear hierarchy in the prior $p(f(t),\allowbreak u^\ell(t), u^\sigma(t))$. Approximation methods for the posterior inference include, for example, variational Bayes and Markov chain Monte Carlo methods~\citep{heinonen2016, Salimbeni2017ns}. In the temporal case, Bayesian smoothing methods can be used~\citep{zhao2020SSDGP}. 

Our interest is now to approximate the posterior density~\eqref{equ:ns-gp-posterior} while having the processes $f(t)$, $u^\ell(t)$, and $u^\sigma(t)$ estimated under $L^1$-regularization. Similarly to the LASSO regression and elastic nets~\citep{TrevorBook2015}, the regularization is realized in the maximum a posteriori (MAP) sense by augmenting the optimized function with additional regularization terms. 

In the sequel, we introduce the $L^1$-regularization on two types of NSGPs. The first one is based on explicit covariance function constructions such as~\eqref{equ:Cf-paciorek}, which we call the \textit{batch} NSGP. The second approach is based on stochastic differential equation representations of GPs, which we call the \textit{state-space} NSGP. The equivalence (in terms of the covariance functions) of these two constructions is given in Lemma~\ref{lemma:ss-cov}.

\subsection{Regularization in batch construction}
\label{sec:reg-batch}
In this section, the NSGP in Equation~\eqref{equ:ns-gp-reg-model} is constructed in a way that the covariance functions $C_f$, $C_\ell$, and $C_\sigma$ are explicitly given such as the one in Equation~\eqref{equ:Cf-paciorek}. Let $f_{1:T}\in\R^{T}$, $u^\ell_{1:T}\in\R^{T}$, and $u^\sigma_{1:T}\in\R^{T}$ denote the vectors of values of $f(t)$, $u^\ell(t)$, and $u^\sigma(t)$ at $t_1, t_2, \allowbreak \ldots, t_T$, respectively. By taking two times the negative logarithm of the unnormalized posterior density in Equation~\eqref{equ:ns-gp-posterior}, we get
\begin{align}
	&\mathcal{L}^{\mathrm{NSGP}} \coloneqq \mathcal{L}^{\mathrm{NSGP}}(f_{1:T}, u^\ell_{1:T}, u^\sigma_{1:T}) \nonumber\\
	&=\norm{f_{1:T} -y_{1:T}}^2_{\cu{R}} + \norm{f_{1:T}}^2_{\cu{C}_f} + \log \abs{2\,\pi\,\cu{C}_f} \label{equ:ns-gp-map}\\
	&\quad+\norm{u^\ell_{1:T}}^2_{\cu{C}_\ell} + \log \abs{2\,\pi\,\cu{C}_\ell}
	+\norm{u^\sigma_{1:T}}^2_{\cu{C}_\sigma} + \log \abs{2\,\pi\,\cu{C}_\sigma}.\nonumber
\end{align}
Above, we denote by $\norm{x}_G = \left(x^\trans\,G^{-1}\,x \right)^{1/2}$ the $G$-weighted Euclidean norm for any real vector $x$ and positive definite matrix $G$. Otherwise $\norm{\cdot}_2$ is understood as the Euclidean norm. We use $\cu{C}_f\in\R^{T\times T}$, $\cu{C}_\ell\in\R^{T\times T}$, and $\cu{C}_\sigma\in\R^{T\times T}$ to denote the covariance matrices obtained by evaluating the corresponding covariance functions $C_f$, $C_\ell$, and $C_\sigma$ at the Cartesian grid $(t_1, \ldots, t_T)\times (t_1, \ldots, t_T)$. It is important to recall that the matrix $\cu{C}_f$ depends on variables $u^\ell_{1:T}$ and $u^\sigma_{1:T}$ non-linearly. Hence the objective function~\eqref{equ:ns-gp-map} is non-linear and non-convex with respect to arguments $u^\ell_{1:T}$ and $u^\sigma_{1:T}$.

To estimate the variables $f_{1:T}$, $u^\ell_{1:T}$, and $u^\sigma_{1:T}$ under the $L^1$-regularization, let us introduce an additional term
\begin{equation}
	\begin{split}
		&\mathcal{L}^{\mathrm{REG}} \coloneqq \mathcal{L}^{\mathrm{REG}}(f^f_{1:T},u^\ell_{1:T}, u^\sigma_{1:T})\\
		&= \lambda_f \norm{\Phi_f\, f_{1:T}}_1 + \lambda_\ell \norm{\Phi_\ell\, u^\ell_{1:T}}_1 + \lambda_\sigma \norm{\Phi_\sigma\, u^\sigma_{1:T}}_1
	\end{split}
\end{equation}
to Equation~\eqref{equ:ns-gp-map}, where $\Phi_f\in\R^{T\times T}$, $\Phi_\ell\in\R^{T\times T}$, $\Phi_\sigma\in\R^{T\times T}$ are some regularization matrices; $\lambda_f>0$, $\lambda_\ell>0$, and $\lambda_\sigma>0$ are strength parameters; and $\norm{\cdot}_1$ is the $L^1$-norm. This $L^1$-regularization term enforces the $\Phi$-transformed parameters to be sparse in the $L^1$ sense. Example choices of $\Phi$ include the identity matrix and finite difference matrices which lead to sparsity promoting and total variation regularizations, respectively. Now the batch R-NSGP aims to solve 
\begin{equation}
	\left\lbrace f_{1:T}, u^\ell_{1:T}, u^\sigma_{1:T}\right\rbrace = \argmin_{f_{1:T}, u^\ell_{1:T}, u^\sigma_{1:T}} \mathcal{L}^{\mathrm{NSGP}} + \mathcal{L}^{\mathrm{REG}}.
	\label{equ:ns-gp-sparse}
\end{equation}
However, solving the above Equation~\eqref{equ:ns-gp-sparse} is often computationally expensive for large $T$ since $\mathcal{L}^{\mathrm{NSGP}}$ contains $T$-dimensional matrix inversions. Another computational difficulty is that the covariance matrices $\cu{C}_f$, $\cu{C}_\ell$, and $\cu{C}_\sigma$ might be numerically close to singular when the data positions $t_1,t_2,\ldots, t_T$ are dense.

\subsection{Regularization in state-space construction}
\label{sec:reg-ss}
A state-space NSGP (SS-NSGP) is an alternative construction to batch NSGP, which does not require choosing valid covariance functions beforehand~\citep{zhao2020SSDGP}. State-space NSGPs use stochastic differential equations (SDEs) to represent each conditional GP in Equation~\eqref{equ:ns-gp-reg-model} instead of specifying the mean and covariance functions. This state-space construction is especially useful because one can leverage the Markov property of SDEs to perform the regression in linear computational time.

To proceed, let us define states $\cu{f}(t)\colon\T\to\R^{D_f}$, $\cu{u}^\ell(t)\colon\T\to\R^{D_\ell}$, and $\cu{u}^\sigma(t)\colon\T\to\R^{D_\sigma}$, so that the GPs $f(t) = \cu{H}_f\,\cu{f}(t)$, $u^\ell(t) = \cu{H}_\ell\,\cu{u}^\ell(t)$, and $u^\sigma(t) = \cu{H}_\sigma \,\cu{u}^\sigma(t)$ are extracted from the states by some matrices $\cu{H}_f\in\R^{D_f}$, $\cu{H}_\ell\in\R^{D_\ell}$, and $\cu{H}_\sigma\in\R^{D_\sigma}$. The states might also contain the derivatives of the process (e.g., in smooth \matern construction). The state-space representation of the corresponding NSGP in Equation~\eqref{equ:ns-gp-reg-model} reads
\begin{align}
	\diff\cu{f}(t) &= \cu{A}\big(\cu{u}^\ell(t)\big)\,\cu{f}(t)\diff t + \cu{B}\big(\cu{u}^\ell(t), \cu{u}^\sigma(t)\big)\diff \cu{W}(t), \nonumber\\
	\diff\cu{u}^\ell(t)& = \cu{A}^\ell\,\cu{u}^\ell(t)\diff t + \cu{B}^\ell\diff \cu{W}^\ell(t), \label{equ:ss-gp-sde}\\
	\diff\cu{u}^\sigma(t) &= \cu{A}^\sigma\,\cu{u}^\sigma(t)\diff t + \cu{B}^\sigma\diff \cu{W}^\sigma(t), \nonumber\\
	y_k &= \cu{H}_f\,\cu{f}(t_k) + r_k,\nonumber
\end{align}
where $\cu{A}(t)\coloneqq\big( \cu{A}\circ \cu{u}^\ell\big) (t)\colon\T\rightarrow \R^{D_f\times D_f}$, $\cu{B}(t)\coloneqq\big( \cu{B}\circ (\cu{u}^\ell, \cu{u}^\sigma)\big) (t)\colon\T\rightarrow \R^{D_f\times D^W_f}$ are the SDE coefficients of $\cu{f}(t)$. The driving terms $\cu{W}(t)\colon\T\rightarrow \R^{D^W_f}$, $\cu{W}^\ell(t)\colon\T\rightarrow \R^{D^W_\ell}$, and $\cu{W}^\sigma(t)\colon\T\rightarrow \R^{D^W_\sigma}$ are mutually independent Wiener processes with unit spectral densities, and we define the associated stochastic integrals in the It\^{o} sense. The coefficients $\cu{A}^\ell \in\R^{D_\ell\times D_\ell}$, $\cu{A}^\sigma \in\R^{D_\sigma\times D_\sigma}$, $\cu{B}^\ell \in\R^{D_\ell\times D_\ell^W}$, and $\cu{B}^\sigma \in\R^{D_\sigma\times D_\sigma^W}$ are some constant matrices. The SDEs start from a suitable random initial condition $\left( \cu{f}(t_0), \cu{u}^\ell(t_0), \cu{u}^\sigma(t_0)\right)$ which has a finite second moment. For the model formulation to make sense, the SDE should have a solution which is also unique in some sense. Hence, without abusing the main scope of the paper, we present the existence of strong unique solution of the SDE~\eqref{equ:ss-gp-sde} in Appendix~\ref{appendix:proof-sde-solution}.

The state-space and batch NSGP constructions are equivalent in the sense that it is possible to find the implied covariance function by a state-space construction. To proceed, we need to fix the randomness from processes $\cu{u}^\ell(t)$ and $\cu{u}^\sigma(t)$, so that we can derive the covariance function of $\cu{f}(t)$ given $\cu{u}^\ell(t)$ and $\cu{u}^\sigma(t)$. Let the SDE system~\eqref{equ:ss-gp-sde} be defined on a probability space $\left(\Omega, \mathcal{F}, \mathbb{P} \right)$, and let $\mathcal{F}^u \subset \mathcal{F}$ be the sub-sigma-algebra generated by $\cu{u}^\ell(t)$ and $\cu{u}^\sigma(t)$ for all $t\in\T$. Also let $\mathbb{P}|_{\mathcal{F}^u}\colon \Omega \to [0, 1]$ be the restricted probability measure on $\left(\Omega, \mathcal{F}^u, \mathbb{P}|_{\mathcal{F}^u} \right)$, that is, $\mathbb{P}|_{\mathcal{F}^u}(E) = \mathbb{P}(E)$ for all $E\in\mathcal{F}^u$~\citep{reneProbabilityBook2005}. We understand the covariance function $C^S_f$ as the conditional covariance function $C_f^S(t, t'; \cu{u}^\ell, \cu{u}^\sigma) \coloneqq \cu{H}_f\,\cov\left[\cu{f}(t), \cu{f}(t') \mid \mathcal{F}^u\right]\,\cu{H}^\trans_f$. We formulate the state-space GP covariance function in the following Lemma~\ref{lemma:ss-cov}.

\begin{lemma}
	\label{lemma:ss-cov}
	Suppose that $\mathbb{P}|_{\mathcal{F}^u}$-almost surely $\cu{A}(t)$ is $t$-continuous. Then $\mathbb{P}|_{\mathcal{F}^u}$-almost surely the conditional covariance function is
	\begin{align}
		&C_f^S(t, t'; \cu{u}^\ell, \cu{u}^\sigma) \label{equ:ss-covs}\\
		&\quad= \cu{H}_f \bigg[ \cu{\Lambda}(t,t_0)\cov\big[\cu{f}(t_0)\mid \cu{u}^\ell(t_0), \cu{u}^\sigma(t_0)\big]\cu{\Lambda}^\trans(t,t_0) \nonumber\\
		&\qquad\qquad+ \int^{t\,\wedge\,t'}_{t_0}\cu{\Lambda}(t,s)\,\cu{B}(s)\,\cu{B}^\trans(s)\,\cu{\Lambda}^\trans(t,s)\diff s \bigg]\,\cu{H}_f^\trans,\nonumber
	\end{align}
	where $\cu{\Lambda}(t,t_0)$ is given by the Peano--Baker series generated by $\cu{A}(t)$. 
\end{lemma}
\begin{proof}
	See Appendix~\ref{appendix:proof-ss-cov}.
\end{proof}
\begin{remark}
	The Peano--Baker series ${\cu{\Lambda}(t, t_0)}$ does not have a closed-form representation in general, except for some special cases. For example, if $\cu{A}$ is one-dimensional or $\cu{A}$ is self-commuting (i.e., $\cu{A}(t)\,\cu{A}(\tau) - \cu{A}(\tau)\,\cu{A}(t)=\cu{0}$ for all $t,\tau\in\T$), then ${\cu{\Lambda}(t, t_0)} = \exp \big( \int^t_{t_0} \cu{A}(s) \diff s\big)$~\citep{Baake2011}.
\end{remark}
\begin{remark}
	The covariance functions of $\cu{u}^\ell(t)$ and $\cu{u}^\sigma(t)$ have closed-form solutions because their associated SDEs are linear time invariant. The results are already known in literature, see, for example, Chapter~6 of~\citet{sarkkabook2019} for details.
\end{remark}

The above Lemma~\ref{lemma:ss-cov} gives covariance function $C^S_f$ of state-space formulation of $f(t)$ given $\cu{u}^\ell(t)$ and $\cu{u}^\sigma(t)$, and we also neglect null events defined by $\mathbb{P}|_{\mathcal{F}^u}$. The proof is rooted in the Peano--Baker series which is widely used for solving linear time dependent ODEs in control theory~\citep{Brogan2011}. In addition to the Peano--Baker series, it is also possible to use the Magnus expansion if one needs an exponential type of transition matrix $\cu{\Lambda}(t, t_0) = \exp\left( \cdot\right)$. However the convergence of Magnus expansion often requires strict conditions on $\cu{A}(t)$~\citep{Moan2008MagnusConv}. Lemma~\ref{lemma:ss-cov} is useful in the sense that it bridges the batch and state-space NSGP constructions. 

The state-space covariance function $C^S_f$ is different from $C_f$ in the sense that $C_f$ uses only the point values of parameters at $t$ and $t'$, while $C^S_f$ also uses all the past information of parameters in integrals. Although evaluating $C^S_f$ often requires numerical techniques, we do not necessarily need to compute it in the state-space NSGP regression. Instead, we can leverage the Markov property to compute the posterior density sequentially. Let $\cu{z}_k\coloneqq\cu{z}(t_k) = \begin{bmatrix}(\cu{f}(t_k))^\trans & (\cu{u}^\ell(t_k))^\trans & (\cu{u}^\sigma(t_k))^\trans\end{bmatrix}^\trans\in\R^{\varrho}$ be the full state vector starting from a Gaussian initial condition $\cu{z}_0 \sim \mathcal{N}(\cu{0}, \cu{P}_0)$, and $\cu{H} = \begin{bmatrix}\cu{H}_f & \cu{0}&\cu{0}\end{bmatrix}$. Due to the Markov property of the process, the posterior density of $\cu{z}_{0:T} = \left\lbrace \cu{z}_k\colon k=0,1,\ldots, T\right\rbrace$ reads
\begin{equation}
	\begin{split}
		&p(\cu{z}_{0:T}\mid y_{1:T})  \\
		&\quad\propto p(\cu{z}_0)\prod^T_{k=1} \mathcal{N}(y_k \mid \cu{H}\, \cu{z}_k, R_k) \prod^T_{k=1} p(\cu{z}_k \mid \cu{z}_{k-1}).
	\end{split}
	\label{equ:ss-nsgp-posterior}
\end{equation}
However, the transition density $p(\cu{z}_k \mid \cu{z}_{k-1})$ of SDEs~\eqref{equ:ss-gp-sde} is usually intractable, and thus we have to approximate it with a discrete-time state-space model of the form
\begin{equation}
	\begin{split}
		\cu{z}_{k} &= \cu{a}(\cu{z}_{k-1}) + \cu{q}(\cu{z}_{k-1}), \\
		\cu{q}(\cu{z}_{k-1})&\sim\mathcal{N}(\cu{0},\cu{Q}(\cu{z}_{k-1})),\\
		y_k &= \cu{H}\,\cu{z}_k + r_k,
		\label{equ:ss-gp-disc}
	\end{split}
\end{equation}
where functions $\cu{a}\colon\R^{\varrho}\to \R^{\varrho}$ and $\cu{q}\colon\R^{\varrho}\to \R^{\varrho}$ depend on the discretization scheme (e.g., the Euler--Maruyama or the Milstein's method). The transition density $p(\cu{z}_k \mid \cu{z}_{k-1}) \approx \mathcal{N}(\cu{z}_k \mid \cu{a}(\cu{z}_{k-1}), \cu{Q}(\cu{z}_{k-1}))$. It is worth pointing out that the Euler--Maruyama routine can lead to singular covariance $\cu{Q}$ in many state-space NSGP constructions, albeit the simplicity in terms of implementation. One can alternatively use the moment-based discretization method to obtain functions $\cu{a}$ and $\cu{Q}$~\citep{zhao2020TME}. 

Similarly to the batch method in Equation~\eqref{equ:ns-gp-sparse}, the regularized state-space NSGP aims to solve
\begin{equation}
	\begin{split}
		\cu{z}_{0:T}  &= \argmin_{\cu{z}_{0:T}} \mathcal{L}^{\mathrm{S-NSGP}} + \mathcal{L}^\mathrm{S-REG},
	\end{split}
	\label{equ:ss-gp-sparse}
\end{equation}
where by taking two times the negative logarithm of Equation~\eqref{equ:ss-nsgp-posterior} we have
\begin{align}
		&\mathcal{L}^{\mathrm{S-NSGP}}= \cu{z}_0^\trans\,\cu{P}_0^{-1}\,\cu{z}_0 \nonumber\\
		&\quad+\sum^T_{k=1}\Big[ \norm{y_k - \cu{H}\,\cu{z}_k}^2_{R_k} + \norm{\cu{z}_k - \cu{a}(\cu{z}_{k-1}) }_{\cu{Q}(\cu{z}_{k-1})}^2 \nonumber\\
		&\qquad\qquad+ \log\abs{2\,\pi\,\cu{Q}(\cu{z}_{k-1})} \Big].\label{equ:s-map}
\end{align}
As opposed to the batch $\mathcal{L}^\mathrm{NSGP}$, the state-space $\mathcal{L}^{\mathrm{S-NSGP}}$ requires linear computational complexity with respect to time. The regularization term $\mathcal{L}^\mathrm{S-REG}$ in sequential formulation reads
\begin{equation}
	\begin{split}
		&\mathcal{L}^\mathrm{S-REG} \\
		&= \sum^T_{k=0}\left[ \lambda_f\norm{\cu{\Psi}_f\,\cu{z}_k}_1 + \lambda_\ell\norm{\cu{\Psi}_\ell\,\cu{z}_k}_1 + \lambda_\sigma\norm{\cu{\Psi}_\sigma\,\cu{z}_k}_1\right],
	\end{split}
\end{equation}
where $\cu{\Psi}_f$, $\cu{\Psi}_\ell$, and $\cu{\Psi}_\sigma$ are regularization matrices chosen to keep consistency with $\mathcal{L}^\mathrm{REG}$. For example, if $\Phi_\ell$ is an identity matrix then we let $\cu{\Psi}_\ell\,\cu{z}_k = u^\ell(t_k)$ such that $\sum^T_{k=0}\norm{\cu{\Psi}_\ell\,\cu{z}_k}_1 = \norm{\Phi_\ell\,u^\ell_{1:T}}_1$. However it might be difficult to find an equivalent $\cu{\Psi}_\ell$ of $\Phi_\ell$ if $\Phi_\ell$ gives correlation across time because $\left\lbrace \cu{z}_k\colon k=0,1,\ldots, T \right\rbrace $ are independent in $\sum^T_{k=0}\norm{\cu{\Psi}_\ell\,\cu{z}_k}_1$. On the other hand, it is straightforward to formulate the regularization among the state components (e.g., derivatives), while in the batch NSGP it is harder. 

Despite the equivalence of the two NSGP constructions shown in Lemma~\ref{lemma:ss-cov}, it is important to recall that the objective functions $\mathcal{L}^{\mathrm{NSGP}}$ and $\mathcal{L}^{\mathrm{S-NSGP}}$ are not necessarily equal because there are approximations involved in the SDE discretization. The accuracy of the approximation varies depending on the discretization scheme that is used~\citep{peter-sde-num, pedersenSDEMLE1995, sarkkabook2019, zhao2020TME}. 

\section{ADMM solution of regularized NSGPs}
\label{sec:admm}
Solving of the R-NSGP regression problem in the MAP sense is a challenging optimization problem because the objective functions are non-linear, non-smooth, and non-convex in general. There is a vast number of studies on solving this type of optimization problems~\citep{Ruszczynski2006}. One commonly used approach is the gradient descent (GD) method, and the gradient in the non-differentiable part is often interpreted as a subgradient~\citep{TrevorBook2015}. However, the GD-based approaches usually suffer from a slow rate of convergence.

In this paper, we solve the problem by using the alternating direction method of multipliers \citep[ADMM,][]{Boyd2011admm}. The ADMM method belongs to the class of variable-splitting methods. It solves the optimization problem by alternatively updating the primal and dual variables in an augmented Lagrangian function until convergence. The main benefit of ADMM is its ability to split a complicated problem into simple subproblems, and in this case, the MAP and regularization terms are optimized one by one. In~\citet{Gao2020Variable, Gao2019ieks}, ADMM was used to solve regularized and constrained non-linear state estimation problems which are mathematically close to the problems in this paper.

\subsection{Batch solution}
\label{sec:batch-nsgp-admm}
In this section, we solve the batch formulation in Equation~\eqref{equ:ns-gp-sparse} using ADMM. We start by introducing auxiliary variables $v_{1:T}^{\ell}\in\R^{T}$ and $v_{1:T}^{\sigma}\in\R^{T}$, and rewrite Equation~\eqref{equ:ns-gp-sparse} as the equality constrained optimization problem 
\begin{align}
		&\min_{\substack{f_{1:T}, u^\ell_{1:T},u^\sigma_{1:T}}}
		\norm{f_{1:T} -y_{1:T}}^2_{\cu{R}} + \norm{f_{1:T}}^2_{\cu{C}_f}
		+ \log \abs{2\, \pi\,\cu{C}_f} \nonumber\\
		&\qquad+ \norm{u^\ell_{1:T}}^2_{\cu{C}_\ell}
		+ \lambda_\ell \norm{ v_{1:T}^{\ell} }_1
		+ \norm{u^\sigma_{1:T}}^2_{\cu{C}_\sigma}
		+ \lambda_\sigma \norm{ v_{1:T}^{\sigma} }_1 \nonumber\\
		&\qquad + \lambda_f \norm{ v_{1:T}^{f} }_1+ \log \abs{2\,\pi\,\cu{C}_\ell} + \log \abs{2\,\pi\,\cu{C}_\sigma} \label{eq:admm}\\
		&\qquad \mathrm{s.t.} \quad v_{1:T}^{f} =\Phi_f\, f_{1:T}, \,
		v_{1:T}^{\ell} =\Phi_\ell\, u^\ell_{1:T}, \,
		v_{1:T}^{\sigma} = \Phi_\sigma\, u^\sigma_{1:T}.\nonumber
\end{align}
Note that the logarithms of determinants of $\cu{C}_\ell$ and $\cu{C}_\sigma$ are constant. We then define the augmented Lagrangian function $\mathcal{L}(f_{1:T}, u^\ell_{1:T},u^\sigma_{1:T},v_{1:T}^{f},v_{1:T}^{\ell},v_{1:T}^{\sigma},\allowbreak \eta_{1:T}^f,\eta_{1:T}^\ell,$ $\eta_{1:T}^\sigma)$ associated with the problem~\eqref{eq:admm} as
\begin{align}
		&\mathcal{L}(f_{1:T}, u^\ell_{1:T},u^\sigma_{1:T},v_{1:T}^{f},v_{1:T}^{\ell},v_{1:T}^{\sigma},\eta_{1:T}^f,\eta_{1:T}^\ell,\eta_{1:T}^\sigma) \nonumber
		\\
		&=\norm{f_{1:T} -y_{1:T}}^2_{\cu{R}} + \norm{f_{1:T}}^2_{\cu{C}_f} + \log \abs{2\,\pi\,\cu{C}_f} \nonumber\\
		&\quad+ \norm{u^\ell_{1:T}}^2_{\cu{C}_\ell} + \norm{u^\sigma_{1:T}}^2_{\cu{C}_\sigma} + \log \abs{2\,\pi\,\cu{C}_\ell} + \log \abs{2\,\pi\,\cu{C}_\sigma} \nonumber
		\\ 
		&\quad+ \lambda_f \norm{ v_{1:T}^{f} }_1  
		+ (\eta_{1:T}^f)^\trans (\Phi_f \,f_{1:T} - v_{1:T}^{f})\nonumber\\
		&\quad+ \lambda_\ell \norm{ v_{1:T}^{\ell} }_1  
		+ (\eta_{1:T}^\ell)^\trans (\Phi_\ell \,u^\ell_{1:T} - v_{1:T}^{\ell}) \label{equ:batch-auglag}\\
		&\quad+ \lambda_\sigma \norm{ v_{1:T}^{\sigma} }_1 
		+ (\eta_{1:T}^\sigma)^\trans (\Phi_\sigma \,u^\sigma_{1:T} - v_{1:T}^{\sigma}) \nonumber\\
		&\quad+ \frac{\rho_f}{2} \norm{\Phi_f \, f_{1:T} -v_{1:T}^{f}}^2_2 \nonumber\\
		&\quad+ \frac{\rho_\ell}{2}  \norm{\Phi_\ell \, u^\ell_{1:T} - v_{1:T}^{\ell}}^2_2
		+ \frac{\rho_\sigma}{2} \norm{\Phi_\sigma \, u^\sigma_{1:T} -v_{1:T}^{\sigma}}^2_2,
\end{align}
where $\rho_f>0$, $\rho_\ell>0$, $\rho_\sigma >0$ are penalty parameters and $\eta_{1:T}^f\in\R^{T}$, $\eta_{1:T}^\ell\in\R^{T}$, $\eta_{1:T}^\sigma\in\R^{T}$ are Lagrange multipliers. The ADMM method iteratively finds the optimum values for the variables $\big\lbrace f_{1:T}, u^\ell_{1:T}, u^\sigma_{1:T}, v_{1:T}^{f}, v_{1:T}^{\ell},\allowbreak v_{1:T}^{\sigma}\big\rbrace $ by forming a sequence of estimates $\Big\lbrace f_{1:T}^{(i)}, u^{\ell, (i)}_{1:T}, u^{\sigma,(i)}_{1:T},v_{1:T}^{f,(i)}, v_{1:T}^{\ell,(i)}, v_{1:T}^{\sigma,(i)}, \eta_{1:T}^{f,(i)}, \eta_{1:T}^{\ell,(i)}, \eta_{1:T}^{\sigma,(i)}\colon$ $ i=0,1,\ldots\Big\rbrace $.

Now suppose that the initial values at $i=0$ are given, then at each iteration the estimates are updated by solving the following subproblems:
\begin{align}
	&\left\lbrace f_{1:T}^{(i+1)}, u_{1:T}^{\ell,(i+1)},u_{1:T}^{\sigma,(i+1)}\right\rbrace \nonumber\\
	&= \argmin_{f_{1:T}, u^\ell_{1:T},u^\sigma_{1:T}} 
	\norm{f_{1:T} -y_{1:T}}^2_{\cu{R}} +\norm{f_{1:T}}^2_{\cu{C}_f} \label{equ:batch-sub-fu}\\ 
	&+\norm{u^\ell_{1:T}}^2_{\cu{C}_\ell}
	+\log \abs{2\,\pi \,\cu{C}_f}  +\norm{u^\sigma_{1:T}}^2_{\cu{C}_\ell} \nonumber\\
	&+ \log \abs{2\,\pi\,\cu{C}_\ell} + \log \abs{2\,\pi\,\cu{C}_\sigma} \nonumber \\
	&+ \Big(\eta_{1:T}^{f,(i)}\Big)^\trans \Big( \Phi_f \,f_{1:T} - v_{1:T}^{f,(i)} \Big) + \frac{\rho_f}{2}  \norm*{\Phi_f \,f_{1:T} - v_{1:T}^{f,(i)} }^2_2 \nonumber\\
	&+ \Big(\eta_{1:T}^{\ell,(i)}\Big)^\trans \Big( \Phi_\ell \,u^\ell_{1:T} - v_{1:T}^{\ell,(i)} \Big) + \frac{\rho_\ell}{2}  \norm*{\Phi_\ell \,u^\ell_{1:T} - v_{1:T}^{\ell,(i)} }^2_2 \nonumber\\
	&+ \Big(\eta_{1:T}^{\sigma,(i)}\Big)^\trans \Big( \Phi_\sigma u^\sigma_{1:T} - v_{1:T}^{\sigma,(i)} \Big) + \frac{\rho_\sigma}{2}  \norm*{\Phi_\sigma u^\sigma_{1:T} - v_{1:T}^{\sigma,(i)} }^2_2,\nonumber
\end{align}
\begin{equation}
	\begin{split}
		& v_{1:T}^{f,(i+1)} 
		=\argmin_{ v_{1:T}^{f}} 
		\lambda_f \,\norm{ v_{1:T}^{f} }_1 \\
		&\qquad\qquad+ \frac{\rho_f}{2}  \norm*{\Phi_f \, f^{(i+1)}_{1:T} - v_{1:T}^{f} + {\eta_{1:T}^{f,(i)}}\,/\,\rho_f}^2_2,\\  
		& v_{1:T}^{\ell,(i+1)} 
		=\argmin_{ v_{1:T}^{\ell}} 
		\lambda_\ell \,\norm{ v_{1:T}^{\ell} }_1 \\
		&\qquad\qquad+ \frac{\rho_\ell}{2}  \norm*{\Phi_\ell \, u^{\ell,(i+1)}_{1:T} - v_{1:T}^{\ell} + {\eta_{1:T}^{\ell,(i)}}\,/\,\rho_\ell}^2_2,\\  
		& v_{1:T}^{\sigma,(i+1)} 
		=\argmin_{ v_{1:T}^{\sigma}} 
		\lambda_\sigma \, \norm{ v_{1:T}^{\sigma} }_1 \\
		&\qquad\qquad+ \frac{\rho_\sigma}{2}  \norm*{\Phi_\ell \, u^{\sigma,(i+1)}_{1:T} - v_{1:T}^{\sigma} + {\eta_{1:T}^{\sigma,(i)}}\,/\,\rho_\sigma}^2_2,
		\label{equ:batch-sub-v}
	\end{split}
\end{equation}
\begin{equation}
	\begin{split}
		&\eta_{1:T}^{f,(i+1)}= \eta^{f,(i)}_{1:T}+ \rho_f \left( \Phi_f \, f^{(i+1)}_{1:T} - v_{1:T}^{f,(i+1)} \right), \\
		&\eta_{1:T}^{\ell,(i+1)}= \eta^{\ell,(i)}_{1:T}+ \rho_\ell \left( \Phi_\ell \, u^{\ell,(i+1)}_{1:T} - v_{1:T}^{\ell,(i+1)} \right), \\
		&\eta_{1:T}^{\sigma,(i+1)}= \eta^{\sigma,(i)}_{1:T} + \rho_\sigma \left(  \Phi_\sigma \, u^{\sigma,(i+1)}_{1:T} - v_{1:T}^{\sigma,(i+1)}\right).
		\label{equ:batch-sub-eta}
	\end{split}
\end{equation}
The $\{f_{1:T}, u_{1:T}^{\ell},u_{1:T}^{\sigma}\}$-subproblem in Equation~\eqref{equ:batch-sub-fu} is an unconstrained problem, and the objective function is differentiable everywhere. Hence we can use standard non-linear optimization methods~\citep{WrightNumericalOPT2006} to solve it. The solutions to the $v_{1:T}^{\ell}$ and $v_{1:T}^{\sigma}$ subproblems in Equation~\eqref{equ:batch-sub-v} can be obtained via the soft thresholding approach~\citep[see, e.g., Section 2.4.1 of][]{TrevorBook2015}. The ADMM routine for solving the R-NSGP regression problem is summarized in Algorithm~\ref{alg:batch-nd-gp-admm}.

\begin{algorithm}
	\SetAlgoLined
	\KwData{$y_{1:T}$ and $\left\lbrace t_k\colon k=1,2,\ldots,T\right\rbrace $}
	\KwParas{$\lambda_f$, $\lambda_\ell$, $\lambda_\sigma$, $\rho_f$, $\rho_\ell$, $\rho_\sigma$, $\Phi_\ell$, $\Phi_f$, $\Phi_\sigma$, $\cu{R}$, and hyperparameters of GPs $u^\ell(t)$ and $u^\sigma(t)$.}
	\KwInit{$f_{1:T}^{(0)}, u^{\ell, (0)}_{1:T}, u^{\sigma,(0)}_{1:T}, v_{1:T}^{f,(0)},v_{1:T}^{\ell,(0)},\allowbreak v_{1:T}^{\sigma,(0)}, \eta_{1:T}^{f,(0)},\eta_{1:T}^{\ell,(0)}, \allowbreak\eta_{1:T}^{\sigma,(0)}$}
	$i=0$ \;
	\While{not converged}{
	compute $f_{1:T}^{(i+1)}, u_{1:T}^{\ell,(i+1)}$, and $u_{1:T}^{\sigma,(i+1)}$ by optimizing Equation~\eqref{equ:batch-sub-fu} \;
	compute $v_{1:T}^{f,(i+1)}$, $v_{1:T}^{\ell,(i+1)}$, and $ v_{1:T}^{\sigma,(i+1)}$ from Equation~\eqref{equ:batch-sub-v} by using the soft thresholding \;
	compute $\eta_{1:T}^{f,(i+1)}$, $\eta_{1:T}^{\ell,(i+1)}$, and $\eta_{1:T}^{\sigma,(i+1)}$ by Equation~\eqref{equ:batch-sub-eta} \;
	$i = i+1$ \;
	}
	\KwRet{$f_{1:T}^{(i)}, u^{\ell, (i)}_{1:T}, u^{\sigma,(i)}_{1:T}$}
	\caption{Regularized batch NSGP with ADMM (R-NSGP ADMM)}
	\label{alg:batch-nd-gp-admm}
\end{algorithm}

\subsection{State-space solution}
\label{sec:ss-nsgp-admm}
In this section, we solve the state-space-form problem in Equation~\eqref{equ:ss-gp-sparse} using ADMM, which leads to a similar solution to the batch case. Let the collections of auxiliary variables $w^f_k\in\R$, $w^\ell_k\in\R$, and $w^\sigma_k\in\R$ be denoted as $\cu{w}^f_{1:T}=\left\lbrace w^f_k\colon k=0,1,\ldots, T \right\rbrace $,  $\cu{w}^\ell_{1:T}=\left\lbrace w^\ell_k\colon k=0,1,\ldots, T \right\rbrace$, and $\cu{w}^\sigma_{1:T}=\left\lbrace w^\sigma_k\colon k=0,1,\ldots, T \right\rbrace $, respectively. We rewrite Equation~\eqref{equ:ss-gp-sparse} as an equality constrained problem by 
\begin{equation}
	\begin{split}
		&\min_{\cu{z}_{0:T}, \cu{w}^f_{1:T},\cu{w}^\ell_{1:T}, \cu{w}^\sigma_{1:T}} \cu{z}_0^\trans\,\cu{P}_0^{-1}\,\cu{z}_0 \\
		&\quad+ \sum^T_{k=1}\Big[ \norm{y_k - \cu{H}\,\cu{z}_k}^2_{R_k} + \norm{\left(\cu{z}_k - \cu{a}(\cu{z}_{k-1}) \right)}_{\cu{Q}(\cu{z}_{k-1})}^2 \\
		&\qquad\qquad+ \log\abs{2\,\pi\,\cu{Q}(\cu{z}_{k-1})} \Big] \\
		&\quad+\sum^T_{k=0}\left[ \lambda_f\norm{w^f_k}_1 + \lambda_\ell\norm{w^\ell_k}_1 + \lambda_\sigma\norm{w^\sigma_k}_1\right] \\
		&\quad\mathrm{s.t.}\quad w^f_k = \cu{\Psi}_f\,\cu{z}_k, \quad w^\ell_k = \cu{\Psi}_\ell\,\cu{z}_k,\quad w^\sigma_k = \cu{\Psi}_\sigma\,\cu{z}_k.
	\end{split}\nonumber
\end{equation}
The corresponding augmented Lagrangian function reads
\begin{equation}
	\begin{split}
		&\mathcal{L}(\cu{z}_{0:T}, \cu{w}^f_{1:T}, \cu{w}^\ell_{1:T}, \cu{w}^\sigma_{1:T}, \cu{\mu}^f_{1:T},\cu{\mu}^\ell_{1:T}, \cu{\mu}^\sigma_{1:T}) \\
		&= \cu{z}_0^\trans\,\cu{P}_0^{-1}\,\cu{z}_0 + \sum^T_{k=1}\Big[ \norm{y_k - \cu{H}\,\cu{z}_k}^2_{R_k} \\
		&\quad+ \norm{\left(\cu{z}_k - \cu{a}(\cu{z}_{k-1}) \right)}_{\cu{Q}(\cu{z}_{k-1})}^2 + \log\abs{2\,\pi\,\cu{Q}(\cu{z}_{k-1})} \\
		&\quad+ \lambda_f\norm{w^f_k}_1 + \mu_k^f\,(\cu{\Psi}_f \, \cu{z}_k - w^{f}_k) + \frac{\rho_f}{2}\,\norm{\cu{\Psi}_f \, \cu{z}_k - w^{f}_k}^2_2 \\
		&\quad+\lambda_\ell\norm{w^\ell_k}_1 + \mu_k^\ell\,(\cu{\Psi}_\ell \, \cu{z}_k - w^{\ell}_k) + \frac{\rho_\ell}{2}\,\norm{\cu{\Psi}_\ell \, \cu{z}_k - w^{\ell}_k}^2_2 \\
		&\quad+ \lambda_\sigma\norm{w^\sigma_k}_1 + \mu_k^\sigma\,(\cu{\Psi}_\ell \, \cu{z}_k - w^{\sigma}_k) + \frac{\rho_\sigma}{2}\,\norm{\cu{\Psi}_\sigma \, \cu{z}_k - w^{\sigma}_k}^2_2\Big].
	\end{split}
	\label{equ:ss-gp-aug-lag}
\end{equation}

Now suppose that the initial solution $\left\lbrace \cu{z}_{0:T}^{(i)}, \cu{w}^{f, (i)}_{1:T}, \cu{w}^{\ell, (i)}_{1:T}, \cu{w}^{\sigma, (i)}_{1:T}, \cu{\mu}_{1:T}^{f, (i)}, \cu{\mu}_{1:T}^{\ell, (i)}, \cu{\mu}_{1:T}^{\sigma, (i)}\right\rbrace $ at $i=0$ is given. We solve the problem by iteratively solving the following subproblems:
\begin{align}
	\cu{z}_{0:T}^{(i+1)} &= \argmin_{\cu{z}_{0:T}} \cu{z}_0^\trans\,\cu{P}_0^{-1}\,\cu{z}_0 \nonumber\\
	&\quad+\sum^T_{k=1}\Big[ \norm{y_k - \cu{H}\,\cu{z}_k}^2_{R_k} + \norm{\left(\cu{z}_k - \cu{a}(\cu{z}_{k-1}) \right)}_{\cu{Q}(\cu{z}_{k-1})}^2 \nonumber\\
	&\qquad\qquad+ \log\abs{2\,\pi\,\cu{Q}(\cu{z}_{k-1})} \Big] \label{equ:ss-sub-z}\\
	&\quad+ \frac{\rho_f}{2}\sum^T_{k=0}\norm*{\cu{\Psi}_f \, \cu{z}_k - w^{f, (i)}_k + \mu^{f, (i)}_k}_2^2 \nonumber\\
	&\quad+ \frac{\rho_\ell}{2}\sum^T_{k=0}\norm*{\cu{\Psi}_\ell \, \cu{z}_k - w^{\ell, (i)}_k + \mu^{\ell, (i)}_k}_2^2 \nonumber\\
	&\quad+ \frac{\rho_\sigma}{2}\sum^T_{k=0}\norm*{\cu{\Psi}_\sigma \, \cu{z}_k - w^{\sigma, (i)}_k + \mu^{\sigma, (i)}_k}_2^2,\nonumber
\end{align}
and the following for $k=0,1,\ldots,T$:
\begin{equation}
	\begin{split}
		w^{f,(i+1)}_{k} &= \argmin_{w^{f}_{k}} \lambda_f \norm{w^f_{k}}_1\\ 
		&\quad+\frac{\rho_f}{2}\norm*{\cu{\Psi}_f \, \cu{z}_k^{(i+1)} - w^{f, (i)}_k + \mu^{f, (i)}_k}^2_2, \\
		w^{\ell,(i+1)}_{k} &= \argmin_{w^{\ell}_{k}} \lambda_\ell \norm{w^\ell_{k}}_1\\ 
		&\quad+\frac{\rho_\ell}{2}\norm*{\cu{\Psi}_\ell \, \cu{z}_k^{(i+1)} - w^{\ell, (i)}_k + \mu^{\ell, (i)}_k}^2_2, \\
		w^{\sigma,(i+1)}_{k} &= \argmin_{w^{\sigma}_{k}} \lambda_\sigma  \norm{w^\sigma_{k}}_1\\ 
		&\quad+\frac{\rho_\sigma}{2}\norm*{\cu{\Psi}_\sigma \, \cu{z}_k^{(i+1)} - w^{\sigma, (i)}_k + \mu^{\sigma, (i)}_k}^2_2,
		\label{equ:ss-sub-w}
	\end{split}
\end{equation}
along with
\begin{equation}
	\begin{split}
		\mu_{k}^{f,(i+1)} &= \mu^{f,(i)}_{k}+  \left( \cu{\Psi}_f\, \cu{z}_k^{(i+1)} - w_{k}^{f,(i +1 )} \right), \\
		\mu_{k}^{\ell,(i+1)} &= \mu^{\ell,(i)}_{k}+  \left( \cu{\Psi}_\ell\, \cu{z}_k^{(i+1)} - w_{k}^{\ell,(i +1 )} \right), \\
		\mu_{k}^{\sigma,(i+1)}&= \mu^{\sigma,(i)}_{k} +  \left(\cu{\Psi}_\sigma\, \cu{z}_k^{(i+1)} - w_{k}^{\sigma,(i +1 )}\right).
		\label{equ:ss-sub-mu}
	\end{split}
\end{equation}

The solutions to subproblems~\eqref{equ:ss-sub-z},~\eqref{equ:ss-sub-w}, and~\eqref{equ:ss-sub-mu} can be obtained in the same was as in the batch version in Section~\ref{sec:batch-nsgp-admm}. For the sake of space, we do not show the algorithm there, but it can be obtained by modifying Algorithm~\ref{alg:batch-nd-gp-admm} by substituting the corresponding state-space equations. It is also important to mention that the solution to subproblems~\eqref{equ:ss-sub-w} and~\eqref{equ:ss-sub-mu} can be fully vectorized to remove the loop over $k=0,1,\ldots,T$. 

\subsection{Uncertainty quantification}
\label{sec-map-uncertainty}
Because the methods in Section~\ref{sec:admm} for solving the regularized NSGP are based on the MAP approach, they cannot properly account for uncertainties in the estimation results. One approach to tackle this is to use the Laplace's method which approximates the posterior density with a Gaussian by using the MAP estimate as the mean and the negative inverse Hessian as the covariance~\citep{Bishop2006}. However, in our case the objective function is not differentiable and hence this approach is directly applicable. Still, one possible approach is to use the Hessians of the non-regularized objective functions~\eqref{equ:ns-gp-map} and~\eqref{equ:s-map} instead. Although this is in principle straightforward to do by using automatic differentiation toolboxes, it can be computationally demanding. Furthermore, the Hessian is a matrix of dimension determined by the square of the number of measurements $T^2$, which can lead to large memory consumption.

It is also possible to concentrate on approximating the marginal posterior density $p(f(t) \mid y_{1:T})$ instead of the joint posterior density. After the MAP estimates (i.e., $u^\ell_\ast(t)$ and $u^\sigma_\ast(t)$) have been obtained, one can approximate $p(f(t) \mid u^\ell(t), u^\sigma(t), y_{1:T})\approx \widetilde{p}(f(t) \mid u^{\ell}_\ast(t), u^\sigma_\ast(t), y_{1:T})$. Therefore obtaining density $p(f(t) \mid y_{1:T})$ is a conventional GP regression problem~\citep{heinonen2016}. As for the state-space NSGPs, density $p(f(t) \mid y_{1:T})$ can be solved by applying the Kalman filter and Rauch--Tung--Striebel smoother on the linear time varying SDE, where the SDE coefficients are determined by the MAP estimates~\citep{sarkkabook2019}. Although these approaches are computationally cheap, their limitation is that they can only be used to approximate the marginal posterior density $p(f(t) \mid y_{1:T})$ instead of the full posterior density $p(f(t), u^\ell(t), u^\sigma(t) \mid y_{1:T})$.

\section{Convergence analysis}
\label{sec:convergence-analysis}
In this section, we analyze the convergence of the ADMM method for the regularized NSGP. More general discussion on the convergence of augmented Lagrangian splitting methods can be found, for example, in \citet{Boyd2011admm}, \citet{admm2019convergence}, and \citet{Gao2019ieks}. Although the analysis in this section is based on a batch formulation, it also applies to state-space formulations provided that the implied covariance function (when also the discretization is taken into account) satisfies the assumptions. Hence the main goal is to show that the Algorithm~\ref{alg:batch-nd-gp-admm} converges to a local optimum. The main result is revealed in Theorem~\ref{thm:convergence}.

The following notations are used in this section. To simplify the notation, we concatenate variables $f_{1:T}$, $u^\ell_{1:T}$, and $u^\sigma_{1:T}$ into one vector $z_{1:T} = \begin{bmatrix}(f_{1:T})^\trans & (u^\ell_{1:T})^\trans & (u^\sigma_{1:T})^\trans\end{bmatrix}^\trans \in\R^{2\,T}$ so that they are discussed jointly. We let $C^{k}_\infty(M)$ be the space of $k$-times continuously differentiable functions $\varsigma\colon M\to \R$ for which $\sum_{\abs{\alpha}\leq k} \sup_{m\in M} \abs{\varsigma^{(\alpha)}(m)}$ is bounded, where $\varsigma^{(\alpha)}$ denotes the $\alpha$-th derivative of $\varsigma$, and $\alpha$ is a multi-index. We denote by $\nabla$ the gradient, and we also let $\nabla_x^\trans g(x,y)$ be the transpose of the gradient of $g(x,y)$ with respect to $x$. 

\subsection{Theoretical results}
The assumptions used in the theoretical results are the following. 
\begin{assumption}
	\label{assu:g}
	We have $g\in C^{2}_\infty(\R)$ and $g$ is uniformly lower bounded by a positive constant.
\end{assumption}

\begin{assumption}
	\label{assu:eig-of-C}
	$\cu{C}_f$ is strictly positive definite. That is, $\lambda_{\mathrm{min}}(\cu{C}_f) \geq c > 0$ for all $u^\ell_{1:T}\in\R^{T}$ and $u^\sigma_{1:T}\in\R^{T}$, where $\lambda_{\mathrm{min}}$ denotes the smallest eigenvalue. 
\end{assumption}

\begin{assumption}
	\label{assu:f-in-compact}
	$f_{1:T} \in G\subseteq\R^T$, where $G$ is compact. 
\end{assumption}

\begin{assumption}
	\label{assu:rho-large}
	Parameters $\rho_f$, $\rho_\ell$, and $\rho_\sigma$ satisfy
	\begin{equation}
		\begin{split}
			\left( \frac{\rho_f}{2}\,\lambda_{\mathrm{min}}^2(\Phi_f) - \frac{L_z}{2}\right) &\geq 0, \\
			\left( \frac{\rho_\ell}{2}\,\lambda_{\mathrm{min}}^2(\Phi_\ell) - \frac{L_z}{2}\right) &\geq 0, \nonumber
		\end{split}
	\end{equation}
	and
	\begin{equation}
		\left( 	\frac{\rho_\sigma}{2}\,\lambda_{\mathrm{min}}^2(\Phi_\sigma) - \frac{L_z}{2}\right) \geq 0,\nonumber
	\end{equation}
\end{assumption}
where constant $L_z$ is defined in Lemma~\ref{lemma:lip-cont}.

Assumption~\ref{assu:g} is related to the NSGP construction, where we need a suitable positive transformation function $g$. The function is chosen to be twice continuously differentiable, and the function and its derivatives are uniformly bounded. This assumption aims to ensure the Lipschitz continuity of $\nabla\mathcal{L}^\mathrm{NSGP}$ in the following Lemma~\ref{lemma:lip-cont}. This condition is quite strong because it narrows the choice of transformation functions. However, ensuring the positivity of GP parameters is the primary goal for choosing $g$, and in practice one can always bound the functions numerically. The conditions can be further weakened by using weak derivatives and boundedness almost everywhere while still ensuring the Lipschitz continuity of $\nabla\mathcal{L}^\mathrm{NSGP}$~\citep[see, Theorem 4 in Section 5.8 of][]{EvansPDE}.

Assumption~\ref{assu:eig-of-C} is related to the covariance matrix of $f(t)$. Note that $\cu{C}_f$ is the covariance matrix evaluated at the data points, and it is a matrix-valued function with respect to parameters $u^\ell_{1:T}$ and $u^\sigma_{1:T}$ (i.e., $\cu{C}_f\colon \R^{2\,T}\to \R^{T\times T}$). Although $\cu{C}_f$ is positive definite by definition, it might be possible that the smallest eigenvalue of $\cu{C}_f$ approaches zero in limit of some sequences on the domain, and consequently the maximum eigenvalue of $\cu{C}_f^{-1}$ may diverge to infinity. This assumption postulates a positive lower bound on $\cu{C}_f$ so that the determinant of $\cu{C}_f$ does not approach zero. 

Assumption~\ref{assu:f-in-compact} is a technical assumption related to the optimization of subproblem~\eqref{equ:batch-sub-fu}. The value of $f_{1:T}$ during the optimization must be constrained somehow to ensure convergence in all conditions. Due to finite range of floating point values this is always the case. Alternatively, we can use constrained optimizers, such as interior point methods~\citet{WrightNumericalOPT2006} to ensure this.

Assumption~\ref{assu:rho-large} is related to the ADMM solver. The penalty parameters $\rho_f$, $\rho_\ell$, and $\rho_\sigma$ in the augmented Lagrangian function need to be chosen large enough depending on the regularization matrices and the Lipschitz constant $L_z$. The constant $L_z$ which is shown in Lemma~\ref{lemma:lip-cont} is concerned with $\nabla\mathcal{L}^\mathrm{NSGP}$ from the NSGP model. 

The following Lemma~\ref{lemma:lip-cont} is an auxiliary result to the main theorem.
\begin{lemma}
	\label{lemma:lip-cont}
	Under Assumptions~\ref{assu:g} and~\ref{assu:f-in-compact}, there exists a constant $L_z>0$ such that
	\begin{equation}
		\begin{split}
			&\Big\lvert \mathcal{L}^{\mathrm{NSGP}}\left(z^{(i)}_{1:T}\right) - \mathcal{L}^{\mathrm{NSGP}}\left(z_{1:T}^{(j)}\right) \\
			&\quad- \nabla^\trans_{z_{1:T}}\mathcal{L}^{\mathrm{NSGP}}\left(z_{1:T}^{(j)}\right)\,\left(z_{1:T}^{(i)} - z_{1:T}^{(j)}\right) \Big\rvert \\ 
			&\leq\frac{L_z}{2}\,\norm{z_{1:T}^{(i)} - z_{1:T}^{(j)}}^2_2,
			\label{equ:lip-grad}
		\end{split}
	\end{equation}
	holds for all $z_{1:T}^{(i)}, z_{1:T}^{(j)}\in\R^{2 \, T} \times G$.
\end{lemma}
\begin{proof}
	The aim is to show that the Hessian of $\mathcal{L}^{\mathrm{NSGP}}$ (i.e., the the Jacobian of $\nabla_{z_{1:T}} \mathcal{L}^{\mathrm{NSGP}}$) is in a suitable sense bounded, so that we can use the results by~\citet{Yurii2004}. The gradient $\nabla_{z_{1:T}} \mathcal{L}^{\mathrm{NSGP}}$ is derived in Appendix~\ref{appendix:gradients}. Because the Hessian of $\mathcal{L}^{\mathrm{NSGP}}$ is continuous with respect to $f_{1:T}$, the Hessian is entrywise bounded with respect to $f_{1:T}$ due to the compactness in Assumption~\ref{assu:f-in-compact}. We are left with showing the boundedness for $u^\ell_{1:T}$ and $u^\sigma_{1:T}$. For every fixed time steps, the covariance matrix $\cu{C}_f$ is a matrix-valued function of $u^\ell_{1:T}$ and $u^\sigma_{1:T}$. Assumption~\ref{assu:g} implies that $\ell(t) = g(u^\ell(t))$ and $\sigma(t)=g(u^\sigma(t))$ are mapped onto compact intervals in the positive real line for every $t\in\T$. Due to the smoothness of covariance function $C_f$ with respect to $\ell(t)$ and $\sigma(t)$ (i.e., $K_\nu$ being analytic for positive arguments) it follows that the Hessian of $\mathcal{L}^{\mathrm{NSGP}}$ is continuous with respect to $\ell_{1:T}$ and $\sigma_{1:T}$. From Assumption~\ref{assu:g} it further follows that for every $u^\ell_{1:T} \in \R^T$ and $u^\sigma_{1:T} \in \R^T$ the Hessian with respect to $u_{1:T}$ is entrywise bounded. Finally, by Lemma 1.2.2 and 1.2.3 of~\citet{Yurii2004} we arrive at Equation~\eqref{equ:lip-grad}. 
\end{proof}


The following Lemma~\ref{lemma:non-increasing} shows that the sequence generated by Algorithm~\ref{alg:batch-nd-gp-admm} is non-increasing and lower bounded. 

\begin{lemma}
	\label{lemma:non-increasing}
	Let Assumptions~\ref{assu:g}--\ref{assu:rho-large} be satisfied. Then the sequence $\mathcal{L}\Big(z_{1:T}^{(i)},v_{1:T}^{f,(i)}, v_{1:T}^{\ell,(i)},v_{1:T}^{\sigma,(i)},\allowbreak\eta_{1:T}^{f,(i)},\eta_{1:T}^{\ell,(i)},\eta_{1:T}^{\sigma,(i)}\Big)$ generated by Algorithm~\ref{alg:batch-nd-gp-admm} is lower bounded and non-increasing for $i=0,1,2,\ldots$.
\end{lemma}
\begin{proof}
	See, Appendix~\ref{appendix:proof-non-increasing}.
\end{proof}

With help of the above lemma, we now arrive at the main theorem.
\begin{theorem}
	\label{thm:convergence}
	Let Assumptions~\ref{assu:g}--\ref{assu:rho-large} hold and 
	suppose that subproblem~\eqref{equ:batch-sub-fu} has a stationary point, then the sequence $\Big\lbrace z_{1:T}^{(i)},\allowbreak v_{1:T}^{f,(i)}, v_{1:T}^{\ell,(i)}, v_{1:T}^{\sigma,(i)}\colon i=0,1,\ldots\Big\rbrace $ generated 
	by Algorithm~\ref{alg:batch-nd-gp-admm} converges to a local minimum.
\end{theorem}

\begin{proof}
Lemma~\ref{lemma:non-increasing} states that the Lagrangian function $\mathcal{L}\Big(z_{1:T}^{(i)}, v_{1:T}^{f,(i)} ,v_{1:T}^{\ell,(i)},v_{1:T}^{\sigma,(i)},\eta_{1:T}^{f,(i)},\eta_{1:T}^{\ell,(i)},\eta_{1:T}^{\sigma,(i)}\Big)$ is non-increasing and lower bounded for $i=0,1,\ldots$. Also, the subproblem of $v_{1:T}$ in Equation~\eqref{equ:batch-sub-v} is convex and hence has a local minimum~\citep{boyd2004Convex,Nesterov2018optimization}. We deduce that the iterative sequence generated by Algorithm~\ref{alg:batch-nd-gp-admm} converges to a local minimum.
\end{proof}

\section{Experiments}
\label{sec:experiment}
In this section, we numerically demonstrate the regularized NSGPs (R-NSGPs) on a synthetic model as well as on a real dataset\footnote{Companion codes are attached in the supplementary material for the reviewers and will be publicly available upon acceptance.}. Both the batch and state-space constructions (R-SS-NSGP) of NSGPs are present. For the experiments in the sequel we particularly consider sparsity regularization on parameters $u^\ell(t)$ and $u^\sigma(t)$, and the regularization matrices $\Phi_\ell$, $\Phi_\sigma$, $\cu{\Psi}_\ell$, and $\cu{\Psi}_\sigma$ are set to be identity matrices. Regularization on $f(t)$ is not used. Since the proposed R-NSGP is an extension of NSGPs, we are going to compare the R-NSGP with conventional GPs and NSGPs. The fully independent conditional (FIC) and deterministic training conditional (DTC) sparse GPs (SGPs) are also compared as reference methods~\citep{quinonero2005unifying, Csato2002}. The uncertainty quantification of R-NSGP and R-SS-NSGP from Algorithm~\ref{alg:batch-nd-gp-admm} is given by using the marginal approach from Section~\ref{sec-map-uncertainty}.

\subsection{Synthetic rectangular signal}
\label{sec:exp-rect}
Let us consider a noisily measured rectangular signal
\begin{equation}
	\begin{split}
		f(t) &= \begin{cases}
		0, & t\in [0, \frac{1}{3}), \\
		1, & t\in [\frac{1}{3}, \frac{2}{3}), \\
		0.5, & t\in [\frac{2}{3}, 1),
	\end{cases} \\
	y(t_k) &= f(t_k) + r_k, \quad r_k \sim \mathcal{N}(0, 0.002).
	\end{split}
	\label{equ:exp-rect}
\end{equation}

\begin{table*}[t!]
	\centering
	\begin{tabular}{@{}llll@{}}
		\toprule
		\multirow{2}{*}{} & without uncertainty & \multicolumn{2}{l}{with uncertainty} \\ \cmidrule(l){2-4} 
		& RMSE $\times 10^{-2}$ & RMSE $\times 10^{-2}$ & NLPD \\ \midrule
		GP & $3.99 \pm 0.25$ & $3.99 \pm 0.25$ & $-139.84 \pm 6.44$ \\
		SGP FIC & $8.34 \pm 0.28$ & $8.34 \pm 0.28$ & $-119.00 \pm 5.15$ \\
		SGP DTC & $3.47 \pm 0.28$ & $3.47 \pm 0.28$ & $-133.52 \pm 4.66$ \\ \midrule
		NSGP & $4.43 \pm 0.31$ & $4.43 \pm 0.31$ &  $-134.90 \pm 7.05$\\
		R-NSGP GD & $3.69 \pm 0.28$ & $3.69 \pm 0.28$ & $-141.05 \pm 8.96$ \\
		R-NSGP ADMM & $3.56 \pm 0.32$ & $\textbf{3.56} \pm 0.32$ & $\mathbf{-143.17} \pm 8.79$ \\ \midrule
		SS-NSGP & $3.79 \pm 0.26$ & $4.09 \pm 0.25$ & $-136.68 \pm 6.58$ \\
		R-SS-NSGP GD & $1.67 \pm 0.28$ & $6.24 \pm 0.16$ & $-116.05 \pm 10.41$ \\
		R-SS-NSGP ADMM & $\mathbf{1.63} \pm 0.25$ & $6.25 \pm 0.16$ & $-116.14 \pm 10.46$ \\ \bottomrule
	\end{tabular}
	\caption{RMSE and NLPD results with mean and standard deviation on the regression of Equation~\eqref{equ:exp-rect}.}
	\label{tbl:rect-result}
\end{table*}

\begin{figure}[t!]
	\centering
	\includegraphics[width=.8\linewidth]{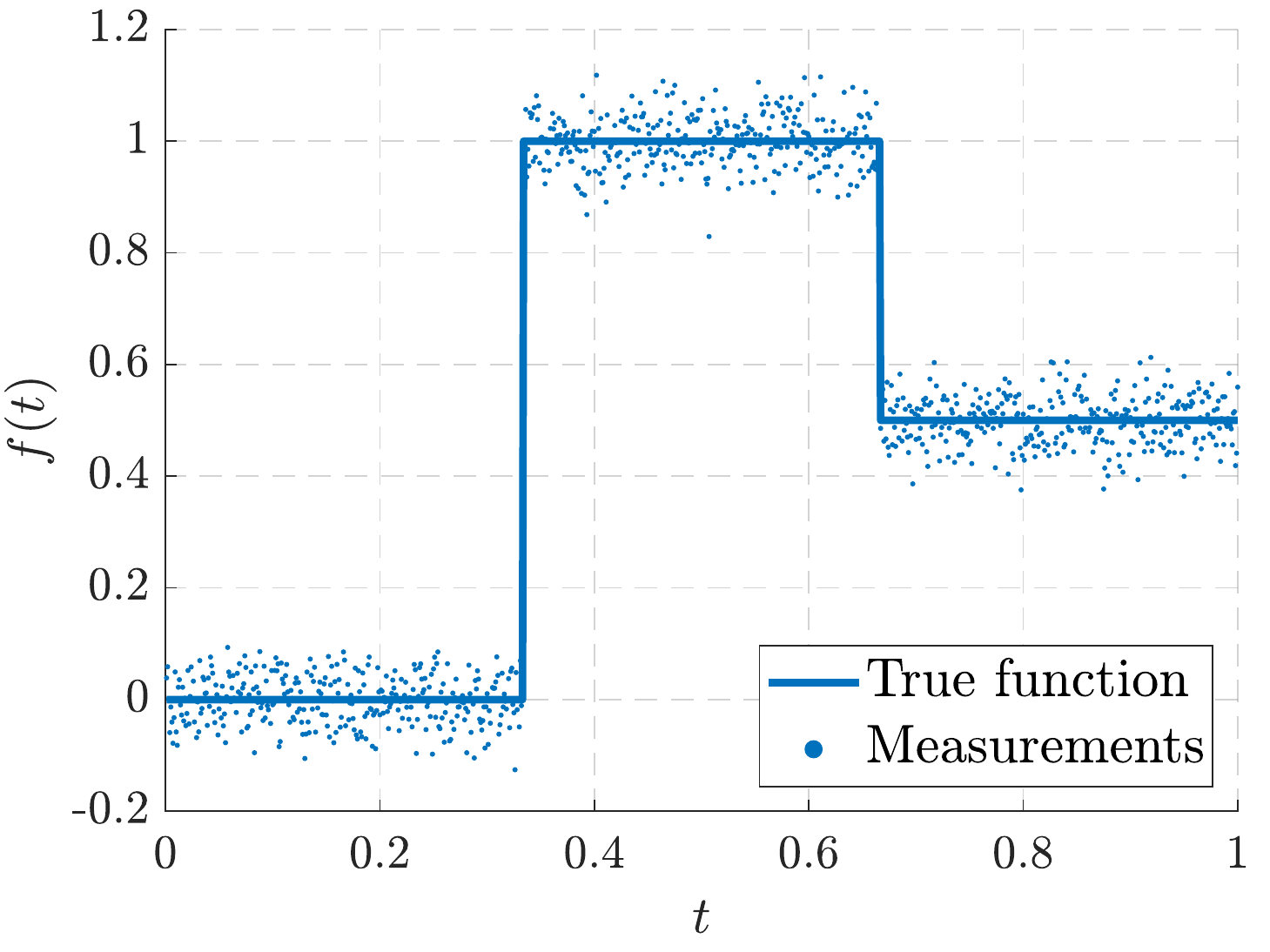}
	\caption{Demonstration of the rectangular signal in Equation~\eqref{equ:exp-rect}.}\label{fig:rect-y} 
\end{figure}

A realization of this model is shown in Figure~\ref{fig:rect-y}. Estimating $f(t)$ from measurements is found to be challenging for conventional GPs and NSGPs because the function is flat almost everywhere and is discontinuous at $t=1/3$ and $t=2/3$. The jumps at $t=1/3$ and $t=2/3$ also have different vertical magnitudes. In this case, GPs require large length-scale $\ell$ and small magnitude $\sigma$ to deal with the flatness. On the other hand, large $\ell$ and small $\sigma$ do not cope well with the discontinuous jumps.

However, we can use R-NSGPs to force the length-scale and magnitude to stay mostly flat in some reasonable baseline levels and then let them jump quickly in some small neighborhoods around the discontinuity points. Consequently, we choose the transformation functions as
\begin{equation}
	\begin{split}
		\ell(t) &= \exp(u^\ell(t) + b^\ell), \\
		\sigma(t) &= \exp(u^\sigma(t) + b^\sigma),\nonumber
	\end{split}
\end{equation}
where $b^\ell = 2$, $b^\sigma = -1$ define the baseline levels. For example, if $u^\ell(t)$ is sparse then $\ell(t)$ is flat around $\exp(2)$. As for the regularization parameters we choose $\lambda_\ell = \lambda_\sigma = 18$ and $\lambda_\ell = \lambda_\sigma = 8$ for the R-NSGP and R-SS-NSGP, respectively. 

\begin{figure*}[t!]
	\centering
	\includegraphics[width=.328\linewidth]{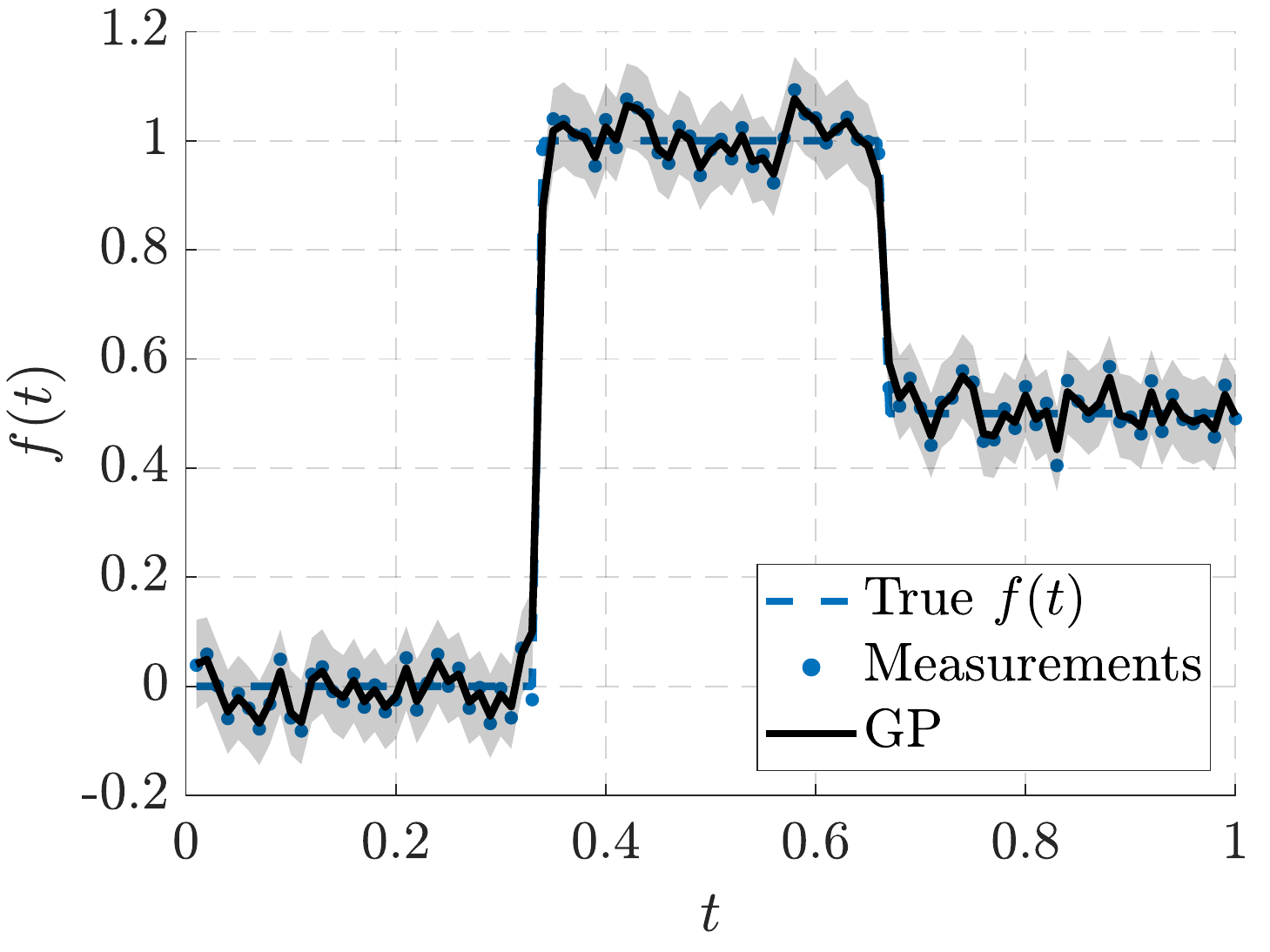}
	\includegraphics[width=.328\linewidth]{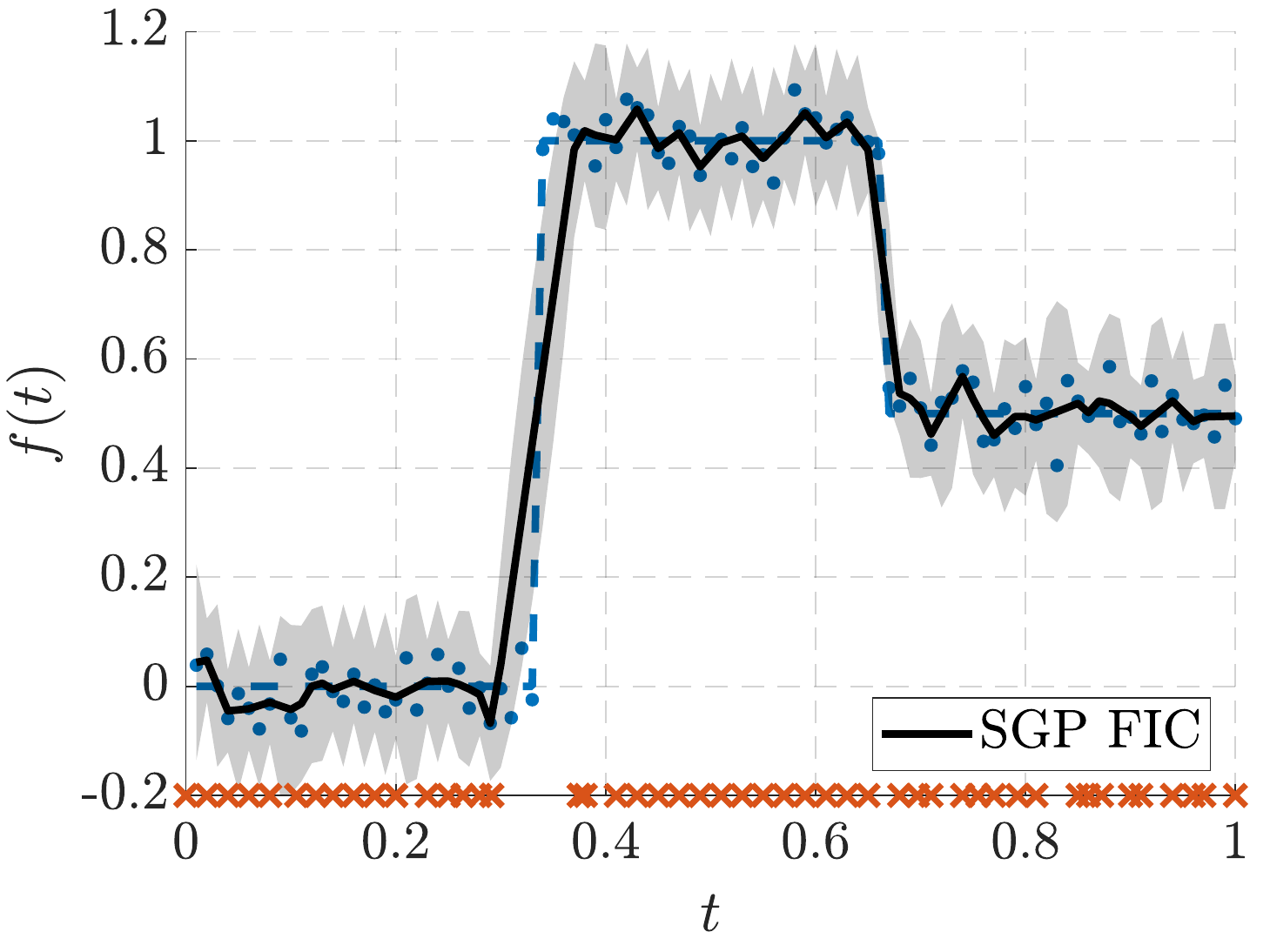}
	\includegraphics[width=.328\linewidth]{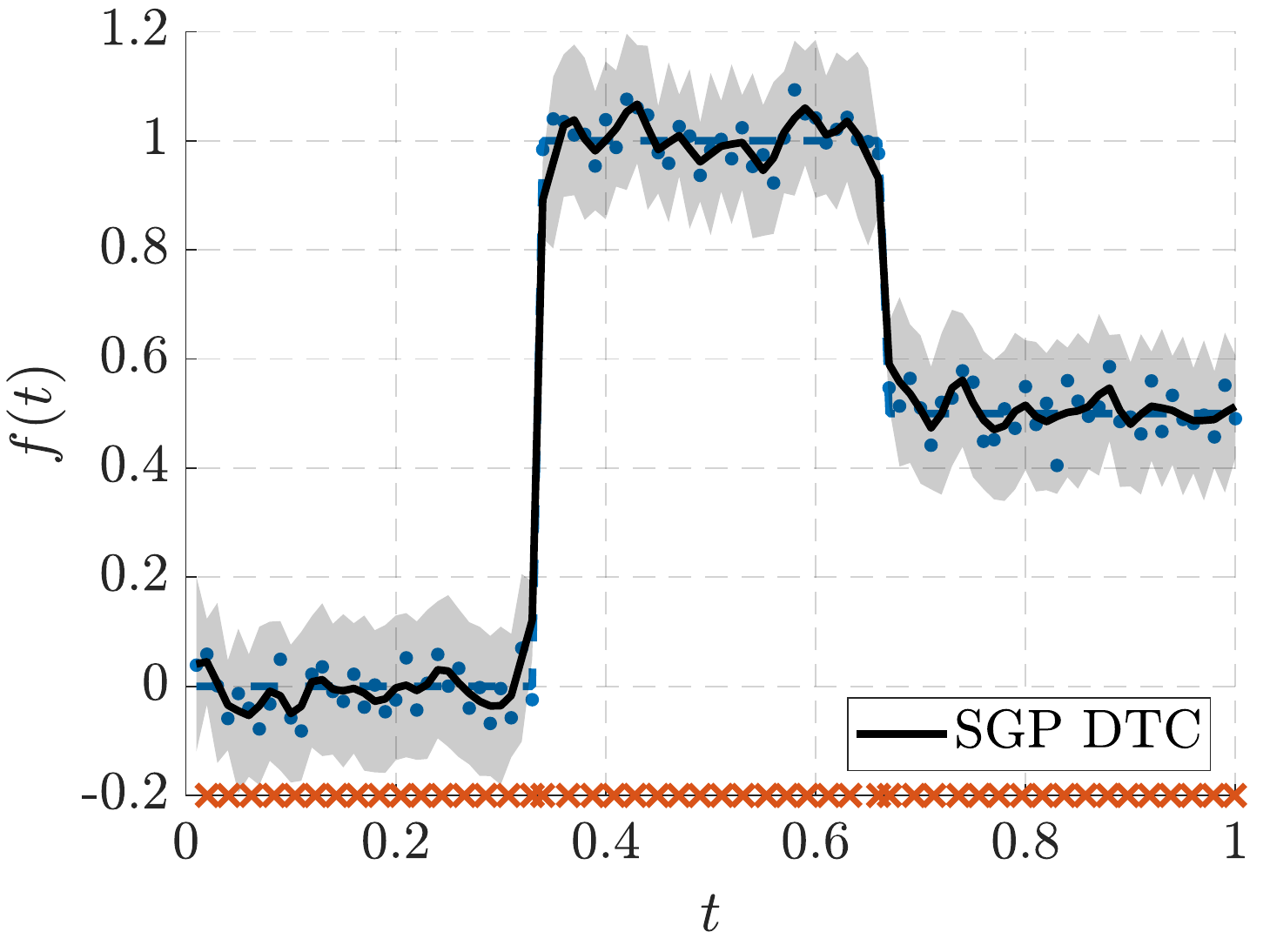}
	\caption{Demonstration of GP, SGP FIC, and SGP DTC on the rectangular signal in Equation~\eqref{equ:exp-rect}. The shaded area stands for 0.95 confidence. The red crosses indicate the positions of learnt 50 inducing points. }
	\label{fig:rect-gp-sgp}
\end{figure*}
\begin{figure*}[t!]
	\centering
	\includegraphics[width=.328\linewidth]{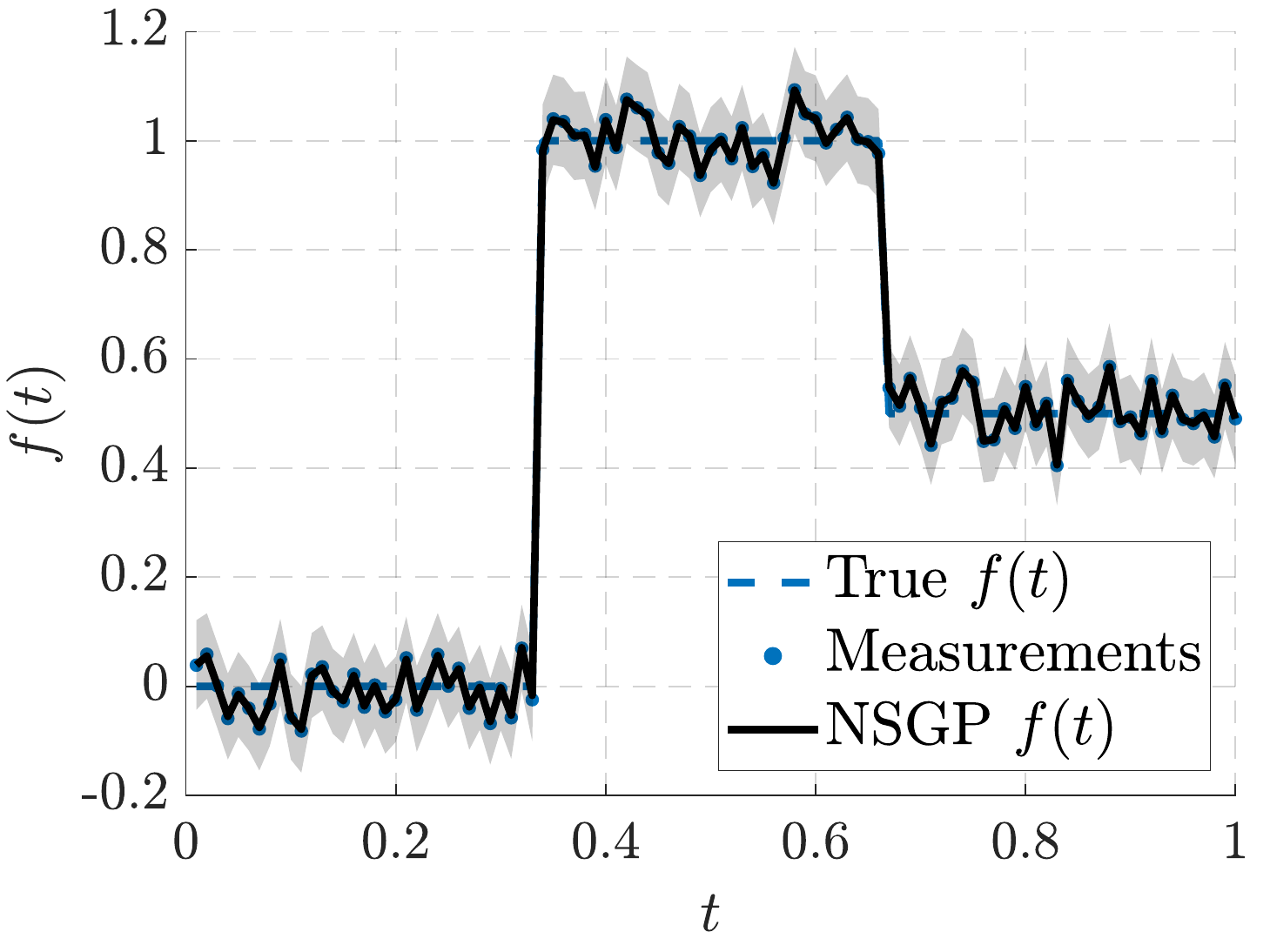}
	\includegraphics[width=.328\linewidth]{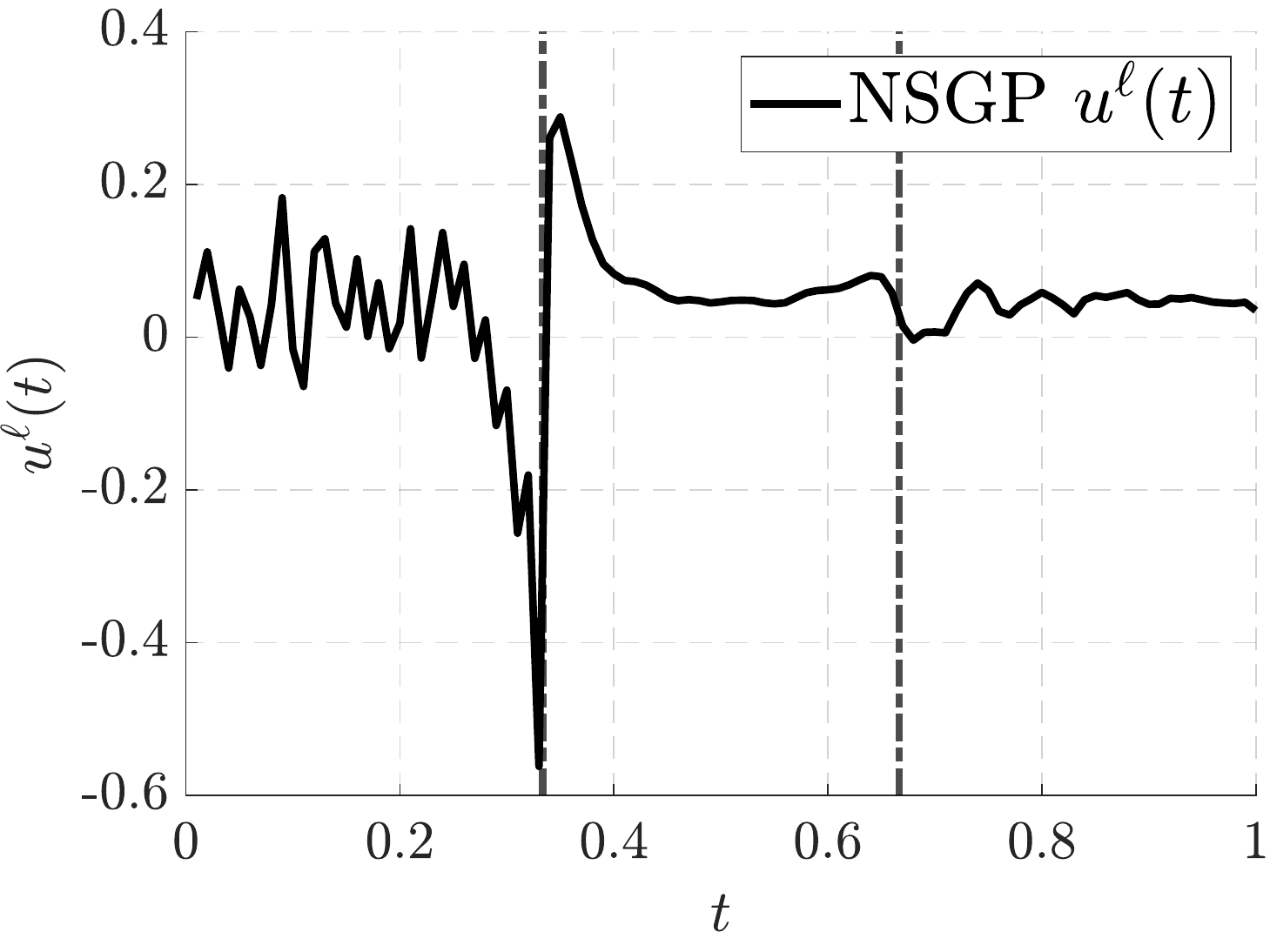}
	\includegraphics[width=.328\linewidth]{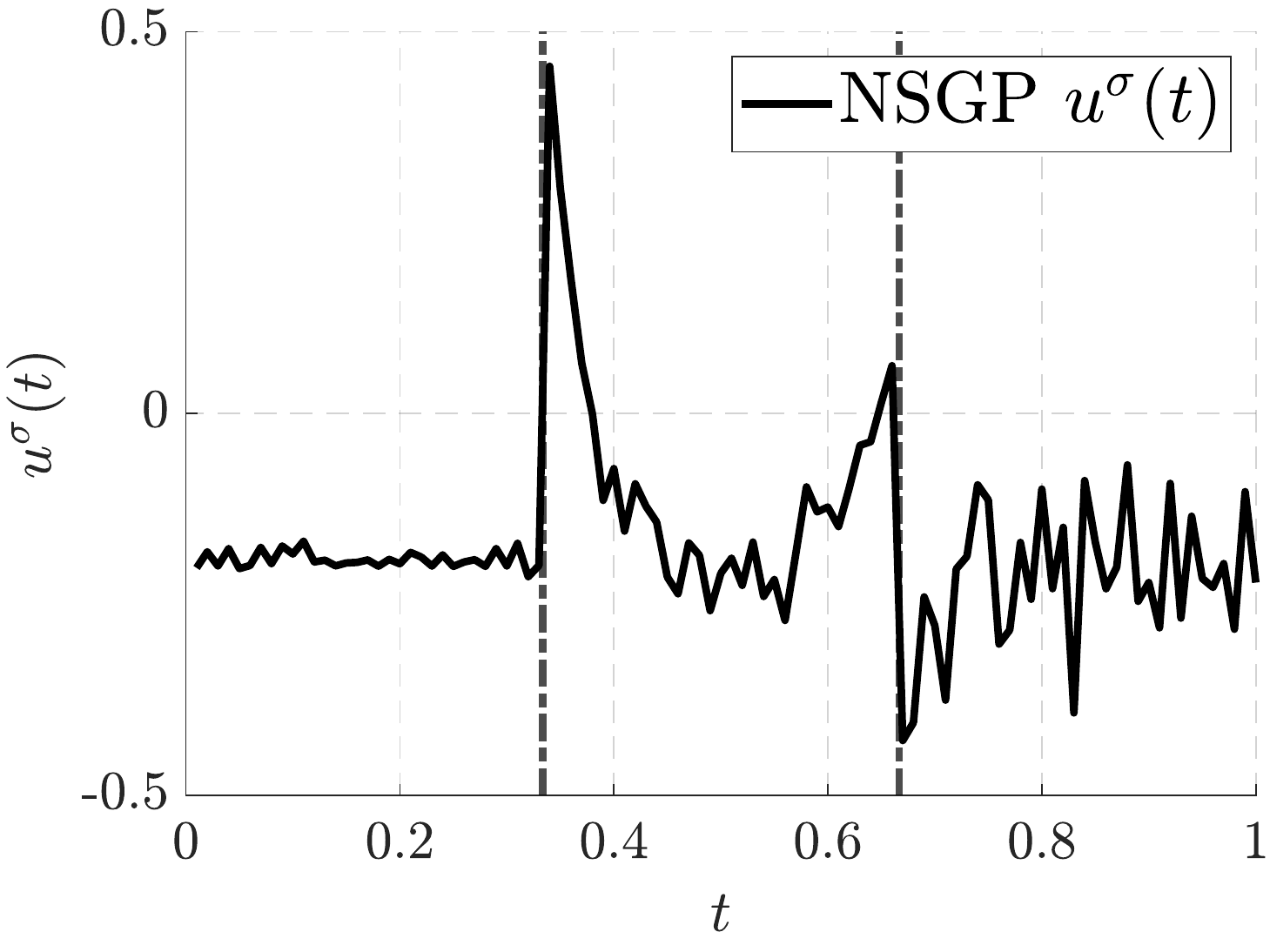}
	\caption{Demonstration of NSGP on the rectangular signal in Equation~\eqref{equ:exp-rect}. The shaded area stands for 0.95 confidence.}
	\label{fig:rect-nsgp}
\end{figure*}
\begin{figure*}[t!]
	\centering
	\includegraphics[width=.245\linewidth]{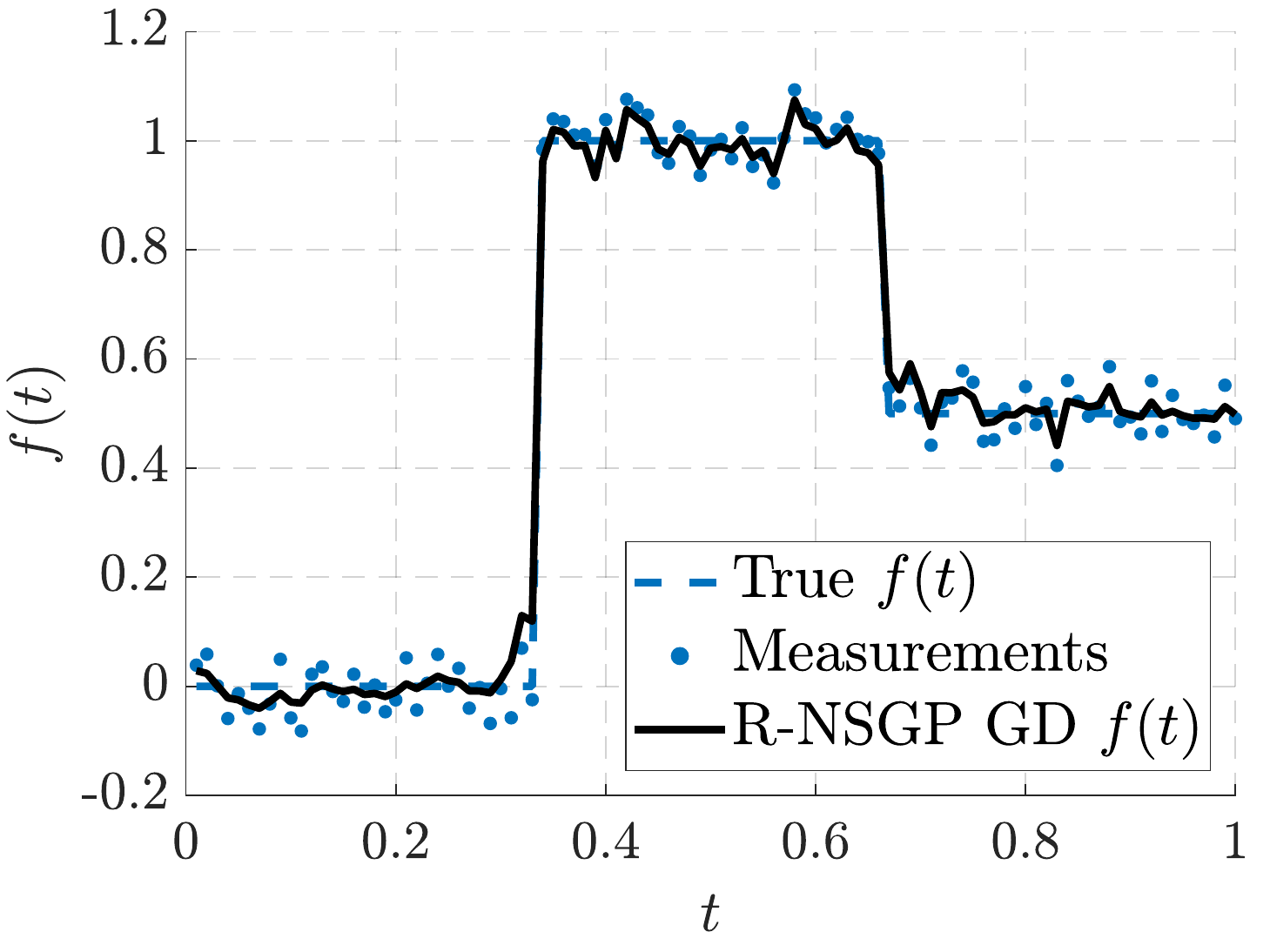}
	\includegraphics[width=.245\linewidth]{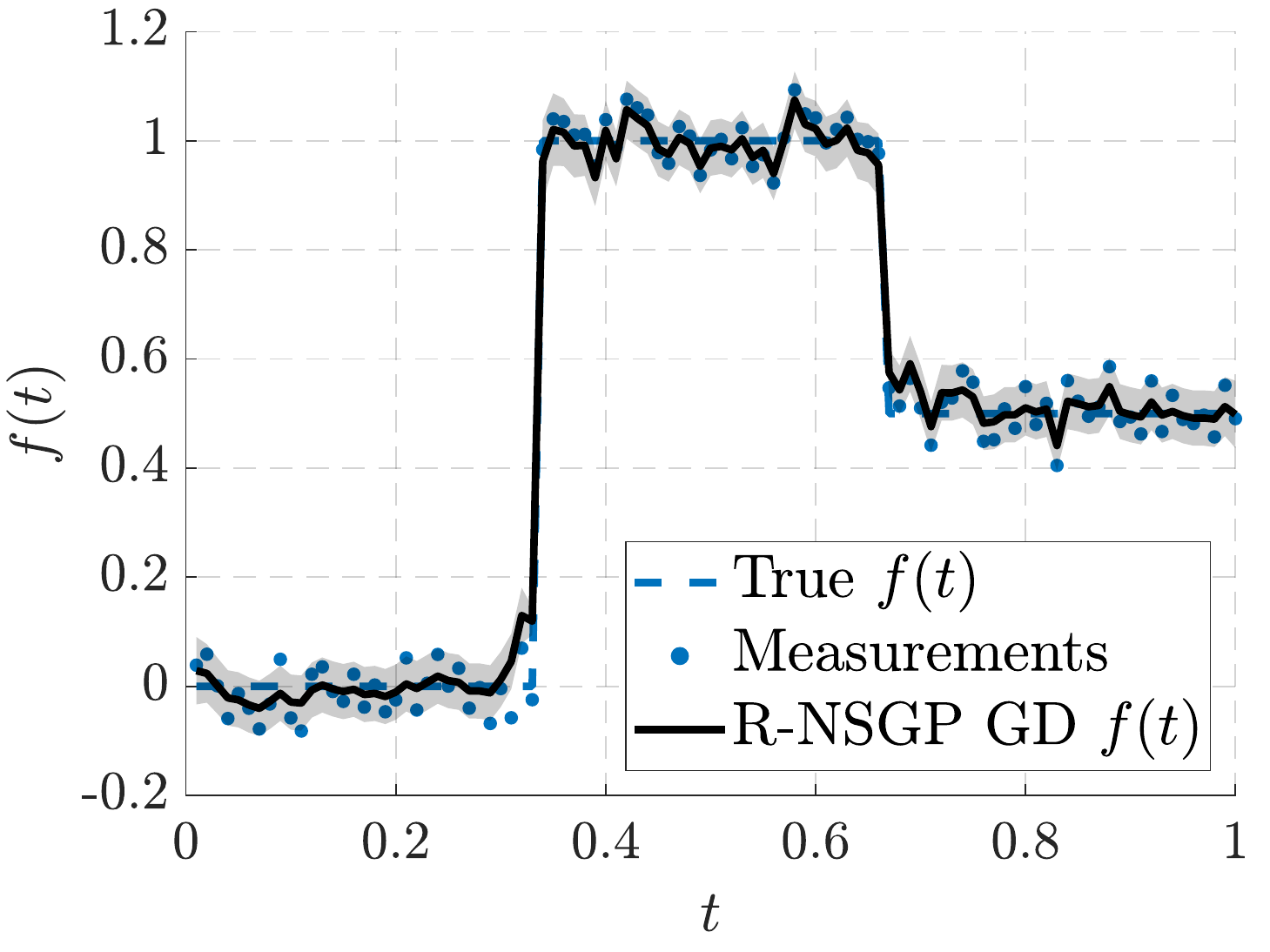}
	\includegraphics[width=.245\linewidth]{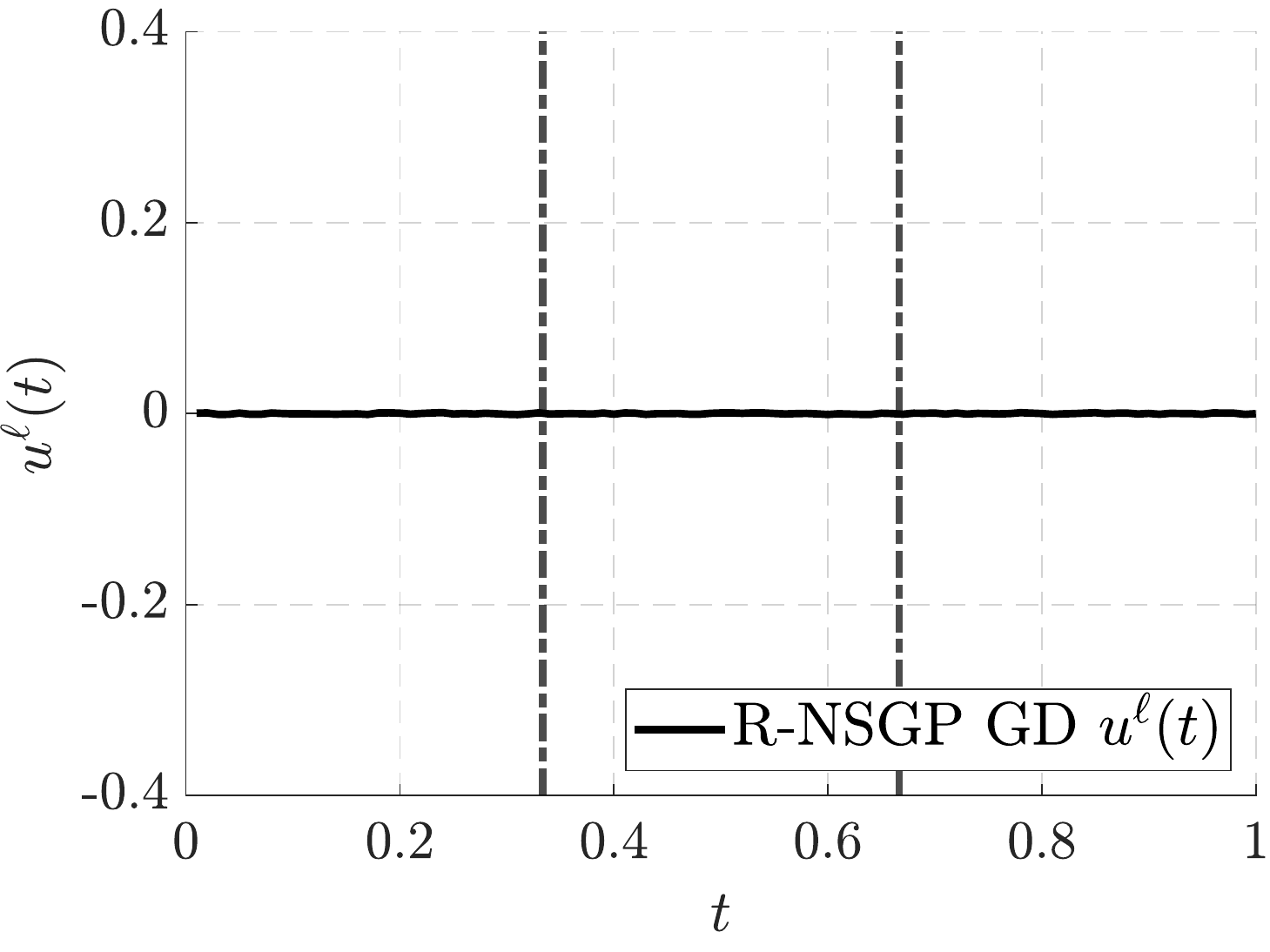}
	\includegraphics[width=.245\linewidth]{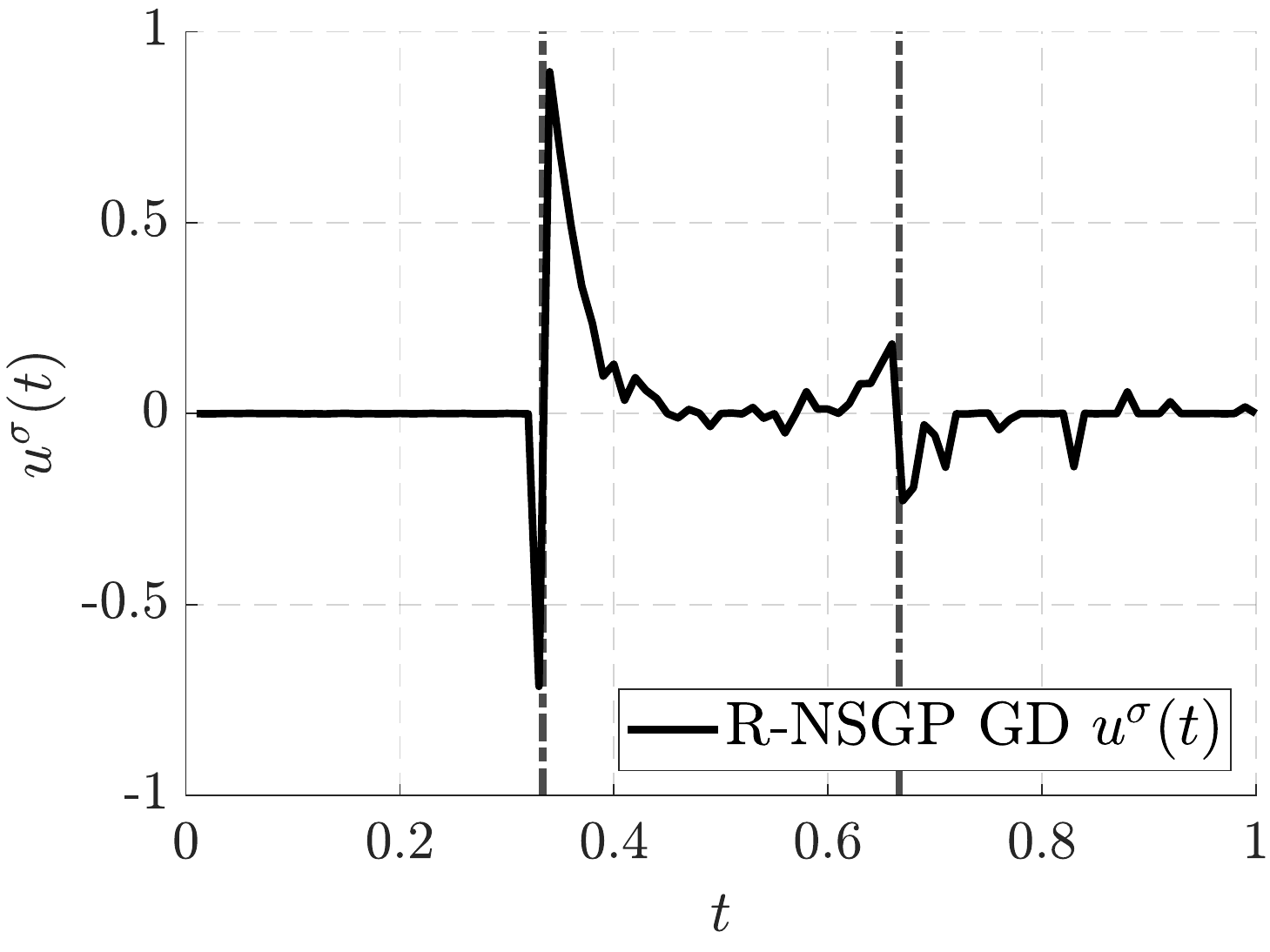}\\
	\includegraphics[width=.245\linewidth]{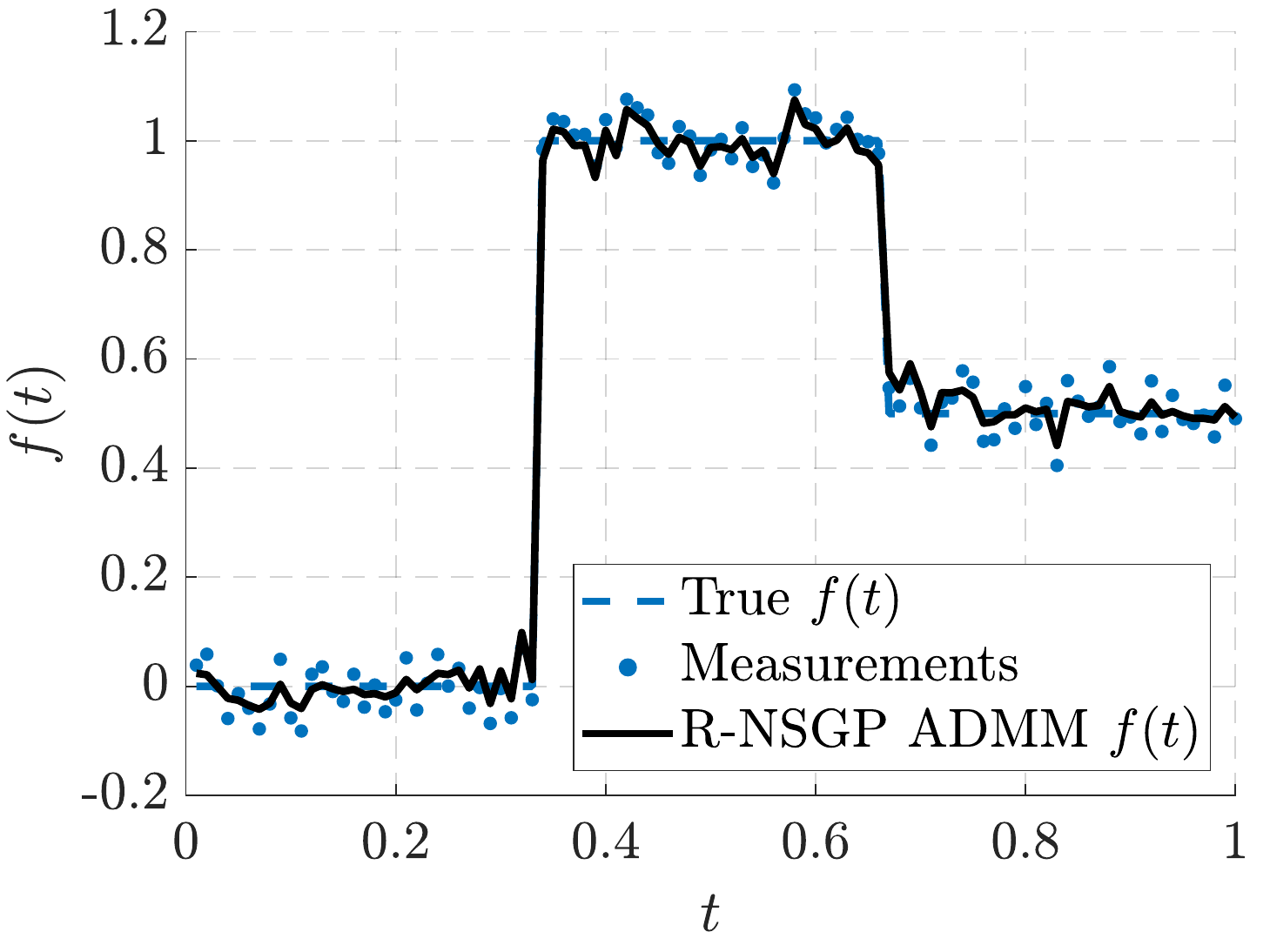}
	\includegraphics[width=.245\linewidth]{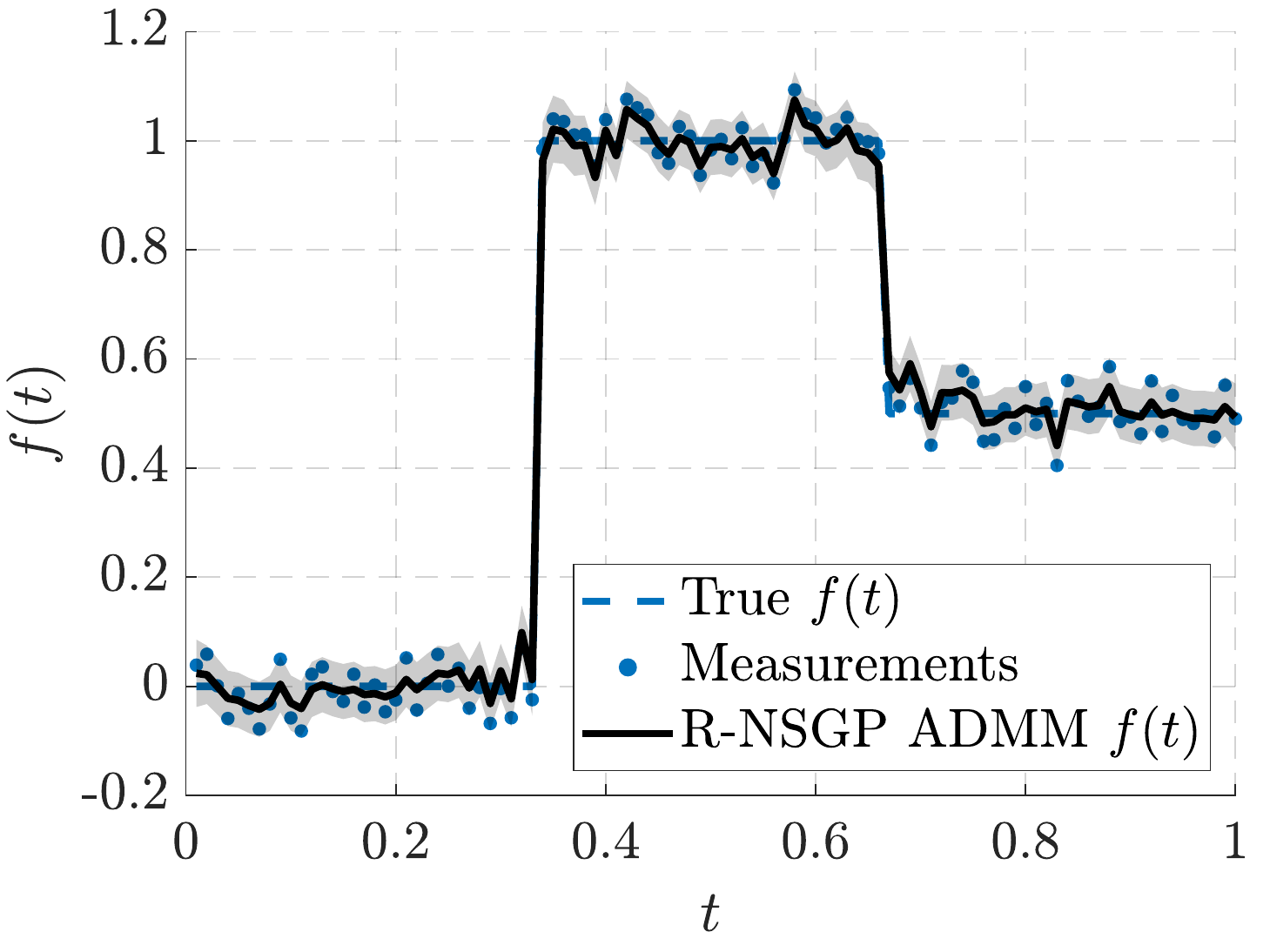}
	\includegraphics[width=.245\linewidth]{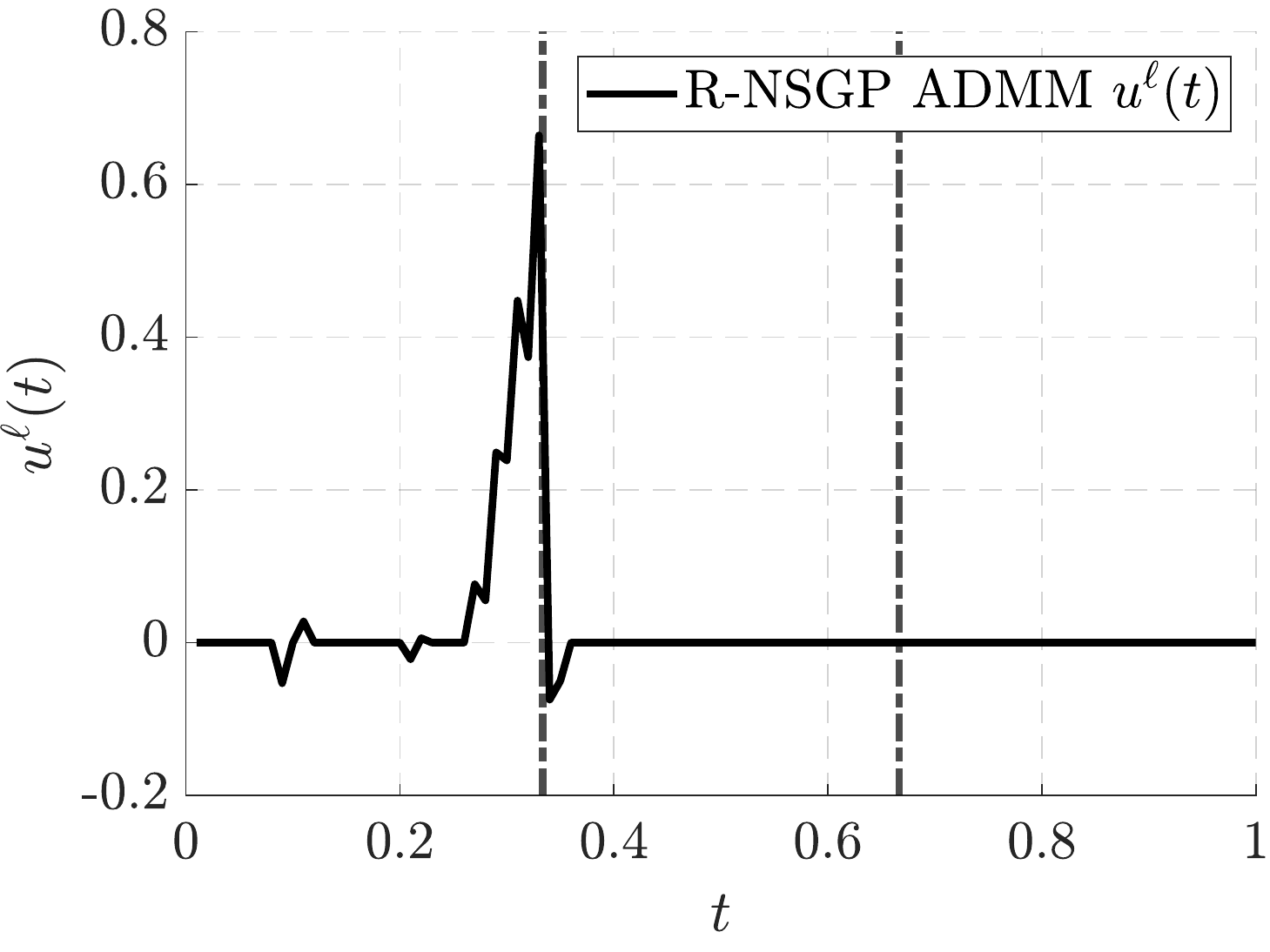}
	\includegraphics[width=.245\linewidth]{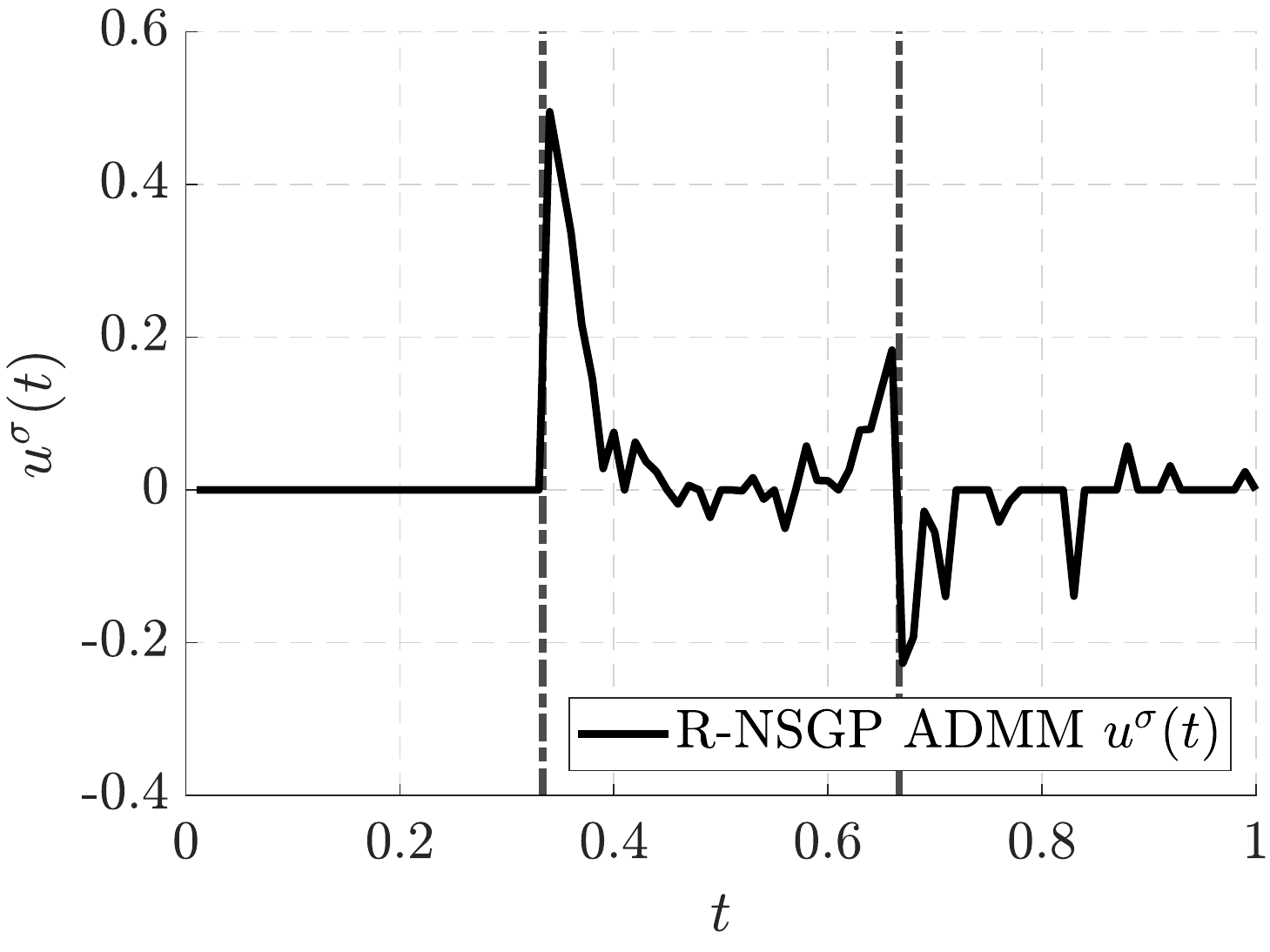}
	\caption{Demonstration of R-NSGP by using GD (first row) and ADMM (second row) on the rectangular signal in Equation~\eqref{equ:exp-rect}. The shaded area stands for 0.95 confidence.}
	\label{fig:rect-r-nsgp}
\end{figure*}
\begin{figure*}[t!]
	\centering
	\includegraphics[width=.245\linewidth]{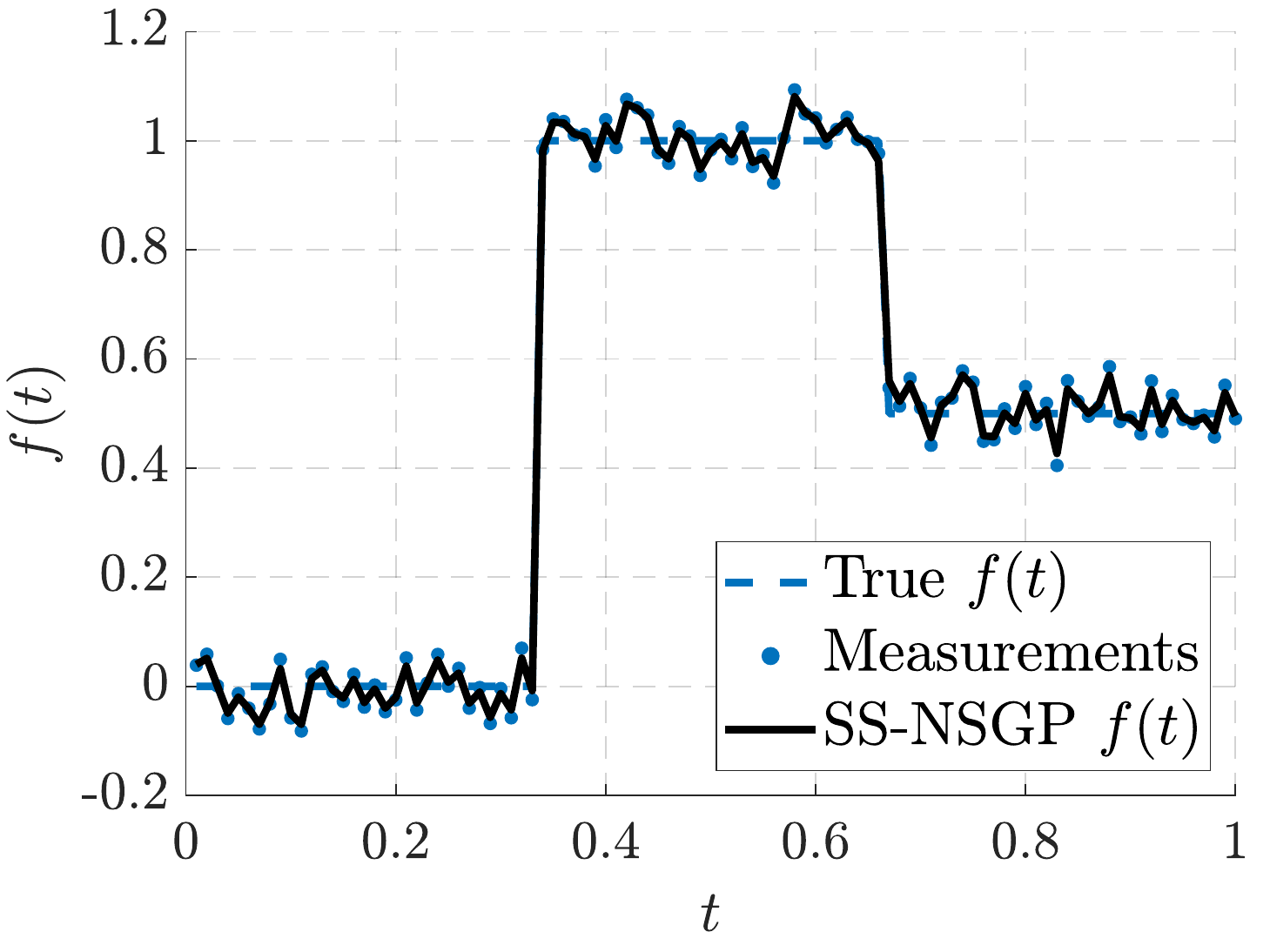}
	\includegraphics[width=.245\linewidth]{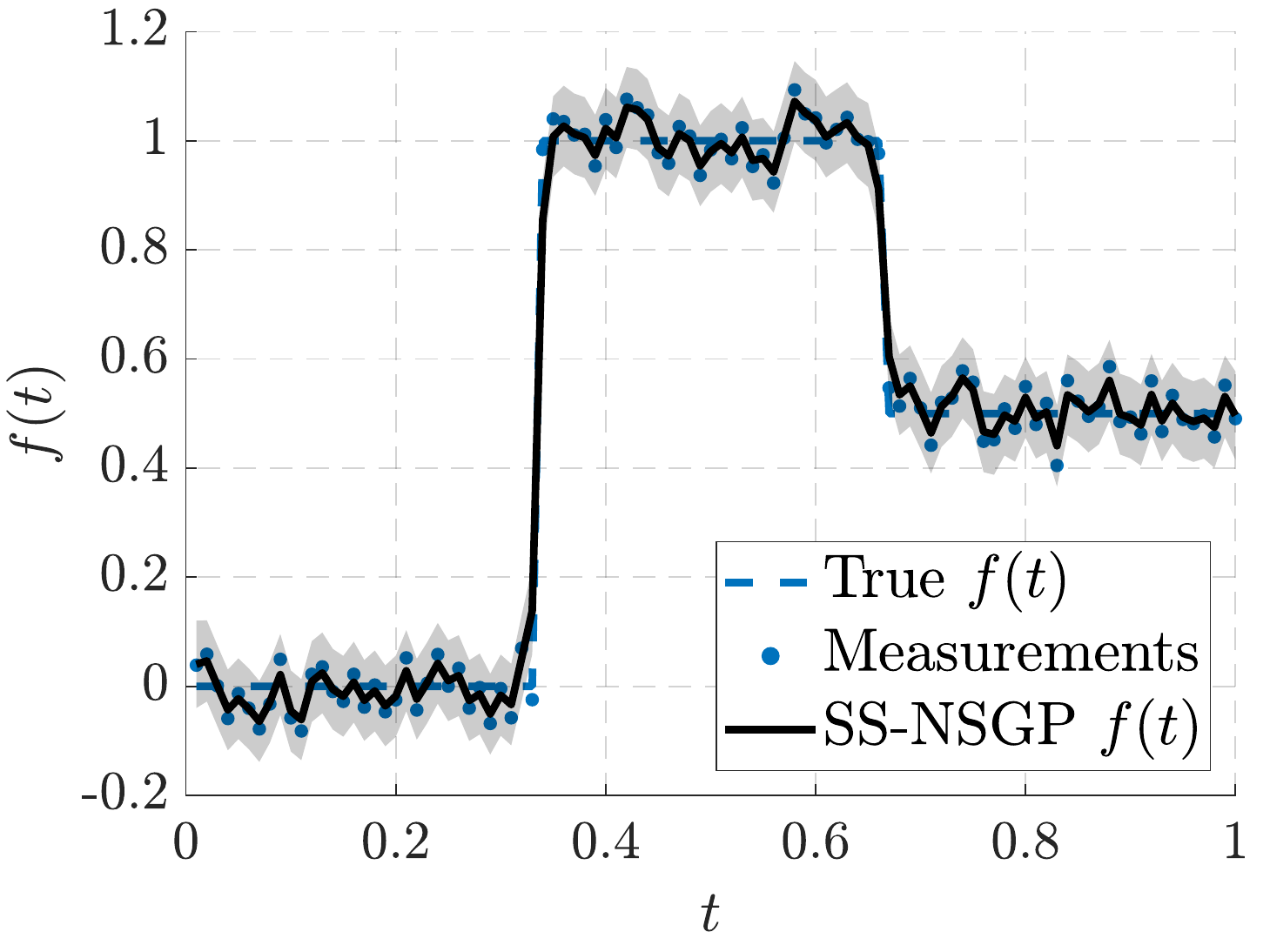}
	\includegraphics[width=.245\linewidth]{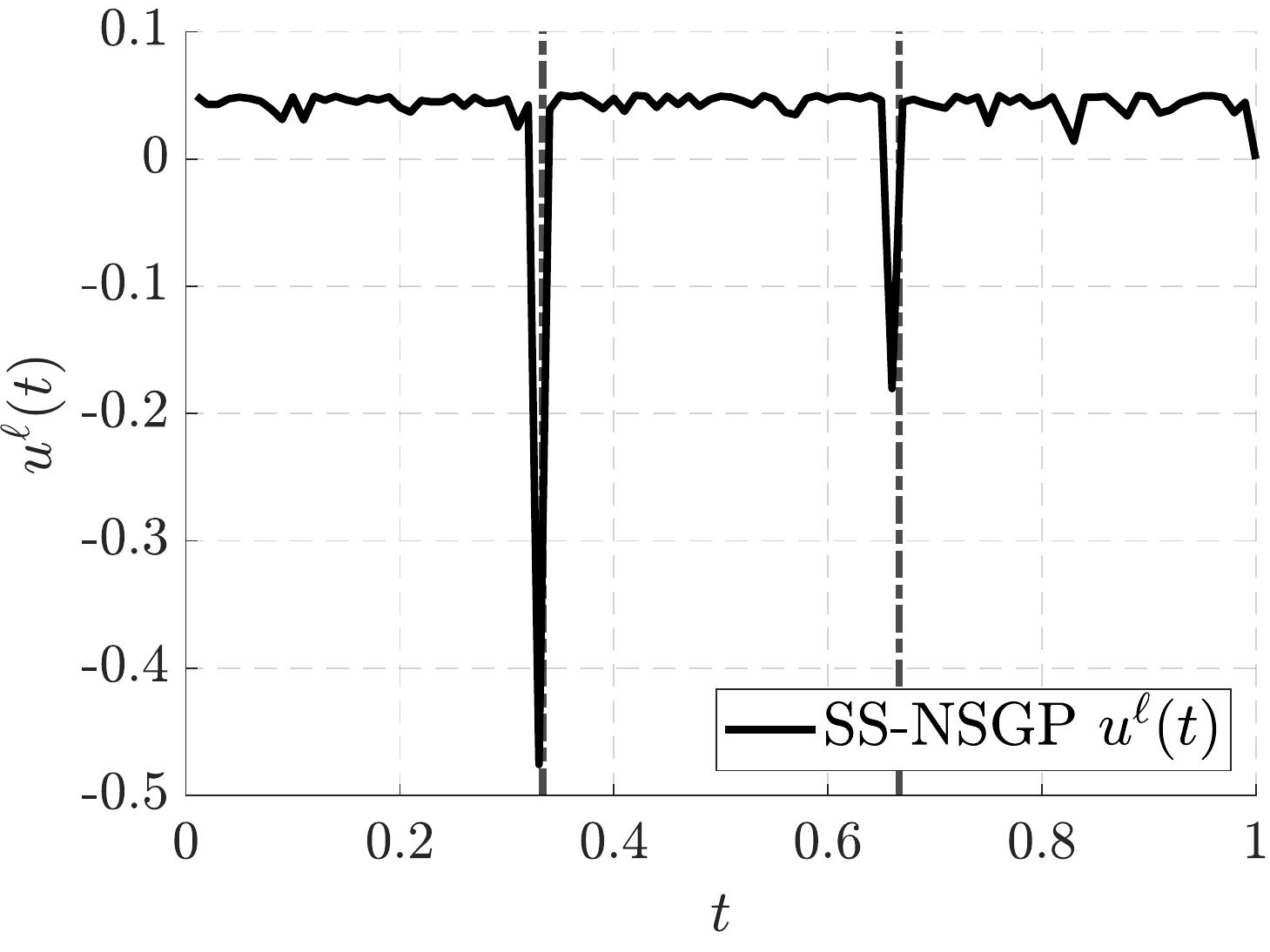}
	\includegraphics[width=.245\linewidth]{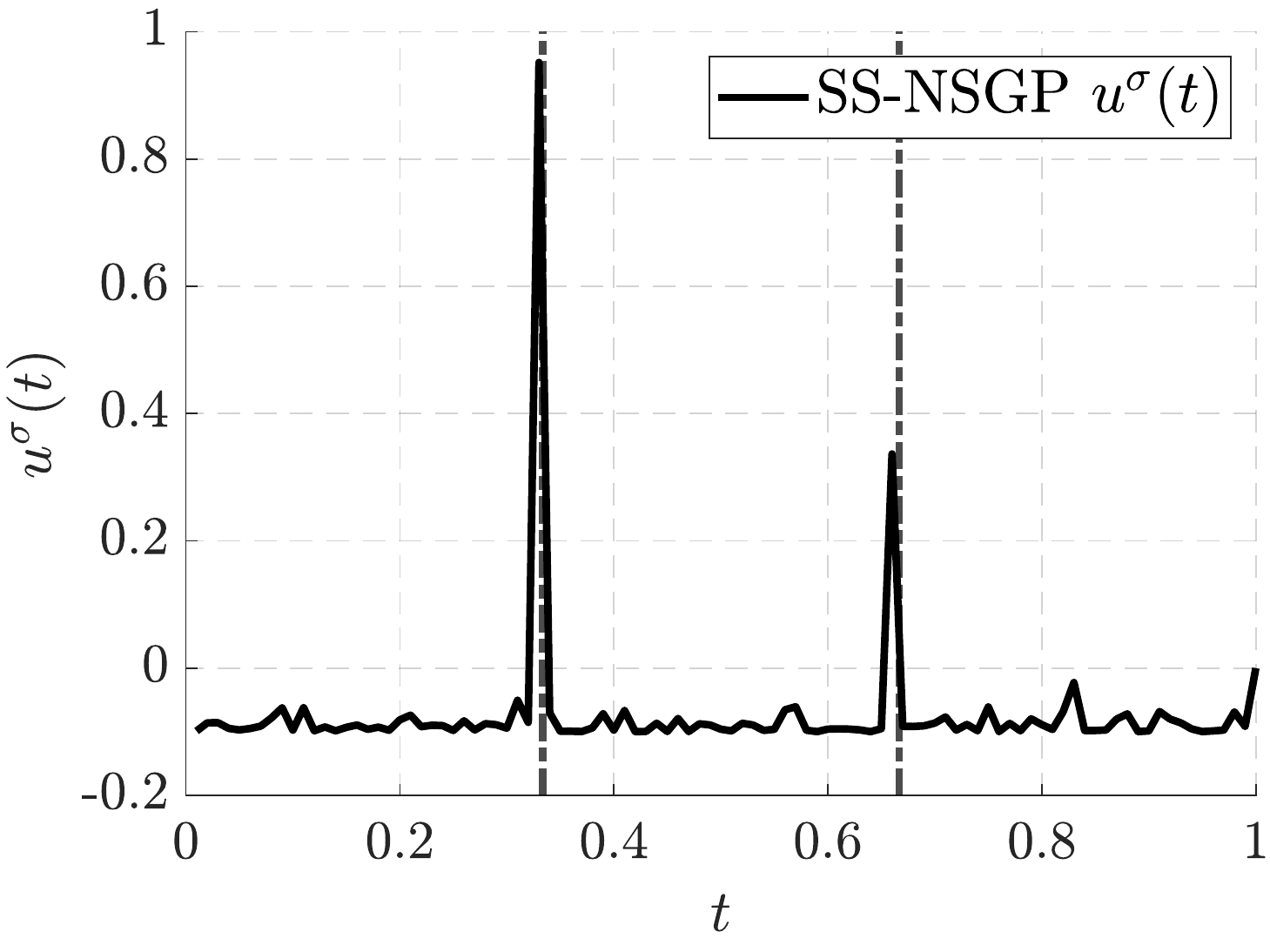}
	\caption{Demonstration of SS-NSGP on the rectangular signal in Equation~\eqref{equ:exp-rect}. The shaded area stands for 0.95 confidence.}
	\label{fig:rect-ss-nsgp}
\end{figure*}
\begin{figure*}[t!]
	\centering
	\includegraphics[width=.245\linewidth]{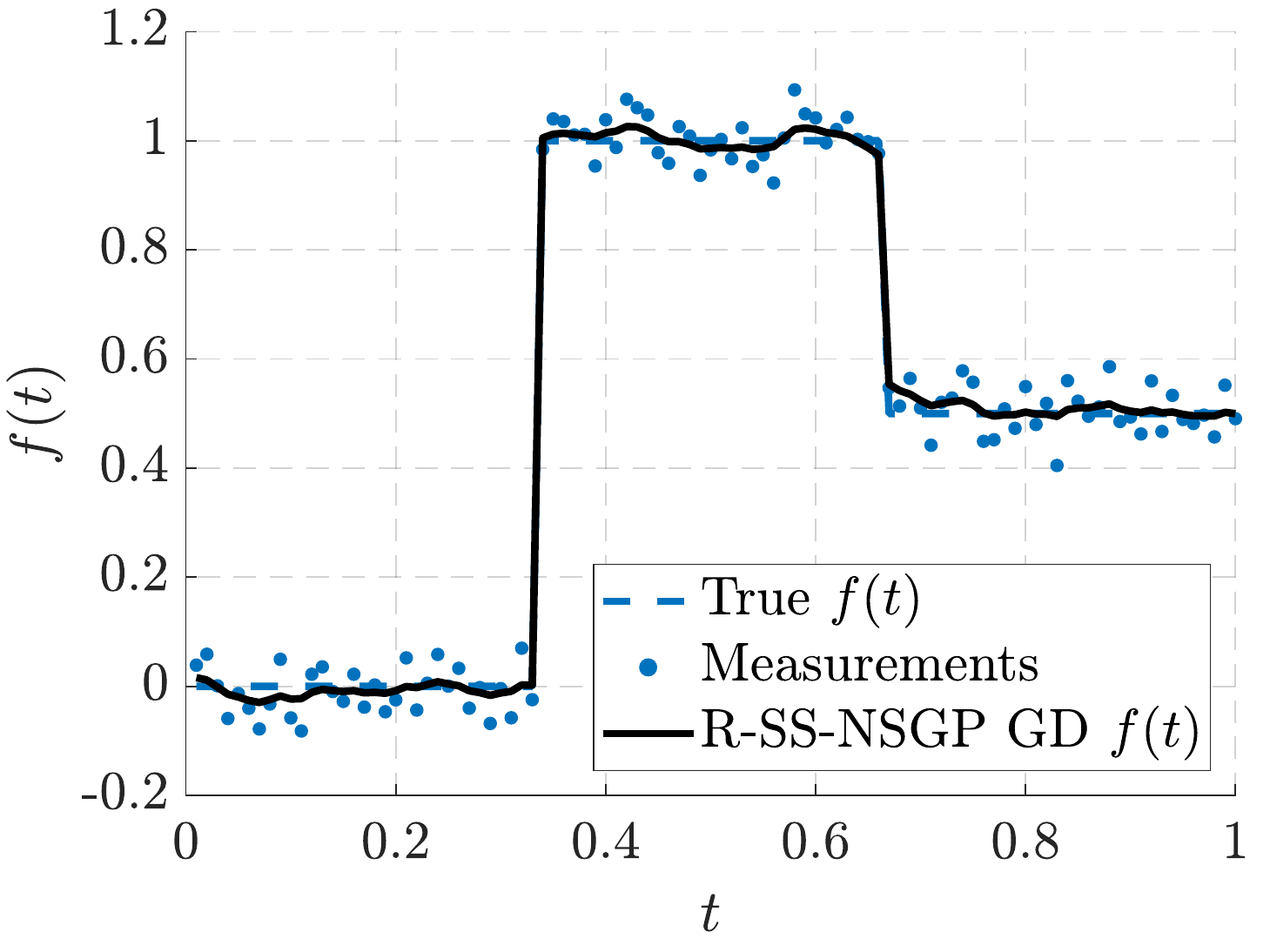}
	\includegraphics[width=.245\linewidth]{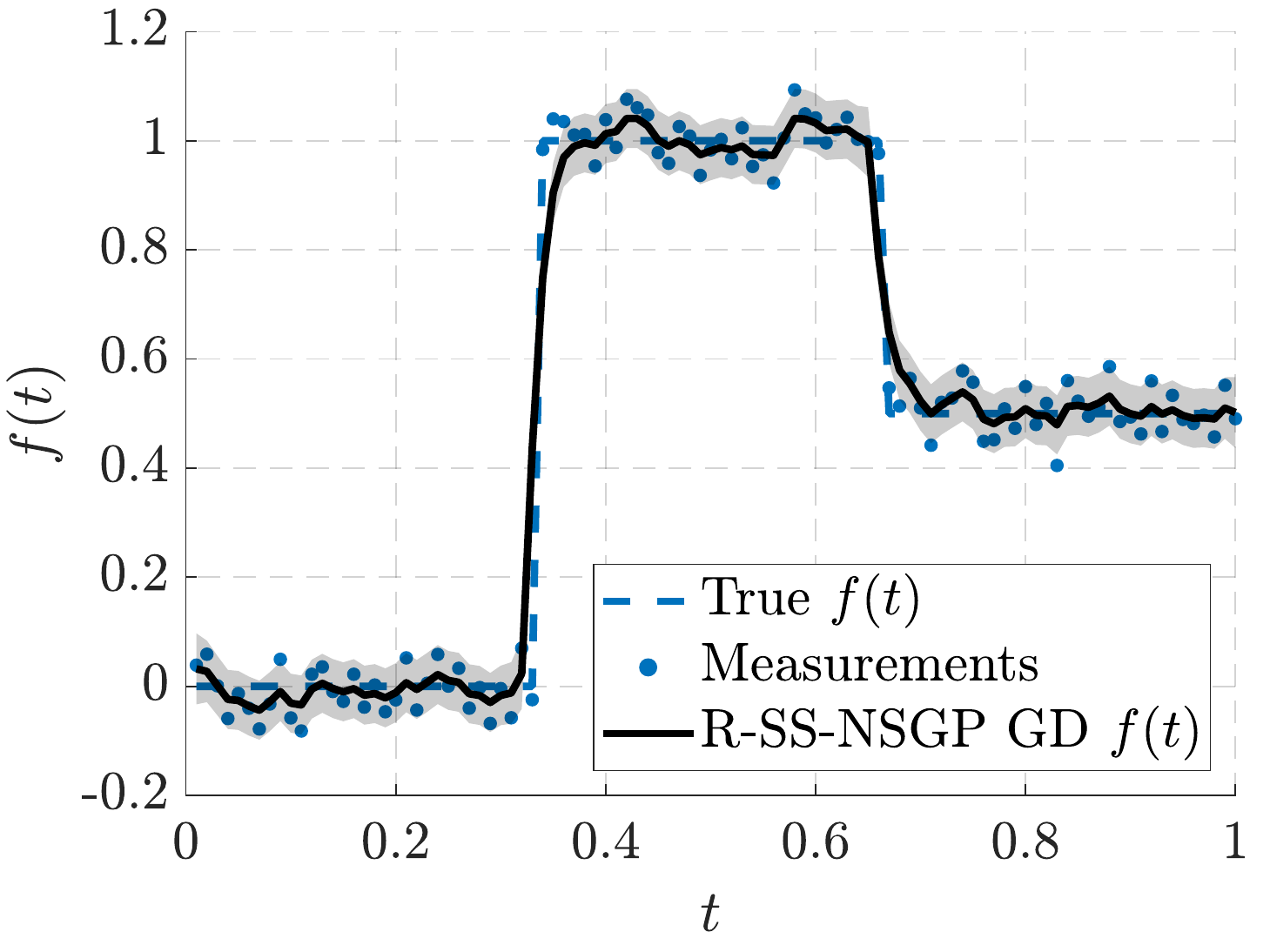}
	\includegraphics[width=.245\linewidth]{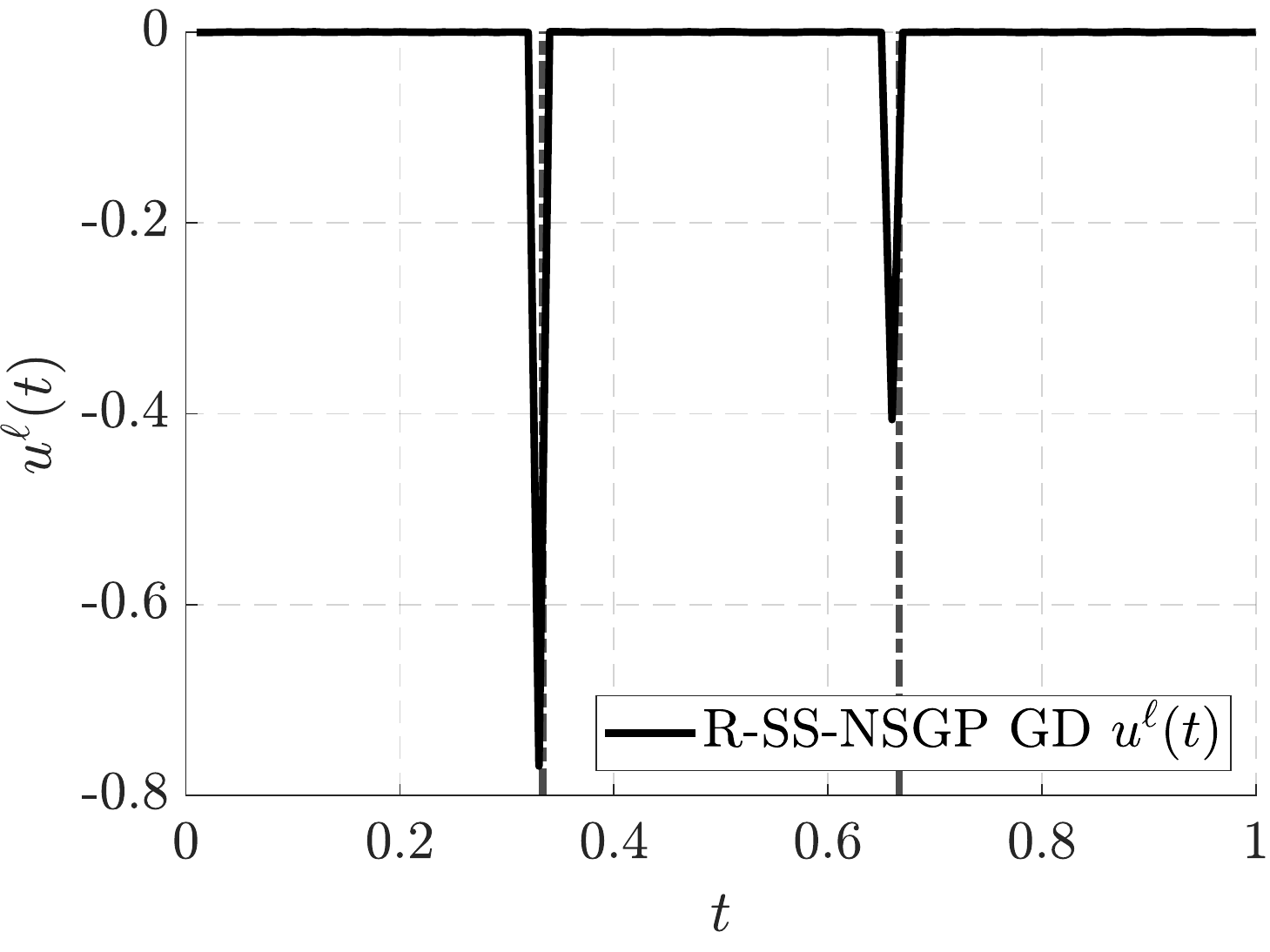}
	\includegraphics[width=.245\linewidth]{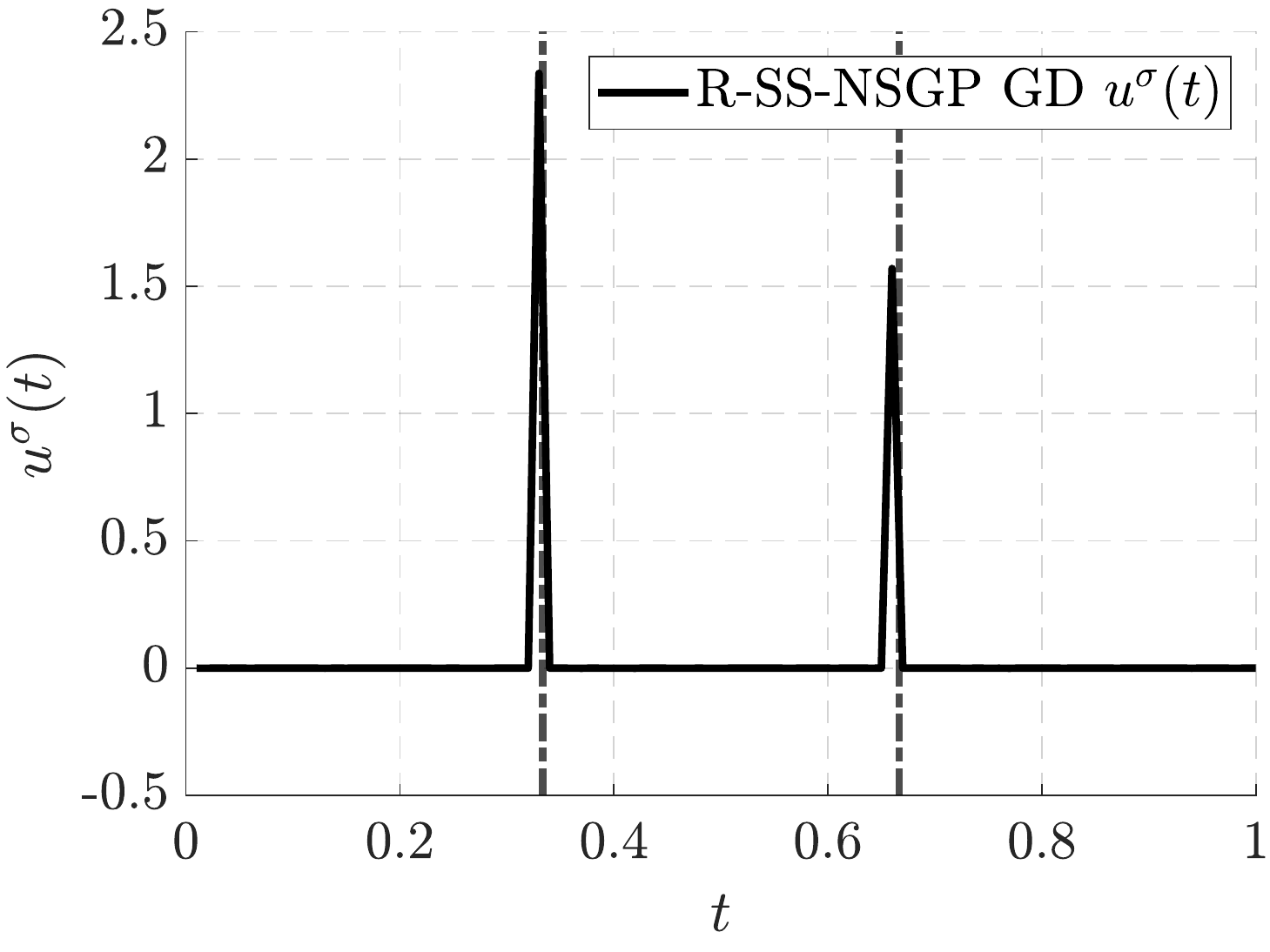}\\
	\includegraphics[width=.245\linewidth]{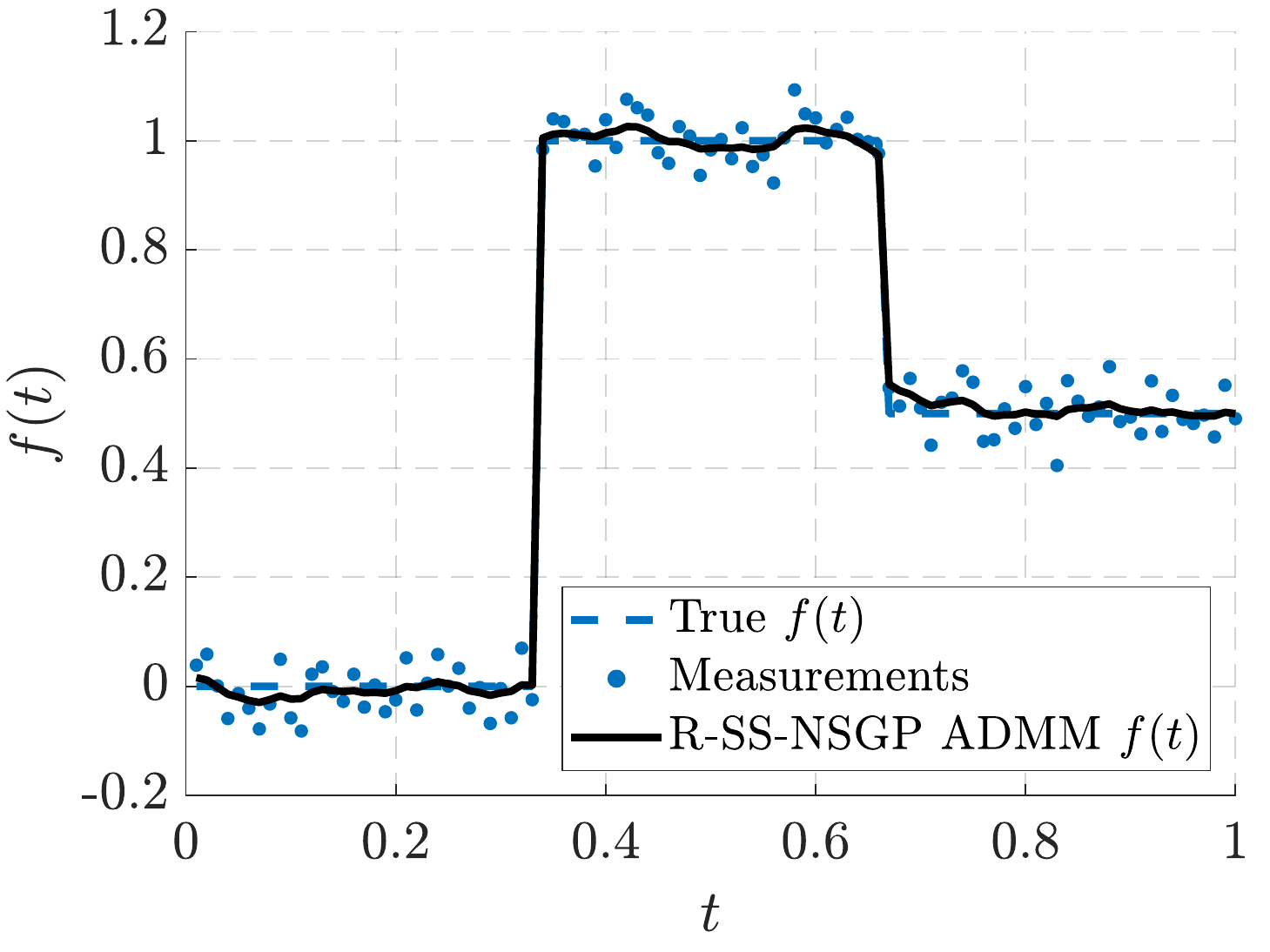}
	\includegraphics[width=.245\linewidth]{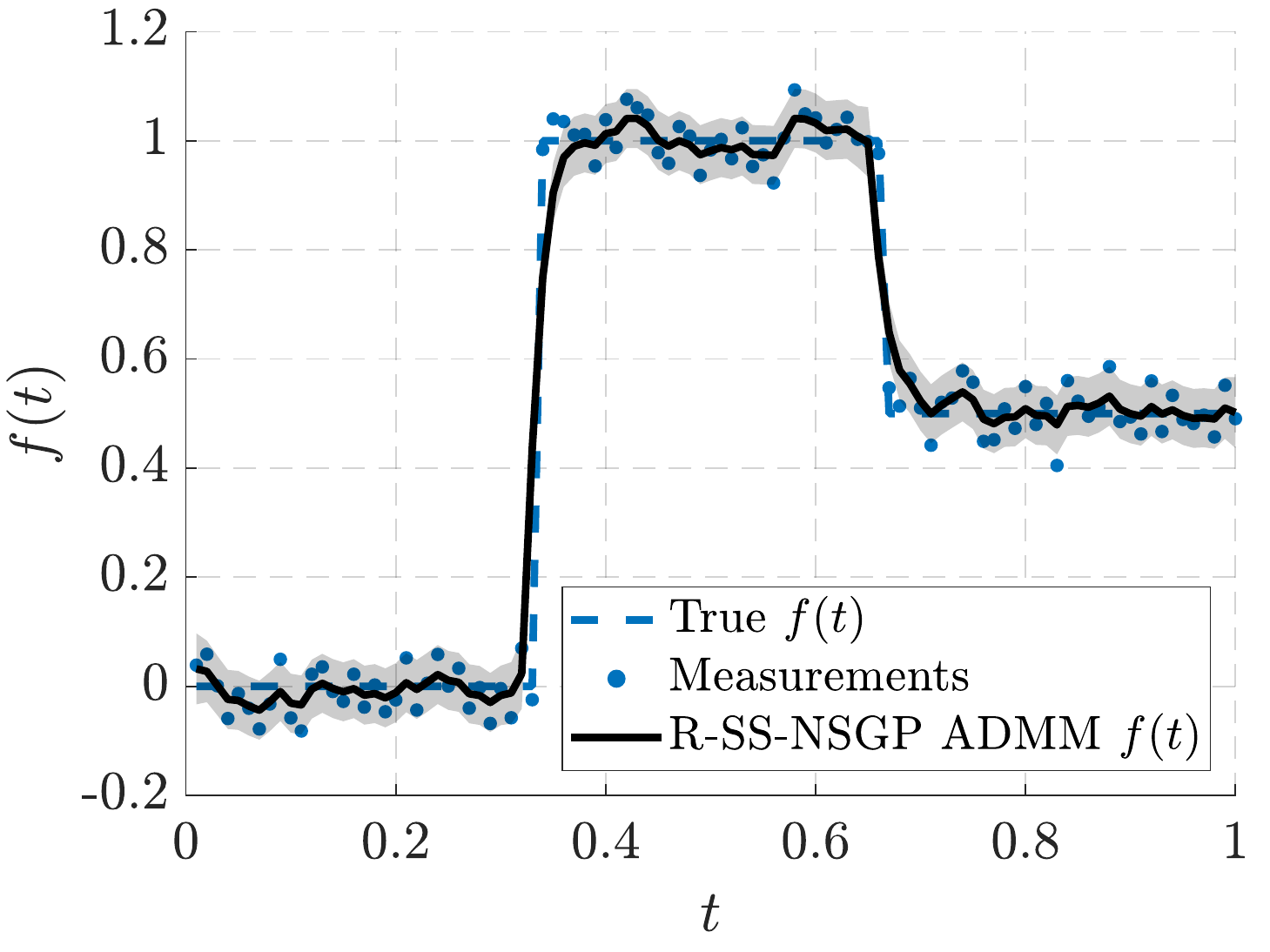}
	\includegraphics[width=.245\linewidth]{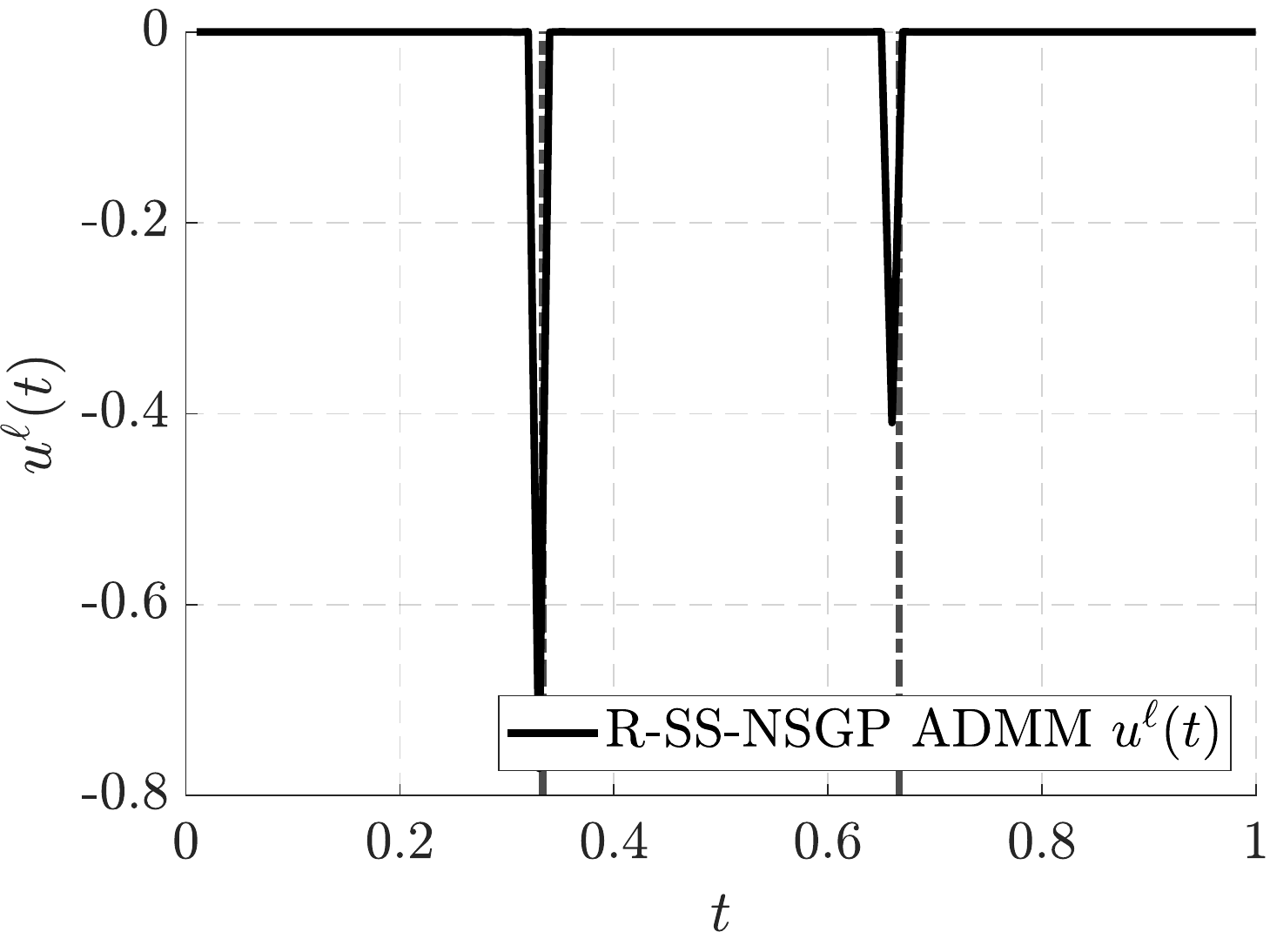}
	\includegraphics[width=.245\linewidth]{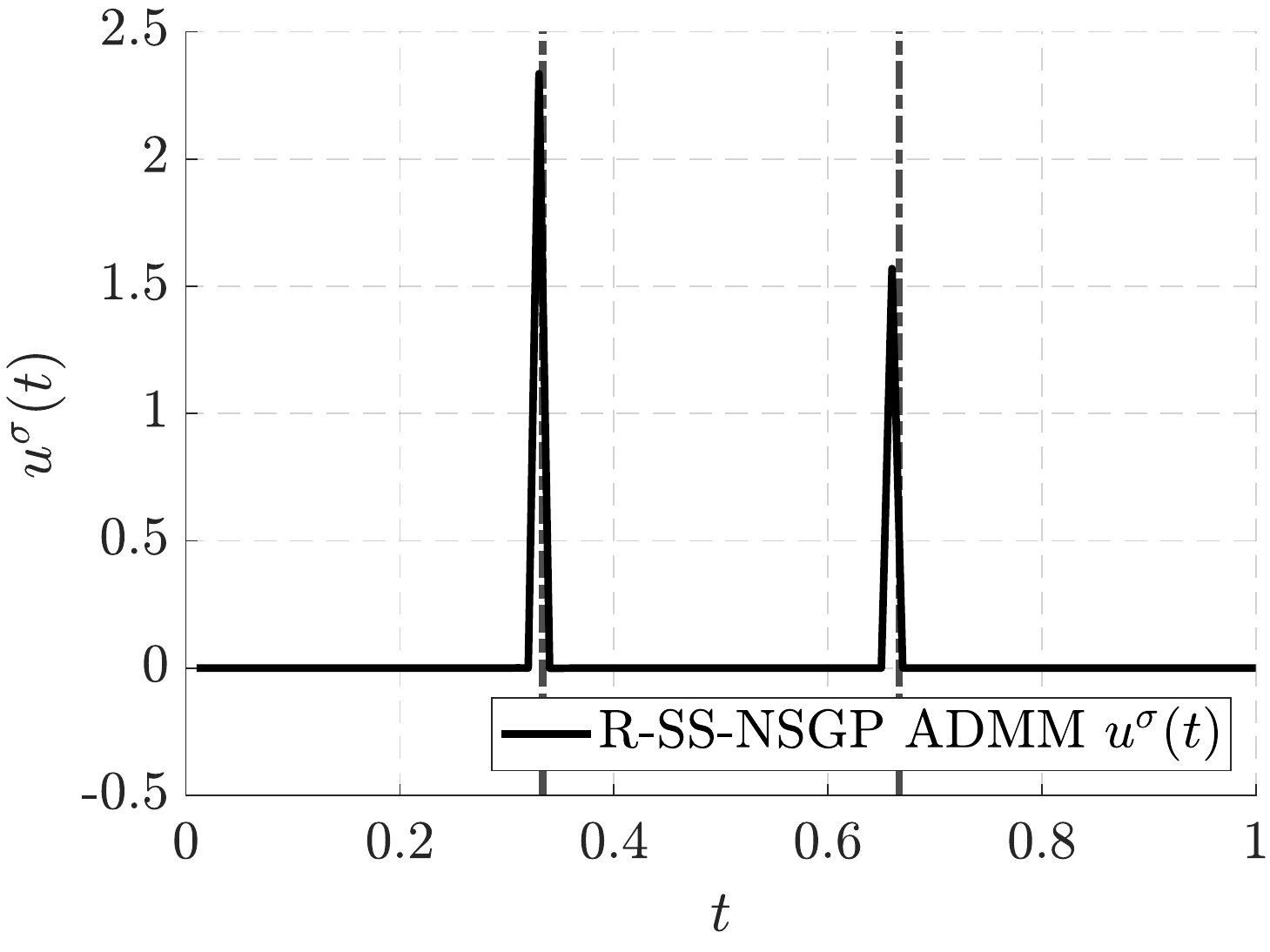}
	\caption{Demonstration of R-SS-NSGP by using GD (first row) and ADMM (second row) on the rectangular signal in Equation~\eqref{equ:exp-rect}. The shaded area stands for 0.95 confidence.}
	\label{fig:rect-r-ss-nsgp}
\end{figure*}

We use the Mat\'{e}rn $\nu=1/2$ covariance functions for the GP, SGP, NSGP, R-NSGP, and R-SS-NSGPs. The construction of the SDEs for SS-NSGPs are given in Appendix~\ref{appendix:sde-matern12}. Hyperparameters of GP are learnt by maximum likelihood estimation with the limited-memory Broyden--Fletcher--Goldfarb--Shanno (L-BFGS) algorithm. The hyperpameters of NSGPs are fixed (i.e., length-scale and magnitude of $u^\ell(t)$ and $u^\sigma(t)$ are $0.01$ and $3$, respectively). The R-NSGP and R-SS-NSGP are solved by using the subgradient method (R-NSGP GD) and the ADMM method (R-NSGP ADMM). The $\rho$ parameters of R-NSGP ADMM are chosen to be large values because of Assumption~\ref{assu:rho-large}, and they are $\rho_\ell = \rho_\sigma = 150$ for the R-NSGPs and $\rho_\ell = \rho_\sigma = 50$ for the R-SS-NSGPs.

The root mean square error (RMSE) and the negative log predictive density (NLPD) metrics are employed to measure the performance. The RMSE is computed with respect to true $f(t)$, and the NLPD is computed by evaluating the predictive density on the test data. Note that the MAP-based approaches do not give the predictive densities needed by NLPD metrics. These two quantities are averaged by conducting 100 independent Monte Carlo (MC) runs. For each MC run we generate a set of measurements of Equation~\eqref{equ:exp-rect} on the uniform grid with $k=1,2,\ldots, 100$. We use the same random seed in all the visualizations. 

The numerical results are reported in Table~\ref{tbl:rect-result}. Let us first focus on the ``without uncertainty'' column. We see that the R-SS-NSGP ADMM method gives the best RMSE 1.63, and SGP FIC gives the worst. The NSGP is worse than GP and SGP DTC, but the regularized NSGPs achieve better results than GP, SGPs, and NSGP. When the regularization is not used, the SS-NSGP is better than GP and NSGP. Moreover, the regularization in the state-space construction yields a substantial performance boost in terms of RMSE. For both the R-NSGP and R-SS-NSGP, the ADMM solver is shown to be slightly better than the subgradient method. 

In the ``with uncertainty'' column of Table~\ref{tbl:rect-result}, the R-NSGP ADMM achieves the best results in terms of RMSE and NLPD. We observe that the RMSEs of NSGP and R-NSGPs have almost no difference compared to the results without uncertainty quantification. It means that the marginal uncertainty quantification method could preserve the original results, while at the same time approximating the marginal posterior distribution to a good extent. The RMSE and NLPD results are improved by involving the regularization for NSGP. However, the state-space NSGPs give poor RMSEs and NLPDs when the uncertainty is quantified. The RMSEs of R-SS-NSGP are significantly increased compared to the ones without uncertainty quantification. The NLPDs of R-SS-NSGP are also the worst among all entries.

Figure~\ref{fig:rect-gp-sgp} plots the results of GP, SGP FIC, and SGP DTC. We find that the estimate of GP shows signs of over-fitting to the measurements because the optimized length-scale and magnitude parameters of GP converge to some small values (which are $0.53$ and $0.56$, respectively in this MC run). The mean estimates of GP are not flat enough to model the smooth part of the test signal. The SGP FIC fails to model the jump at $t=1/3$ properly, and we see that the learnt inducing points tend to not aggregate around this jump point. The inducing points of SGP DTC, however, appear to be more uniformly distributed and experience fewer impacts from discontinuities. In terms of the mean estimate, the SGP DTC method appears to be the best among these three methods because the flatness of the signal is better represented. However, the posterior covariances of SGP FIC and DTC look unreasonably large and jittered. By examining the RMSE and NLPD in Table~\ref{tbl:rect-result} we conclude the same result that SGP DTC has a better mean estimate but not better uncertainty quantification. 

Figure~\ref{fig:rect-nsgp} shows the results of NSGP. Compared to the GP in Figure~\ref{fig:rect-gp-sgp}, the NSGP estimates are shown to be slightly more over-fitted to the measurements from visual inspection. The numerical results in Table~\ref{tbl:rect-result} also conclude the same. We find that the estimated length-scales and magnitudes appear to respond to the signal jumps. In particular, $u^\ell(t)$ shows a significant decrement at $t=1/3$, and $u^\sigma(t)$ shows increment at $t=1/3$ and $t=2/3$. 

Figure~\ref{fig:rect-r-nsgp} illustrates the results of R-NSGP by using the subgradient descent and ADMM methods. Compared to the aforesaid GP, SGP, and NSGP, the R-NSGP shows a better fit to the signal. Also, the parameters $u^\ell(t)$ and $u^\sigma(t)$ are sparse. In particular, the $f(t)$ estimates in $t\in[0, 1/3]$ and $t\in[2/3, 1]$ are shown to be flat, as well as the parameter estimates. The differences between R-NSGP GD and ADMM are found to be subtle in terms of $f(t)$ estimates, but not in $u^\ell(t)$ and $u^\sigma(t)$. As shown on the second row of Figure~\ref{fig:rect-r-nsgp}, the estimated $u^\ell(t)$ by R-NSGP GD is almost zero and does not react to the signal jumps. The estimated $u^\sigma(t)$ however shows significant drop and rise at $t=1/3$ while ADMM does not give the drop prior to $t=1/3$. In contrast to the all-zero $u^\ell(t)$ estimated by GD, the ADMM gives a substantial $u^\ell(t)$ growth at $t=1/3$ while still preserving the sparsity. 

\begin{figure*}[t!]
	\centering
	\includegraphics[width=.75\linewidth]{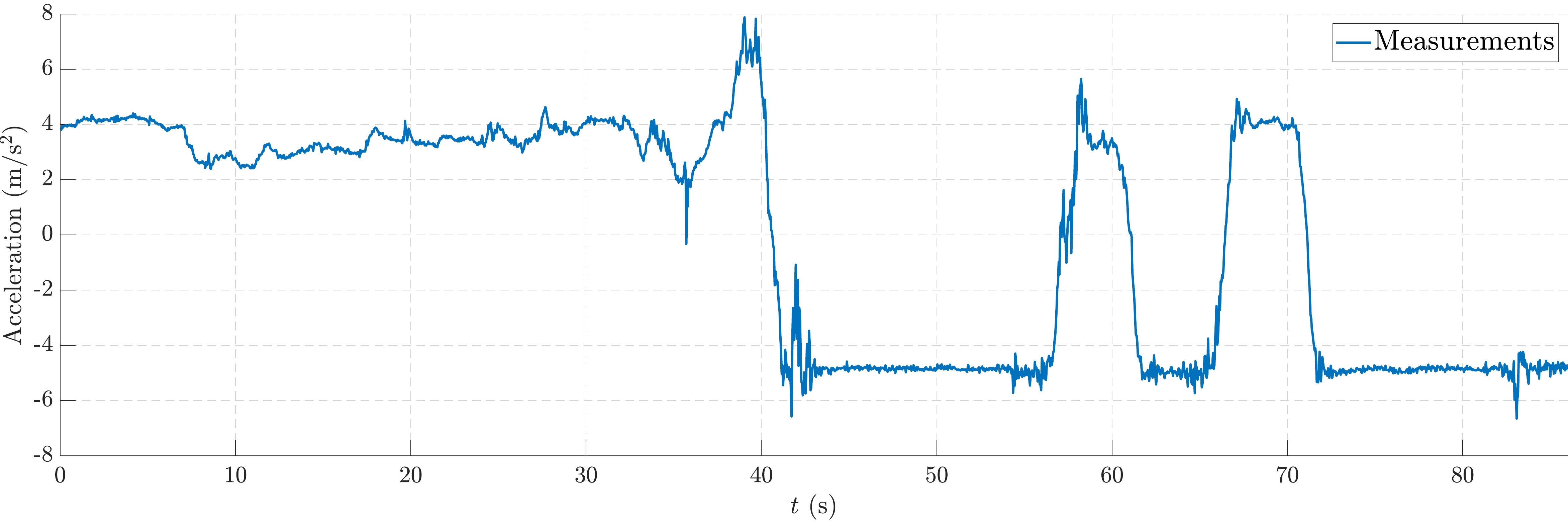}
	\includegraphics[width=.75\linewidth]{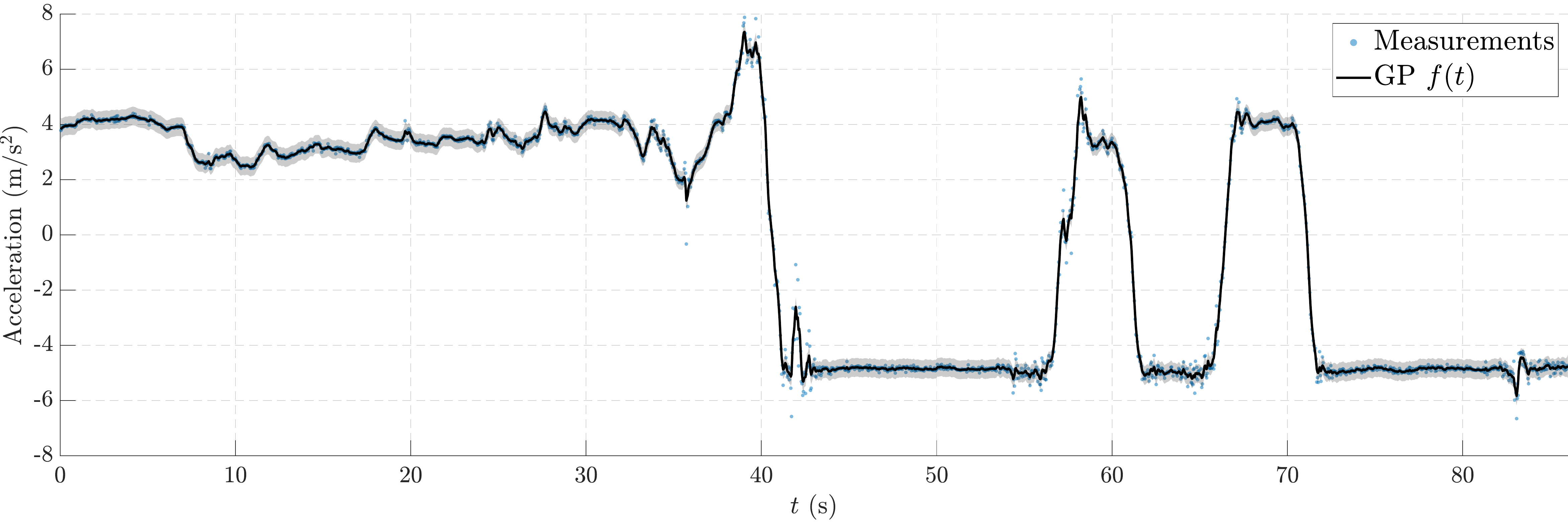}
	\includegraphics[width=.75\linewidth]{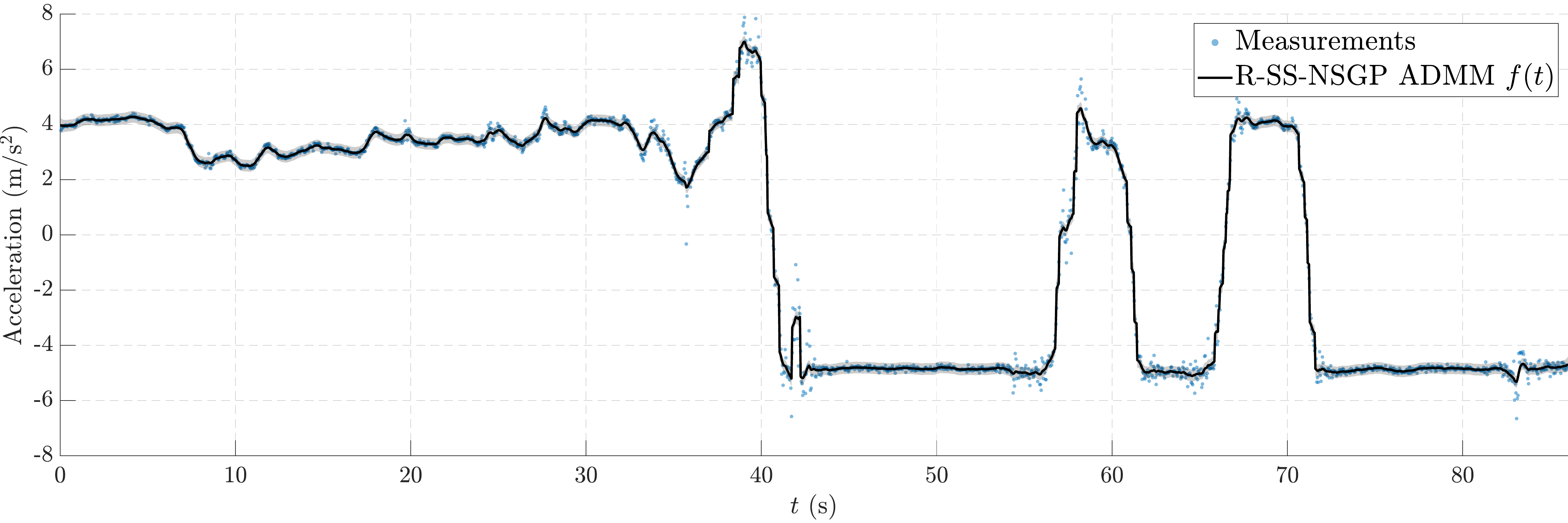}
	\includegraphics[width=.75\linewidth]{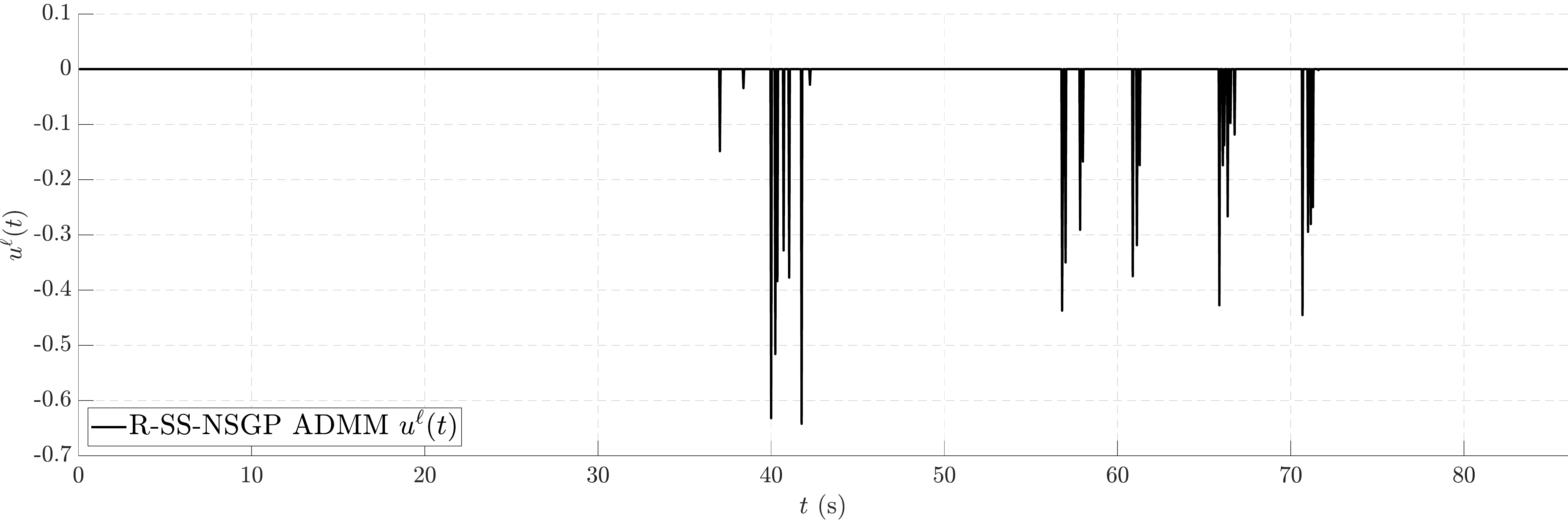}
	\includegraphics[width=.75\linewidth]{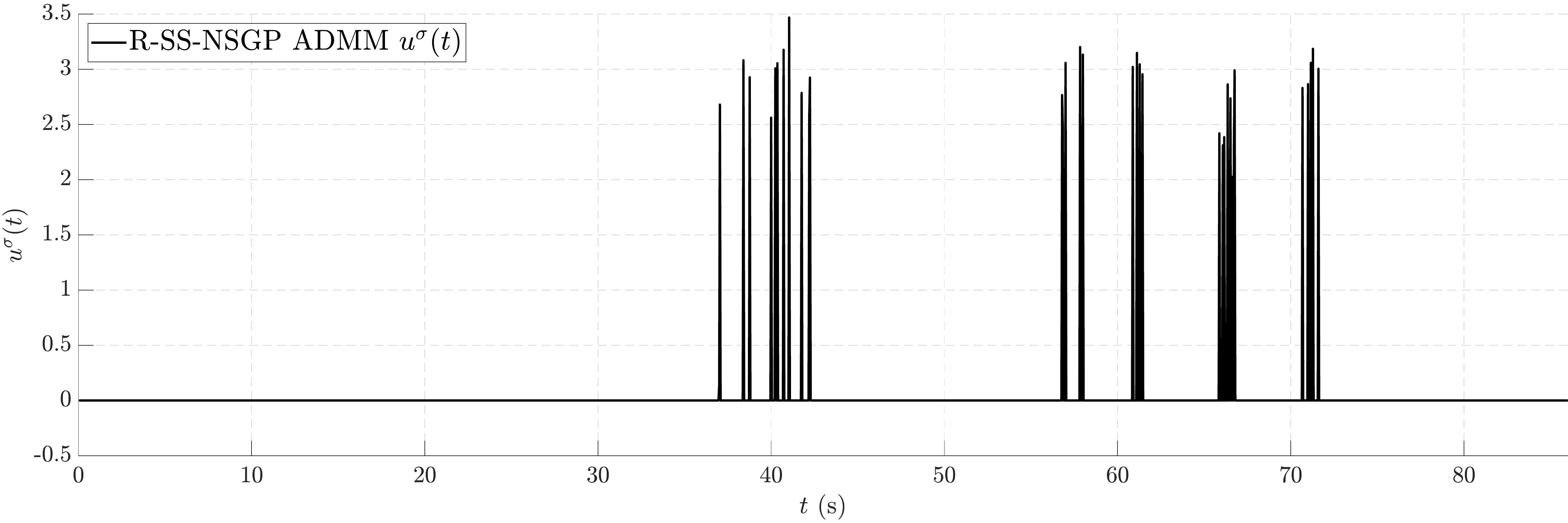}
	\caption{Modeling human motion with R-NSGPs. The shaded area stands for 0.95 confidence.}
	\label{fig:imu}
\end{figure*}

Figure~\ref{fig:rect-ss-nsgp} demonstrates the results of SS-NSGP. By inspecting the first row of the figure we find the same over-fitting problem as GP and NSGP do. However it is worth mentioning that the estimated $u^\ell(t)$ and $u^\sigma(t)$ show more meaningful results compared to NSGP. Notably, $u^\ell(t)$ and $u^\sigma(t)$ give sharp descent and ascent, respectively at those discontinuity points, otherwise they are noisily flat. 

The results of R-SS-NSGP are shown in Figure~\ref{fig:rect-r-ss-nsgp}. In terms of $f(t)$ estimation, the R-SS-NSGP shows significantly better fitting compared to GP, SGP, NSGP, and R-NSGP. The estimated $f(t)$ is flat almost everywhere while is still capable of jumping properly at $t=1/3$ and $t=2/4$. Moreover, the magnitudes of parameter jumps resemble the signal jumps (e.g., $| u^\ell(1/3)| > |u^\ell(2/3)|$). However, $f(t)$ with uncertainty quantification (i.e., $p(f(t) \mid y_{1:T})$) appears to be poorly estimated, and we can see that the estimated $f(t)$ is overly smooth at $t=1/3$ and $t=2/3$. This result corresponds to the numerical results in Table~\ref{tbl:rect-result} where one can see a significant performance decrease in terms of RMSE and NLPD when the uncertainty estimation is enabled. As for the parameters, $u^\ell(t)$ and $u^\sigma(t)$ appear to be well regularized. The parameters are sparse compared to the unregularized SS-NGSP in Figure~\ref{fig:rect-ss-nsgp}.

\subsection{Modeling human motion}
\label{sec:exp-imu}
In this section, we apply the proposed R-NSGPs to model the human motion as measured by an accelerometer. The data was collected in the clinical trials organized by Helsinki University Central Hospital and Aalto University, where the subjects (healthy volunteers) were asked to perform a set of motion sequences. The data contain electrocardiogram and inertial measurements which were captured by using the MetaMotion sensor from MbientLab Inc. For a detailed explanation of this dataset, see~\citet{Roland2018}. In this experiment, we choose one example recording where the subject was lying on their back and rolling, and the measurements exhibit strong artifacts due to the sensor placement. This recording was sampled in 100~Hz from $t=0$~s to $t=86$~s. A demonstration recording is shown on the first row of Figure~\ref{fig:imu}.

For the sake of simplicity, we are only showing the results of GP and R-SS-NSGP ADMM. Also, the batch NSGP models would take an extensive amount of computational time on this data because of the large number of measurements. For this experiment we use the Mat\'{e}rn $\nu=1/2$ covariance function, and set the measurement noise variance as $0.1$. The hyperparameters of GP are learned by MLE and L-BFGS. We set the baseline level of $u^\ell(t)$ and $u^\sigma(t)$ as $2$ and $-1$, respectively, and the regularization parameters are $\lambda_\ell=\lambda_\sigma=2$ and $\rho_\ell=\rho_\sigma=50$.

The results are shown in Figure~\ref{fig:imu}. Compared with GP, the R-SS-NSGP gives a better fit because the estimates are smoother and have fewer impacts from artifacts. In particular, we find that the magnitude parameter of GP is learnt to have a high value (which is $\approx4.2$) in order to cover the dynamic range of the signal globally. This ill-learnt magnitude makes the estimates jitter on the flat parts and sensitive to motion artifacts. As an example, around $t\approx 39$~s, $t\approx55$~s, and $t\approx 63$~s the acceleration measurements exhibit strong noises which are well handled by R-SS-NSGP but not GP. In addition, R-SS-NSGP gives better quantification of the confidence interval which wraps the measurements reasonably. The parameter processes of R-SS-NSGP appear to be sparse due to the introduced regularizations. 

\section{Conclusion}
\label{sec:conclusion}
In this paper, we have introduced regularized non-stationary Gaussian processes (R-NSGPs), which involve $L^1$-regularization on the posterior estimates of the hierarchical non-stationary GPs. This enables promotion of sparsity or regularization of total variation in the GP and its parameter processes (e.g., length-scale process). The proposed models generalize two commonly used NSGP constructions which are the non-stationary covariance function-based NSGPs and state-space NSGPs. The R-NSGP regression is formulated as maximum a posteriori (MAP) estimation with additional regularization terms in the objective function. To solve the resulting optimization problem we use the alternating direction method of multipliers (ADMM) framework, and we also prove the convergence of the resulting method. The numerical experiments demonstrate that the proposed R-NSGPs are particularly useful for dealing with ill-behaved signals, such as piecewise continuous signals. We also successfully applied R-NSGPs to modeling of human motion. 

\begin{appendices}
\section{Solution analysis of SDEs~\eqref{equ:ss-gp-sde}}
\label{appendix:proof-sde-solution}
In this appendix, we show the existence of strong solution and pathwise uniqueness of SDEs~\eqref{equ:ss-gp-sde}. Before we show the lemma, we remark that there are different interpretations of the strong solution of an SDE~\citep[see, e.g., Chapter 10 of][]{Chung1990}. For example, \citet{Karatzas1991}, \citet{oksendal}, and~\citet{JeanGall2016} understand the strong solution as a continuous process which solves the SDE almost surely and is adapted to the filtration generated by a given Wiener process and the initial random variable, while \citet{AchimKlenke2014} and~\citet{Williams2000Vol2} define the solution as a suitable function. In this paper we adopt the interpretation by~\citet[][Definiton 2.1]{Karatzas1991} and~\citet[][Section 10.4]{Chung1990}. 

The SDE coefficients must be chosen suitably. In particular, Let the real parts of the eigenvalues of $\cu{A}^\ell$ and $\cu{A}^\sigma$ be negative. Also let functions $\cu{A}\colon\R^{D_\ell} \to \R^{D_f}$ and $\cu{B}\colon \R^{D_\ell} \times \R^{D_\sigma} \to \R^{D_f \times D_f^W}$ be Borel measurable and elementwise bounded in their arguments.

The proof in the sequel is based on the fact that each sub-SDE is conditionally a linear SDE which has explicit solution~\citep[see, e.g., Chapter 11 of][]{Kuo2006Book}. With the mutually independent conditions on the initial random variables and Wiener processes one can form a product measure space on which the solution is defined. 

\begin{lemma}
	\label{lemma:sde-solution}
	Suppose that $\cu{f}(t_0)$, $\cu{u}^\ell(t_0)$, and $\cu{u}^\sigma(t_0)$ are mutually independent and are also independent of the filtrations generated by $\cu{W}(t)$, $\cu{W}^\ell(t)$, and $\cu{W}^\sigma(t)$ for $t\in\T$. Then the SDE system in Equation~\eqref{equ:ss-gp-sde} has a strong pathwise-unique solution.
\end{lemma}

\begin{proof}
	Since the linear SDEs of $\cu{u}^\ell(t)$ and $\cu{u}^\sigma(t)$ have strong solutions~\citep{Kuo2006Book}, we let $\big(\Omega^\ell, \mathcal{F}^\ell,\lbrace \mathcal{F}^\ell_t \rbrace, \mathbb{P}^\ell\big)$ and $\big(\Omega^\sigma, \mathcal{F}^\sigma, \lbrace \mathcal{F}^\sigma_t \rbrace, \mathbb{P}^\sigma\big)$ be the filtered probability spaces of $\cu{u}^\ell(t)$ and $\cu{u}^\sigma(t)$, where $\lbrace \mathcal{F}^\ell_t\rbrace$ and $\lbrace \mathcal{F}^\sigma_t \rbrace$ are the generated filtrations of $(\cu{W}^\ell(t), \cu{u}^\ell(t_0))$ and $(\cu{W}^\sigma(t), \cu{u}^\sigma(t_0))$, respectively. Also, $\cu{u}^\ell(t)$ and $\cu{u}^\sigma(t)$ are adapted to $\lbrace \mathcal{F}^\ell_t\rbrace$ and $\lbrace\mathcal{F}^\sigma_t\rbrace$, respectively. 
	
	Due to the linearity and the conditions on coefficients $\cu{A}$ and $\cu{B}$, we have that for every trajectory of $\cu{u}^\ell(t)$ and $\cu{u}^\sigma(t)$, the strong existence holds for $\cu{f}(t)$ (because the SDE coefficients of $\cu{f}(t)$ satisfies the usual global Lipschitz condition). Namely, for every $\omega^\ell \in \Omega^\ell_t$ and $\omega^\sigma \in \Omega^\sigma_t$ the solution $\cu{f}(t)$ is defined on a filtered probability space $\big(\Omega^f, \mathcal{F}^f, \lbrace \mathcal{F}^f_t \rbrace, \mathbb{P}^f\big)$, and $\cu{f}(t)$ is adapted to the filtration $\lbrace \mathcal{F}^f_t \rbrace$ generated by $(\cu{W}(t), \cu{f}(t_0))$. Due to the independency of initial variables and Wiener processes, we can now form a product probability space $(\Omega, \mathcal{F}, \lbrace \mathcal{F}_t\rbrace, \mathbb{P})$, where $\Omega = \Omega^f \times\Omega^\ell \times \Omega^\sigma$, $\mathcal{F} = \mathcal{F}^f \otimes \mathcal{F}^\ell \otimes  \mathcal{F}^\sigma$, $\mathcal{F}_t = \mathcal{F}^f_t \otimes \mathcal{F}^\ell_t \otimes  \mathcal{F}^\sigma_t$, and $\mathbb{P}(E^f\times E^\ell \times E^\sigma) = \mathbb{P}^f(E^f)\,\mathbb{P}^\ell(E^\ell)\,\mathbb{P}^\sigma(E^\sigma)$ for every $E^f\in\mathcal{F}$, $E^\ell\in\mathcal{F}^\ell$, and $E^\sigma\in\mathcal{F}^\sigma$. Since probability spaces are sigma-finite, this product measure is uniquely defined~\citep[Theorem 14.14 or 14.5 of][]{AchimKlenke2014, reneProbabilityBook2005}. Therefore we have that $(\cu{f}(t), \cu{u}^\ell(t), \cu{u}^\sigma(t))$ is adapted to the filtration $\lbrace\mathcal{F}_t\rbrace$ which is the same with the filtration generated by $(\cu{W}^f(t), \cu{W}^\ell(t), \cu{W}^\sigma(t), \cu{f}(t_0), \cu{u}^\ell(t_0),\cu{u}^\sigma(t_0))$.
	
	We now show the pathwise uniqueness. Let $\big( \cu{f}_1(t),\cu{u}_1^\ell(t), \cu{u}_1^\sigma(t) \big)$ and $\big( \cu{f}_2(t),\cu{u}_2^\ell(t), \cu{u}_2^\sigma(t) \big)$ be any two solutions to the SDE system. Also let $M^\ell_t = \lbrace \omega \in \Omega^\ell\colon \abs{\cu{u}^\ell_1(t) - \cu{u}^\ell_2(t)} = 0 \rbrace$ and $M^\sigma_t = \lbrace \omega \in \Omega^\sigma\colon \abs{\cu{u}^\sigma_2(t) - \cu{u}^\sigma_1(t)} = 0 \rbrace$. Again, by the strong uniqueness of linear SDEs, we have $\mathbb{P}^\ell(M^\ell_t) = 1$ and $\mathbb{P}^\sigma(M^\sigma_t) = 1$ for all $t\in\T$. More importantly, for every $\omega^\ell \in \Omega^\ell$ and $\omega^\sigma\in\Omega^\sigma$ let the set $M^f_t(\omega^\ell, \omega^\sigma) = \lbrace \omega \in \Omega^f \colon \abs{\cu{f}_1(t) - \cu{f}_2(t)=0}\rbrace$, and we then have $\mathbb{P}^f(M^f_t(\omega^\ell, \omega^\sigma)) = 1$. Now define $M_t = M^f_t \times M^\ell_t \times M^\sigma_t \in\Omega$. By the product measure it is easy to see that $\mathbb{P}(M_t) = \mathbb{P}^f(M^f_t)\,\mathbb{P}^\ell(M^\ell_t)\,\mathbb{P}^\sigma(M^\sigma_t) = 1$ for every $t\in\T$. Due to the continuity of the solution, the two solutions $\big( \cu{f}_1(t),\cu{u}_1^\ell(t), \cu{u}_1^\sigma(t) \big)$ and $\big( \cu{f}_2(t),\cu{u}_2^\ell(t), \cu{u}_2^\sigma(t) \big)$ are indistinguishable~
	\citep[see, e.g., Lemma 21.5 of][]{AchimKlenke2014}. 
\end{proof}

\section{Proof of Lemma~\ref{lemma:ss-cov}}
\label{appendix:proof-ss-cov}
\begin{proof}
	Since the probability space $\left(\Omega, \mathcal{F}, \mathbb{P} \right)$ is fixed, the restricted measure $\mathbb{P}|_{\mathcal{F}^u}$ is uniquely defined. By the continuity of $\cu{A}(t)$, we have $\mathbb{P}|_{\mathcal{F}^u}$-almost surely that $\cu{x}(t)=\cu{\Lambda}(t, t_0)\, \cu{x}(t_0)$ is the fundamental solution to the linear ordinary differential equation 
	\begin{equation}
		\frac{\diff \cu{x}(t)}{\diff t} = \cu{A}(t)\,\cu{x}(t),\nonumber
	\end{equation}
	for any initial $\cu{x}(t_0)$, where $\cu{\Lambda}(t, t_0)$ is given by the Peano--Baker series~\citep{Brogan2011, Baake2011} 
	\begin{equation}
		\cu{\Lambda}(t, t_0) = \cu{I} + \int^t_{t_0} \cu{A}(s)\diff s + \int^t_{t_0} \cu{A}(s) \int^s_{t_0} \cu{A}(r)\diff r\diff s + \cdots.\nonumber
	\end{equation}
	The uniqueness and uniform convergence of $\cu{\Lambda}(t, t_0)$ is proved by~\citet{Baake2011} and~\citet{DaCunha2005}. This matrix $\cu{\Lambda}(t, t_0)$ satisfies the canonical properties of the transition matrix~\citep[see, Section 9.6 of][]{Brogan2011}. Hence by It\^{o}'s formula~\citep{oksendal, sarkkabook2019} one verifies that $\mathbb{P}|_{\mathcal{F}^u}$-almost surely 
	\begin{equation}
		\cu{f}(t) = \cu{\Lambda}(t, t_0)\,\cu{f}(t_0) + \int^t_{t_0}\cu{\Lambda}(t, s)\,\cu{B}(s)\diff \cu{W}(s),
		\label{equ:sde-f-solution}
	\end{equation}
	solves the sub-SDE of $\cu{f}(t)$. The conditional covariance function of $\cu{f}(t)$ is then given by
	\begin{equation}
		\begin{split}
			&\cov\left[\cu{f}(t), \cu{f}(t') \mid \mathcal{F}^u\right] \\
			&\quad= \expec\left[\cu{f}(t)\,\cu{f}^\trans(t')\mid \mathcal{F}^u\right] - \expec\left[\cu{f}(t)\mid \mathcal{F}^u\right] \left( \expec\left[\cu{f}(t')\mid \mathcal{F}^u\right]\right)^\trans \\
			&\quad= \cu{\Lambda}(t, t_0)\cov\left[\cu{f}(t_0)\mid \cu{u}^\ell(t_0), \cu{u}^\sigma(t_0)\right]\cu{\Lambda}^\trans(t', t_0) \\
			&\qquad+ \expec\Bigg[\int^t_{t_0}\cu{\Lambda}(t, s)\,\cu{B}(s)\diff \cu{W}(s) \\
			&\qquad\qquad\quad\times\bigg(\int^{t'}_{t_0}\cu{\Lambda}(t', s)\,\cu{B}(s)\diff \cu{W}(s) \bigg)^\trans \mid \mathcal{F}^u \Bigg] \\
			&\quad= \cu{\Lambda}(t, t_0)\cov\left[\cu{f}(t_0)\mid \cu{u}^\ell(t_0), \cu{u}^\sigma(t_0)\right]\cu{\Lambda}^\trans(t', t_0) \\
			&\qquad+ \int^{t\,\wedge\,t'}_{t_0}\cu{\Lambda}(t, s)\,\cu{B}(s)\,\cu{B}^\trans(s)\,\cu{\Lambda}^\trans(t', s)\diff s.
		\end{split}\nonumber
	\end{equation}
	Thus we arrive $C_f^S(t, t'; \cu{u}^\ell, \cu{u}^\sigma) \coloneqq \cu{H}_f\,\cov\left[\cu{f}(t), \cu{f}(t') \mid \mathcal{F}^u\right]\,\cu{H}^\trans_f$ by definition. 
\end{proof}

\section{Proof of Lemma~\ref{lemma:non-increasing}}
\label{appendix:proof-non-increasing}
Before the proof of Lemma~\ref{lemma:non-increasing}, we need the following auxiliary Lemma~\ref{lemma:useless}. 
\begin{lemma}
	\label{lemma:useless}
	The inequality
	\begin{equation}
	\mathrm{sgn}(v)\cdot (a - v) + \frac{\rho}{2}\norm{a - v}_2^2 \geq -\frac{T}{2\,\rho},
	\end{equation}
	holds for every $v\in\R^T$, $a\in\R^T$, and $\rho>0$. 
\end{lemma}
\begin{proof}
	Let $a_i$ and $v_i$ be the $i$-th elements of $a$ and $v$, respectively. Then
	\begin{equation}
	\begin{split}
	&\mathrm{sgn}(v)\cdot (a - v) + \frac{\rho}{2}\norm{a - v}_2^2  \\
	&\quad= \sum^T_{i=1} \left[ \mathrm{sgn}(v_i)\,(a_i - v_i) + \frac{\rho}{2}\,(a_i - v_i)^2\right] \\
	&\quad\geq \sum^T_{i=1} \left[ -\frac{(\mathrm{sgn}(v_i))^2}{2\,\rho}\right] \geq -\frac{T}{2\,\rho}.
	\end{split}\nonumber
	\end{equation}
\end{proof}

The proof of Lemma~\ref{lemma:non-increasing} is shown in the following.
\begin{proof}
	Let us start with the update from the first subproblem. We show the non-increasing property of Lagrangian function with respect to argument $f_{1:T}$ in the following equation
	\begin{equation}
		\begin{split}
			&\mathcal{L}\left( z_{1:T}^{(i)},v_{1:T}^{f,(i)},  v_{1:T}^{\ell,(i)}, v_{1:T}^{\sigma,(i)}, \eta_{1:T}^{f,(i)}, \eta_{1:T}^{\ell,(i)}, \eta_{1:T}^{\sigma,(i)}\right) \\
			&\quad- \mathcal{L}\left( z_{1:T}^{(i+1)}, v_{1:T}^{f,(i)}, v_{1:T}^{\ell,(i)}, v_{1:T}^{\sigma,(i)}, \eta_{1:T}^{f,(i)}, \eta_{1:T}^{\ell,(i)}, \eta_{1:T}^{\sigma,(i)}\right) \\
			&= \mathcal{L}^{\mathrm{NSGP}}\left(z_{1:T}^{(i)}\right) - \mathcal{L}^{\mathrm{NSGP}}\left(z_{1:T}^{(i+1)}\right) \\
			&\quad+ \left(\eta_{1:T}^{f,(i)}\right)^\trans \left( \Phi_f \left( f^{(i)}_{1:T} - f^{(i+1)}_{1:T} \right)  - v_{1:T}^{f,(i)} \right) \\
			&\quad+ \frac{\rho_f}{2}  \left( \norm*{\Phi_f \, f^{(i)}_{1:T} - v_{1:T}^{f,(i)} }^2_2 - \norm*{\Phi_f \, f^{(i+1)}_{1:T} - v_{1:T}^{f,(i)} }^2_2\right)  \\
			&\quad+ \left(\eta_{1:T}^{\ell,(i)}\right)^\trans \left( \Phi_\ell \left( u^{\ell, (i)}_{1:T} - u^{\ell, (i+1)}_{1:T} \right)  - v_{1:T}^{\ell,(i)} \right) \\
			&\quad+ \frac{\rho_\ell}{2}  \left( \norm*{\Phi_\ell \, u^{\ell, (i)}_{1:T} - v_{1:T}^{\ell,(i)} }^2_2 - \norm*{\Phi_\ell \, u^{\ell,(i+1)}_{1:T} - v_{1:T}^{\ell,(i)} }^2_2\right)  \\
			&\quad+ \left(\eta_{1:T}^{\sigma,(i)}\right)^\trans \left(\Phi_\sigma \left( u^{\sigma, (i)}_{1:T} - u^{\sigma, (i+1)}_{1:T} \right)  - v_{1:T}^{\sigma,(i)} \right) \\
			&\quad+ \frac{\rho_\sigma}{2}  \left( \norm*{\Phi_\sigma \, u^{\sigma, (i)}_{1:T} - v_{1:T}^{\sigma,(i)} }^2_2 - \norm*{\Phi_\sigma \, u^{\sigma,(i+1)}_{1:T} - v_{1:T}^{\sigma,(i)} }^2_2\right) \\
			&= \rho_f \bigg(\Phi_f^\trans\big( \Phi_f\, f^{(i+1)}_{1:T} - v^{f,(i)}_{1:T} + \frac{\eta^{f,(i)}_{1:T}}{\rho_f} \big)\bigg)^\trans \big( f^{(i)}_{1:T} - f^{(i+1)}_{1:T} \big) \\
			&+ \rho_\ell \bigg(\Phi_\ell^\trans\big( \Phi_\ell\, u^{\ell, (i+1)}_{1:T} - v^{\ell,(i)}_{1:T} + \frac{\eta^{\ell,(i)}_{1:T}}{\rho_\ell} \big)\bigg)^\trans \big( u^{\ell,(i)}_{1:T} - u^{\ell,(i+1)}_{1:T} \big) \\
			&+\rho_\sigma \bigg(\Phi_\sigma^\trans\big( \Phi_\sigma\, u^{\sigma, (i+1)}_{1:T} - v^{\sigma,(i)}_{1:T} + \frac{\eta^{\sigma,(i)}_{1:T}}{\rho_\sigma} \big)\bigg)^\trans \big( u^{\sigma,(i)}_{1:T} - u^{\sigma,(i+1)}_{1:T} \big) \\
			&\quad+ \frac{\rho_f}{2} \norm*{\Phi_f\left(  f^{(i+1)}_{1:T} - f^{(i)}_{1:T}\right) }^2_2 \\
			&\quad+ \frac{\rho_\ell}{2} \norm*{\Phi_\ell\left(  u^{\ell, (i+1)}_{1:T} - u^{\ell, (i)}_{1:T}\right) }^2_2 \\
			&\quad+ \frac{\rho_\sigma}{2} \norm*{\Phi_\sigma\left(  u^{\sigma, (i+1)}_{1:T} - u^{\sigma, (i)}_{1:T}\right) }^2_2, \nonumber
		\end{split}
	\end{equation}
	where we used the cosine rule which states that $\norm{\cu{a} + \cu{c}}^2 - \norm{\cu{b} + \cu{c}}^2 - \norm{\cu{a} - \cu{b}}^2 = 2\,(\cu{b}+\cu{c})^\trans(\cu{a}-\cu{b})$ for every real vectors $\cu{a},\cu{b}$, and $\cu{c}$ to arrive at the last equality in the above equation. Now by substituting the first-order optimal condition of subproblem~\eqref{equ:batch-sub-fu} (see, Equation~\eqref{equ:optimal-cond} in the Appendix~\ref{appendix:gradients}), we continue by
	\begin{equation}
		\begin{split}
			&\mathcal{L}\left( z_{1:T}^{(i)},v_{1:T}^{f,(i)},  v_{1:T}^{\ell,(i)}, v_{1:T}^{\sigma,(i)}, \eta_{1:T}^{f,(i)}, \eta_{1:T}^{\ell,(i)}, \eta_{1:T}^{\sigma,(i)}\right) \\
			&\quad- \mathcal{L}\left( z_{1:T}^{(i+1)}, v_{1:T}^{f,(i)}, v_{1:T}^{\ell,(i)}, v_{1:T}^{\sigma,(i)}, \eta_{1:T}^{f,(i)}, \eta_{1:T}^{\ell,(i)}, \eta_{1:T}^{\sigma,(i)}\right) \\
			&= \mathcal{L}^{\mathrm{NSGP}}\left(z_{1:T}^{(i)}\right) - \mathcal{L}^{\mathrm{NSGP}}\left(z_{1:T}^{(i+1)}\right) \\
			&\quad- \nabla^\trans_{z_{1:T}} \mathcal{L}^{\mathrm{NSGP}}\left(z^{(i+1)}_{1:T}\right)\left( z_{1:T}^{(i)} - z_{1:T}^{(i+1)} \right)\\
			&\quad+ \frac{\rho_f}{2} \norm*{\Phi_f\left(  f^{(i+1)}_{1:T} - f^{(i)}_{1:T}\right) }^2_2 \\
			&\quad+ \frac{\rho_\ell}{2} \norm*{\Phi_\ell\left(  u^{\ell, (i+1)}_{1:T} - u^{\ell, (i)}_{1:T}\right) }^2_2 \\
			&\quad+ \frac{\rho_\sigma}{2} \norm*{\Phi_\sigma\left(  u^{\sigma, (i+1)}_{1:T} - u^{\sigma, (i)}_{1:T}\right) }^2_2 \nonumber
		\end{split}
	\end{equation}
	Hence, by applying Lemma~\ref{lemma:lip-cont} we arrive at
	\begin{equation}
		\begin{split}
			&\mathcal{L}\left( z_{1:T}^{(i)},v_{1:T}^{f,(i)},  v_{1:T}^{\ell,(i)}, v_{1:T}^{\sigma,(i)}, \eta_{1:T}^{f,(i)}, \eta_{1:T}^{\ell,(i)}, \eta_{1:T}^{\sigma,(i)}\right) \\
			&\quad- \mathcal{L}\left( z_{1:T}^{(i+1)}, v_{1:T}^{f,(i)}, v_{1:T}^{\ell,(i)}, v_{1:T}^{\sigma,(i)}, \eta_{1:T}^{f,(i)}, \eta_{1:T}^{\ell,(i)}, \eta_{1:T}^{\sigma,(i)}\right) \\
			&\geq -\frac{L_z}{2}\norm*{z^{(i)}_{1:T} - z^{(i+1)}_{1:T}}^2_2 \\
			&\quad+\frac{\rho_f}{2}\,\lambda_{\mathrm{min}}^2(\Phi_f)\norm*{f^{(i+1)}_{1:T} - f^{(i)}_{1:T}}^2_2 \\
			&\quad+\frac{\rho_\ell}{2}\,\lambda_{\mathrm{min}}^2(\Phi_\ell)\norm*{u^{\ell, (i+1)}_{1:T} - u^{\ell, (i)}_{1:T}}^2_2 \\
			&\quad+\frac{\rho_\sigma}{2}\,\lambda_{\mathrm{min}}^2(\Phi_\sigma)\norm*{u^{\sigma, (i+1)}_{1:T} - u^{\sigma, (i)}_{1:T}}^2_2 \\
			&=\left( \frac{\rho_f}{2}\,\lambda_{\mathrm{min}}^2(\Phi_f) - \frac{L_z}{2}\right)\norm*{f^{(i+1)}_{1:T} - f^{(i)}_{1:T}}^2_2 \\
			&=\left( \frac{\rho_\ell}{2}\,\lambda_{\mathrm{min}}^2(\Phi_\ell) - \frac{L_z}{2}\right)\norm*{u^{\ell, (i+1)}_{1:T} - u^{\ell, (i)}_{1:T}}^2_2 \\
			&\quad+ \left( \frac{\rho_\sigma}{2}\,\lambda_{\mathrm{min}}^2(\Phi_\sigma) - \frac{L_z}{2}\right)\norm*{u^{\sigma, (i+1)}_{1:T} - u^{\sigma, (i)}_{1:T}}^2_2 \\
			&\geq 0.\nonumber
		\end{split}
	\end{equation}
	Since subproblems~\eqref{equ:batch-sub-v} are convex we have
	\begin{equation}
		\begin{split}
			&\mathcal{L}\left( z_{1:T}^{(i+1)}, v_{1:T}^{f,(i+1)}, v_{1:T}^{\ell,(i+1)}, v_{1:T}^{\sigma,(i+1)}, \eta_{1:T}^{f,(i)},\eta_{1:T}^{\ell,(i)}, \eta_{1:T}^{\sigma,(i)}\right)\\ 
			&\quad\leq \mathcal{L}\left( z_{1:T}^{(i+1)}, v_{1:T}^{f,(i)}, v_{1:T}^{\ell,(i)}, v_{1:T}^{\sigma,(i)},  \eta_{1:T}^{f,(i)}, \eta_{1:T}^{\ell,(i)}, \eta_{1:T}^{\sigma,(i)}\right). \nonumber
		\end{split}
	\end{equation}
	For subproblems~\eqref{equ:batch-sub-eta}, we have
	\begin{equation}
		\begin{split}
			&\mathcal{L}\left( z_{1:T}^{(i+1)},  v_{1:T}^{f,(i+1)}, v_{1:T}^{\ell,(i+1)}, v_{1:T}^{\sigma,(i+1)}, \eta_{1:T}^{f,(i)}, \eta_{1:T}^{\ell,(i)}, \eta_{1:T}^{\sigma,(i)}\right) \\
			&\quad-\mathcal{L}\big( z_{1:T}^{(i+1)}, v_{1:T}^{f,(i+1)}, v_{1:T}^{\ell,(i+1)}, v_{1:T}^{\sigma,(i+1)}, \eta_{1:T}^{f,(i+1)}, \\
			&\qquad\quad\eta_{1:T}^{\ell,(i+1)}, \eta_{1:T}^{\sigma,(i+1)}\big) \\
			&= \frac{1}{\rho_f}\norm*{\eta^{f, (i)}_{1:T} - \eta^{f, (i+1)}_{1:T}}^2_2 \\
			&\quad+ \frac{1}{\rho_\ell}\norm*{\eta^{\ell, (i)} - \eta^{\ell, (i+1)}}^2_2 + \frac{1}{\rho_\sigma}\norm*{\eta^{\sigma, (i)}_{1:T} - \eta^{\sigma, (i+1)}_{1:T}}^2_2 \\
			&\geq 0.\nonumber
		\end{split}
	\end{equation}
	By using the above results, the difference of Lagrangian function on one step reads
	\begin{equation}
		\begin{split}
			&\mathcal{L}\left( z_{1:T}^{(i)}, v_{1:T}^{f,(i)} v_{1:T}^{\ell,(i)}, v_{1:T}^{\sigma,(i)}, \eta_{1:T}^{f,(i)}, \eta_{1:T}^{\ell,(i)}, \eta_{1:T}^{\sigma,(i)}\right) \\
			&\quad- \mathcal{L}\Big( z_{1:T}^{(i+1)}, v_{1:T}^{f,(i+1)}, v_{1:T}^{\ell,(i+1)}, v_{1:T}^{\sigma,(i+1)}, \eta_{1:T}^{f,(i+1)}, \\
			&\qquad\quad\eta_{1:T}^{\ell,(i+1)}, \eta_{1:T}^{\sigma,(i+1)}\Big) \\
			&= \mathcal{L}\left( z_{1:T}^{(i)}, v_{1:T}^{f,(i)} v_{1:T}^{\ell,(i)}, v_{1:T}^{\sigma,(i)}, \eta_{1:T}^{f,(i)},\eta_{1:T}^{\ell,(i)}, \eta_{1:T}^{\sigma,(i)}\right) \\
			&\quad- \mathcal{L}\left( z_{1:T}^{(i+1)}, v_{1:T}^{f,(i)}, v_{1:T}^{\ell,(i)}, v_{1:T}^{\sigma,(i)}, \eta_{1:T}^{f,(i)}, \eta_{1:T}^{\ell,(i)}, \eta_{1:T}^{\sigma,(i)}\right) \\
			&\quad+ \mathcal{L}\left( z_{1:T}^{(i+1)}, v_{1:T}^{f,(i)},  v_{1:T}^{\ell,(i)}, v_{1:T}^{\sigma,(i)}, \eta_{1:T}^{\ell,(i)}, \eta_{1:T}^{f,(i)}, \eta_{1:T}^{\sigma,(i)}\right) \\
			&\quad- \mathcal{L}\left( z_{1:T}^{(i+1)}, v_{1:T}^{f,(i+1)} v_{1:T}^{\ell,(i+1)}, v_{1:T}^{\sigma,(i+1)}, \eta_{1:T}^{f,(i)}, \eta_{1:T}^{\ell,(i)}, \eta_{1:T}^{\sigma,(i)}\right) \\
			&\quad+ \mathcal{L}\left( z_{1:T}^{(i+1)}, v_{1:T}^{f,(i+1)}, v_{1:T}^{\ell,(i+1)}, v_{1:T}^{\sigma,(i+1)}, \eta_{1:T}^{f,(i)}, \eta_{1:T}^{\ell,(i)}, \eta_{1:T}^{\sigma,(i)}\right) \\
			&\quad- \mathcal{L}\Big( z_{1:T}^{(i+1)}, v_{1:T}^{f,(i+1)}, v_{1:T}^{\ell,(i+1)}, v_{1:T}^{\sigma,(i+1)}, \eta_{1:T}^{f,(i+1)}, \\
			&\qquad\quad\eta_{1:T}^{\ell,(i+1)}, \eta_{1:T}^{\sigma,(i+1)}\Big) \\
			&\geq 0. \nonumber
		\end{split}
	\end{equation}
	This shows that the sequence  $\mathcal{L}\big( f_{1:T}^{(i)}, u^{(i)}_{1:T}, v_{1:T}^{\ell,(i)}, v_{1:T}^{\sigma,(i)}, \eta_{1:T}^{\ell,(i)}, \eta_{1:T}^{\sigma,(i)}\big)$ is non-increasing. We now proceed to prove that the Lagrangian function has a lower bound uniformly. 
	
	By the optimality condition of subproblem~\eqref{equ:batch-sub-v} we have
	\begin{equation}
		\begin{split}
			& \lambda_f \,\partial\norm*{ v_{1:T}^{f,(i)}}_1  - 
			\rho_f \left(\Phi_f \, f^{(i)}_{1:T} - v_{1:T}^{f,(i)} +
			{\eta_{1:T}^{f,(i-1)}}\,/\,\rho_f\right) = 0, \\
			& \lambda_\ell \,\partial\norm*{ v_{1:T}^{\ell,(i)}}_1  - 
			\rho_\ell \left(\Phi_\ell \, u^{\ell,(i)}_{1:T} - v_{1:T}^{\ell,(i)} +
			{\eta_{1:T}^{\ell,(i-1)}}\,/\,\rho_\ell\right) = 0, \\
			&\lambda_\sigma \, \partial \norm*{ v_{1:T}^{\sigma,(i)} }_1  - 
			\rho_\sigma \left(\Phi_\sigma \, u^{\sigma,(i)}_{1:T} - v_{1:T}^{\sigma,(i)} +
			{\eta_{1:T}^{\sigma,(i-1)}}\,/\,\rho_\sigma\right) = 0, \\
			&\lambda_f \, \partial \norm*{ v_{1:T}^{f,(i)} }_1  = \eta_{1:T}^{f,(i)}, \\
			&\lambda_\ell \, \partial \norm*{ v_{1:T}^{\ell,(i)} }_1  = \eta_{1:T}^{\ell,(i)}, \\
			&\lambda_\sigma \, \partial \norm*{ v_{1:T}^{\sigma,(i)} }_1  = \eta_{1:T}^{\sigma,(i)},
		\end{split}\nonumber
	\end{equation}
	where derivative $\partial \norm{\cdot}_1$ is interpreted as $\mathrm{sgn}(\cdot)$. By using Lemma~\ref{lemma:useless} the Lagrangian function then becomes
	\begin{equation}
		\begin{split}
			&\mathcal{L}\Big( z_{1:T}^{(i)}, v_{1:T}^{f,(i)}, v_{1:T}^{\ell,(i)}, v_{1:T}^{\sigma,(i)}, \eta_{1:T}^{f,(i)}, \eta_{1:T}^{\ell,(i)}, \eta_{1:T}^{\sigma,(i)}\Big)
			\\
			&=\norm*{f_{1:T}^{(i)} -y_{1:T}}^2_{\cu{R}} + \norm*{f_{1:T}^{(i)}}^2_{\cu{C}_f^{(i)}} + \log \abs*{2\,\pi\,\cu{C}_f^{(i)}} \\
			&\quad+ \norm*{u^{\ell,(i)}_{1:T}}^2_{\cu{C}_\ell} + \norm*{u^{\sigma,(i)}_{1:T}}^2_{\cu{C}_\sigma} + \log \abs{2\,\pi\,\cu{C}_\ell} + \log \abs{2\,\pi\,\cu{C}_\sigma}
			\\ 
			&\quad+ \lambda_f \norm*{ v_{1:T}^{f,(i)} }_1  
			+ \left(\lambda_f \, \partial \norm*{ v_{1:T}^{f,(i)} }_1\right)^\trans \left( \Phi_f \, f^{(i)}_{1:T} - v_{1:T}^{f,(i)}  \right) \\
			&\quad+ \lambda_\ell \norm*{ v_{1:T}^{\ell,(i)} }_1  
			+ \left(\lambda_\ell \, \partial \norm*{ v_{1:T}^{\ell,(i)} }_1\right)^\trans \left( \Phi_\ell \, u^{\ell,(i)}_{1:T} - v_{1:T}^{\ell,(i)}  \right) \\
			&\quad
			+ \lambda_\sigma \norm*{ v_{1:T}^{\sigma,(i)} }_1 
			+ \left(\lambda_\sigma \, \partial \norm*{ v_{1:T}^{\sigma,(i)} }_1\right)^\trans \left(  \Phi_\sigma \, u^{\sigma,(i)}_{1:T} - v_{1:T}^{\sigma,(i)}  \right) \\
			&\quad+ \frac{\rho_f}{2}  \norm*{\Phi_f \, f^{(i)}_{1:T} - v_{1:T}^{f,(i)}}^2_2\\
			&\quad  + \frac{\rho_\ell}{2}  \norm*{\Phi_\ell \, u^{\ell,(i)}_{1:T} - v_{1:T}^{\ell,(i)}}^2_2
			+ \frac{\rho_\sigma}{2} \norm*{\Phi_\sigma \, u^{\sigma,(i)}_{1:T} -v_{1:T}^{\sigma,(i)}}^2_2 \\
			&\geq\mathcal{L}^{\mathrm{NSGP}}\left( z_{1:T}^{(i)}\right) - \frac{T\,(\rho_f + \rho_\ell + \rho_\sigma)}{2\,\rho_f\,\rho_\ell\,\rho_\sigma}\\
			&\geq \log \abs*{2\,\pi\,\cu{C}^{(i)}_f} + \log\abs{2\,\pi\,\cu{C}_\ell} + \log\abs{2\,\pi\,\cu{C}_\sigma} \\
			&\quad- \frac{T\,(\rho_f+\rho_\ell + \rho_\sigma)}{2\,\rho_f\,\rho_\ell\,\rho_\sigma}.
		\end{split}\nonumber
	\end{equation}
	By Assumption~\ref{assu:eig-of-C}, the minimum eigenvalue of $\cu{C}_f^{(i)}$ admits a lower bound independent of arguments $u^{\ell, (i)}_{1:T}$, $u^{\sigma, (i)}_{1:T}$, and $f^{(i)}_{1:T}$ for every $i=0, 1,\ldots$. Hence, 
	\begin{equation}
		\log \abs*{2\,\pi\,\cu{C}^{(i)}_f} \geq
		\begin{cases}
			\log c, & 1 \leq c, \\
			 T \, \log c, & 0 < c < 1.
		\end{cases}\nonumber
	\end{equation}
\end{proof}

\section{The \matern $\nu=1/2$ construction of SS-NSGP}
\label{appendix:sde-matern12}
Let the SDE coefficients be
\begin{equation}
	\cu{A}\big(\cu{u}^\ell(t)\big) = -\frac{1}{g(u^\ell(t))}, \nonumber
\end{equation}
and
\begin{equation}
	\cu{B}\big(\cu{u}^\ell(t), \cu{u}^\sigma(t)\big) = \frac{\sqrt{2}\,g(u^\sigma(t))}{\sqrt{g(u^\ell(t))}}.\nonumber
\end{equation}
Then the state-space NSGP construction for the Ornstein--Uhlenbeck/\matern $\nu=1/2$ covariance function reads
\begin{equation}
	\diff \cu{z}(t) = \begin{bmatrix}
		-\frac{f(t)}{g(u^\ell(t))} \\
		-\frac{u^\ell(t)}{\ell} \\
		-\frac{u^\sigma(t)}{\ell}
	\end{bmatrix}\diff t + \begin{bmatrix}
		\frac{\sqrt{2}\,g(u^\sigma(t))}{\sqrt{g(u^\ell(t))}} & 0 & 0\\
		0 & \frac{\sqrt{2}\,\sigma}{\sqrt{\ell}} & 0\\
		0 & 0& \frac{\sqrt{2}\,\sigma}{\sqrt{\ell}}
	\end{bmatrix}\diff \cu{W}(t),\nonumber
\end{equation}
where $\cu{z}(t) = \begin{bmatrix} f(t) & u^\ell(t) & u^\sigma(t) \end{bmatrix}^\trans$, $\cu{f}(t)=f(t)$, $\cu{u}^\ell(t)=u^\ell(t)$, $\cu{u}^\sigma(t)=u^\sigma(t)$, and $\ell$ and $\sigma$ are shared hyperparameters of $u^\ell(t)$ and $u^\sigma(t)$. In this formulation, $u^\ell(t)$ and $u^\sigma(t)$ are the conventional (stationary) \matern $\nu=1/2$ GPs, and $f(t)$ is a non-stationary \matern $\nu=1/2$ GP depending on $u^\ell(t)$ and $u^\sigma(t)$. Also, by Lemma~\ref{lemma:ss-cov} we have $\cu{\Lambda}(t, t_0) = \exp\Big( -\int^t_{t_0} 1\,/ \,g(u^\ell(s))\diff s\Big) $. 

\section{First-order optimality}
\label{appendix:gradients}
The first-order optimality condition of subproblem~\eqref{equ:batch-sub-fu} at iteration $i+1$ implies that the gradient of Equation~\eqref{equ:batch-sub-fu} is zero at $\left\lbrace f_{1:T}^{(i+1)}, u_{1:T}^{\ell,(i+1)},u_{1:T}^{\sigma,(i+1)}\right\rbrace$. This gives
\begin{equation}
	\begin{split}
		&\nabla_{u^\ell_{1:T}, u^\sigma_{1:T}}\mathcal{L}^{\mathrm{NSGP}}\left(z^{(i+1)}_{1:T}\right) \\
		&+ \begin{bmatrix}
			\Phi_\ell^\trans\, \eta^{\ell,(i)}_{1:T} + \rho_\ell\,\Phi_\ell^\trans\left(\Phi_\ell\, u^{\ell,(i+1)}_{1:T} - v^{\ell,(i)}_{1:T} \right)  \\
			\Phi_\sigma^\trans\, \eta^{\sigma,(i)}_{1:T} + \rho_\sigma\,\Phi_\sigma^\trans\left(\Phi_\sigma\, u^{\sigma,(i+1)}_{1:T} - v^{\sigma,(i)}_{1:T} \right)
		\end{bmatrix} = \cu{0},
	\end{split}
	\label{equ:optimal-cond}
\end{equation}
and $\nabla_{f_{1:T}}\mathcal{L}^{\mathrm{NSGP}}\left(f^{(i+1)}_{1:T}, u^{(i+1)}_{1:T}\right) = \cu{0}$. 

The gradients of the objective function $\mathcal{L}^{\mathrm{NSGP}}$ with respect to each argument are
\begin{equation}
	\begin{split}
		\nabla_{f_{1:T}}\mathcal{L}^{\mathrm{NSGP}} &= 2\left(\cu{C}_f^{-1}\,f_{1:T} + \cu{R}^{-1}(f_{1:T} - y_{1:T})\right),\\
		\nabla_{u^\ell_{1:T}}\mathcal{L}^{\mathrm{NSGP}} &=  2\,\cu{C}_\ell^{-1}\,u^\ell_{1:T} + g^\ell_{1:T}, \\
		\nabla_{u^\sigma_{1:T}}\mathcal{L}^{\mathrm{NSGP}} &=  2\, \cu{C}_\sigma^{-1}\,u^\sigma_{1:T} + g^\sigma_{1:T},
	\end{split}
	\label{equ:nsgp-grads}
\end{equation}
where the $m$-th elements of $g^\ell_{1:T}\in\R^{T}$ and $g^\sigma_{1:T}\in\R^{T}$ are
\begin{equation}
	\begin{split}
		[g^\ell_{1:T}]_m &= \trace\left[\left(\cu{C}_f^{-1} - \cu{\tau}_\ell\,\cu{\tau}_\ell^\trans\right)\,\tash{\cu{C}_f}{u^\ell_m}\right], \\
		[g^\sigma_{1:T}]_m &= \trace\left[\left(\cu{C}_f^{-1} - \cu{\tau}_\sigma\,\cu{\tau}_\sigma^\trans\right)\,\tash{\cu{C}_f}{u^\sigma_m}\right],\nonumber
	\end{split}
\end{equation}
respectively. Above, $\cu{\tau}_\ell = \cu{C}_f^{-1}\,u^\ell_{1:T}$ and $\cu{\tau}_\sigma = \cu{C}_f^{-1}\,u^\sigma_{1:T}$. 

\end{appendices}

\bibliographystyle{spbasic}  
\bibliography{refs}

\end{document}